\documentclass[12pt]{article}
\usepackage{setspace}

\usepackage[hang,flushmargin]{footmisc}
\usepackage[vmargin=1in,hmargin=1.1in]{geometry}				
\usepackage{amssymb,amsmath,amsfonts,amsthm,mathtools,bm,bbm,mathrsfs} 	
\usepackage{xparse}
\usepackage[dvipsnames]{xcolor} 								
\usepackage[hidelinks]{hyperref} 								
\usepackage{graphicx} 											
\usepackage{wrapfig} 											
\usepackage{caption} 											
\usepackage{verbatim}
\usepackage[misc]{ifsym} 										
\usepackage[ruled, linesnumbered]{algorithm2e} 					
\usepackage{layouts} 
\usepackage[uniquename=false, 
uniquelist = false, 
url = false, 
doi = false, 
isbn = false, 
natbib = true, 
backend = bibtex, 
style = authoryear, 
maxbibnames = 5]{biblatex} 							
\usepackage{mathabx}
\usepackage{cancel}
\usepackage{float}
\usepackage{lineno} 

\newtheorem{theorem}{Theorem}
\newtheorem*{theorem*}{Theorem}

\newtheorem{lemma}{Lemma}
\newtheorem{corollary}{Corollary}

\theoremstyle{definition}
\newtheorem{definition}{Definition}
\theoremstyle{remark}
\newtheorem{remark}{Remark}
\newtheorem{example}{Example}
\makeatletter
\newcommand\newtarget[2]{\Hy@raisedlink{\hypertarget{#1}{}}#2}
\makeatother

\newcommand{\E}{\mathbb E}								
\newcommand{\V}{\mathrm{Var}}							
\renewcommand{\P}{\mathbb{P}}							
\newcommand{\Q}{\mathbb{Q}}								
\newcommand{\R}{\mathbb{R}}								
\newcommand{\norm}[1]{\left\lVert{#1}\right\rVert}		
\newcommand{\independent}{{\ \perp \! \! \! \perp\ }}	
\newcommand{\nindependent}{\ \cancel{\perp \! \! \! \perp}\ } 
\newcommand{\iidsim}{\stackrel{\mathrm{i.i.d.}}{\sim}} 	
\newcommand{\indsim}{\stackrel{\mathrm{ind}}{\sim}}		
\newcommand{\convp}{\overset p \rightarrow}             
\newcommand{\convd}{\overset d \rightarrow}             
\newcommand{\cS}{\mathcal{S}}
\newcommand{\cF}{\mathcal F}							
\newcommand{\prx}{\bm X}								
\newcommand{\srx}{X}									

\newcommand{\peps}{\varepsilon}						
\newcommand{\seps}{\varepsilon}							
\newcommand{\pxi}{\xi}						
\newcommand{\sxi}{\xi}							
\newcommand{\law}{\mathcal L}							
\newcommand{\nulllaws}{\mathscr L^0}					
\newcommand{\lawhat}{\widehat{\mathcal L}}				
\newcommand{\minus}{\textnormal{-}} 						    
\newcommand{\mj}{{\textnormal{-}j}} 					

\NewDocumentCommand{\exa}{g}{%
  X\IfValueT{#1}{_{#1}}%
}
\newcommand{\cxa}{\bm X}  

\NewDocumentCommand{\rz}{g}{%
  \bm Z\IfValueT{#1}{_{#1}}%
}
\newcommand{\mz}{\bm Z}   

\NewDocumentCommand{\rx}{g g}{%
  \bm X%
  \IfValueT{#1}{%
    _{#1\IfValueT{#2}{,#2}}%
  }%
}

\NewDocumentCommand{\rxs}{g g}{%
  \bm x%
  \IfValueT{#1}{%
    _{#1\IfValueT{#2}{,#2}}%
  }%
}

\NewDocumentCommand{\rxk}{g g}{%
  \widetilde{\bm X}%
  \IfValueT{#1}{%
    _{#1,\IfValueTF{#2}{#2}{\bullet}}%
  }%
}

\NewDocumentCommand{\rxka}{g g}{%
  \widetilde{\bm X}%
  \IfValueT{#1}{%
    _{#1,\IfValueTF{#2}{#2}{\bullet}}%
  }%
}

\NewDocumentCommand{\ex}{m g}{%
X_{#1\IfValueT{#2}{,#2}}%
}

\NewDocumentCommand{\exk}{m g}{%
\widetilde X_{#1\IfValueT{#2}{,#2}}%
}

\NewDocumentCommand{\exka}{m g}{%
\widetilde X_{#1\IfValueT{#2}{,#2}}%
}

\NewDocumentCommand{\exs}{m g}{%
x_{#1\IfValueT{#2}{,#2}}%
}

\NewDocumentCommand{\cx}{g g}{%
  \bm X%
  \IfValueT{#1}{%
    _{#1\IfValueT{#2}{,#2}}%
  }%
}

\NewDocumentCommand{\cxk}{g g}{%
  \widetilde{\bm X}%
  \IfValueT{#1}{%
    _{#1\IfValueT{#2}{,#2}}%
  }%
}

\NewDocumentCommand{\mx}{g g}{%
  \bm X_{%
    \IfValueTF{#1}{#1}{\bullet},%
    \IfValueTF{#2}{#2}{\bullet}%
  }%
}

\NewDocumentCommand{\mxk}{g g}{%
  \widetilde{\bm X}_{%
    \IfValueTF{#1}{#1}{\bullet},%
    \IfValueTF{#2}{#2}{\bullet}%
  }%
}

\NewDocumentCommand{\ey}{g}{%
  Y%
  \IfValueT{#1}{_{#1}}%
}

\NewDocumentCommand{\cy}{g}{%
  \bm Y%
  \IfValueT{#1}{_{#1}}%
}

\let\oldnl\nl
\newcommand{\nonl}{\renewcommand{\nl}{\let\nl\oldnl}} 

\begin{document}

\pagenumbering{arabic}

	\title{Doubly robust and computationally efficient \\ high-dimensional variable selection}
	\author{Abhinav Chakraborty\footnote{Equal contribution.},\addtocounter{footnote}{-1} Jeffrey Zhang\footnotemark, and Eugene Katsevich}
	
	\maketitle
	\thispagestyle{empty}
	\begin{abstract}
	
Variable selection can be performed by testing conditional independence (CI) between each predictor and the response, given the other predictors. The projected covariance measure (PCM) test is a doubly robust and powerful CI test. However, directly deploying PCM for variable selection brings computational challenges: testing a single variable involves a few machine learning fits, so testing $p$ variables requires $O(p)$ fits. Inspired by model-X ideas, we observe that an estimate of the joint predictor distribution and a single response-on-all-predictors fit can be used to reconstruct all PCM fits. This yields tower PCM (tPCM), a computationally efficient extension of PCM to variable selection. When the joint predictor distribution is sufficiently tractable, as in applications like genome-wide association studies, tPCM offers a substantial speedup over PCM---up to 130x in our simulations---while matching its power. tPCM also improves on model-X methods like knockoffs and holdout randomization test (HRT) by returning per-variable $p$-values and improving speed, respectively. We prove that tPCM is doubly robust and asymptotically equivalent to both PCM and HRT. We thus extend the bridge between model-X and doubly robust approaches, demonstrating their independent arrival at equivalent methods and showing that this intersection is a fruitful source of new methodologies like tPCM.
	
	\end{abstract}

	\section{Introduction}

	\subsection{The variable selection problem}
	
	Variable selection, which involves identifying a subset of predictors that are relevant to a response variable of interest, is a common statistical challenge. For example, in genome-wide association studies (GWAS), researchers aim to identify genetic variants that influence disease susceptibility, while in finance, analysts seek indicators that predict stock prices. In these problems and many others, only a small fraction of the available predictors are expected to have an impact on the response. 
	
	Let us denote the predictor variables $\rx = (\ex{1}, \dots, \ex{p}) \in \R^p$ and the response variable $\ey \in \R$. We have $2n$ i.i.d. observations $(\rx{i}{\bullet}, \ey{i}) \iidsim \law(\rx, \ey)$, denoted $(\mx, \cy) \equiv \{(\rx{i}{\bullet}, \ey{i})\}_{i = 1, \dots, 2n}$ (we use $2n$ rather than $n$ for convenience). We consider the variable selection problem of testing the conditional independence (CI) of the response $\ey$ and the predictor $\ex{j}$ given the other predictors $\rx{\mj}$, for each $j$ \citep{CetL16}:
	\begin{equation}
		H_{0j}: \ey \independent \ex{j} \mid \rx{\mj} \quad \textnormal{versus} \quad H_{1j}: \ey \nindependent \ex{j} \mid \rx{\mj}.
		\label{eq:conditional-independence}
	\end{equation}
	Depending on the context, it may be desired to control the family-wise error rate (FWER) or the false discovery rate (FDR) among the selected variables. This problem poses a range of statistical and computational challenges (Section~\ref{sec:desirable-properties}), which no existing method has addressed jointly (Section~\ref{sec:existing-approaches}). We introduce a new method that cross-pollinates ideas from two strands of existing work to meet these challenges (Section~\ref{sec:our-contributions}).	

	\subsection{Desirable properties of variable selection methods} \label{sec:desirable-properties}

	We lay out four desirable criteria for variable selection methods. 

	\vspace{-0.15in}

	\paragraph{Type-I error control with nuisance estimation.} A central requirement for general-purpose variable selection is asymptotic Type-I error control when nuisance components such as $\law(\rx)$ and $\law(\ey \mid \rx)$ (or functionals thereof) must be estimated from the observed data, rather than assumed known or estimated from a large auxiliary sample. The latter scenarios do occur \citep{Ham2022a,Aufiero2022a,Zhang2022} but not commonly; see the discussion of model-X methods in Section~\ref{sec:existing-approaches}.
	
	\vspace{-0.15in}

		\begin{wrapfigure}[13]{r}{0.3\textwidth}
		\vspace{-0.20in}
  	\centering
	  \includegraphics[width=0.3\textwidth]{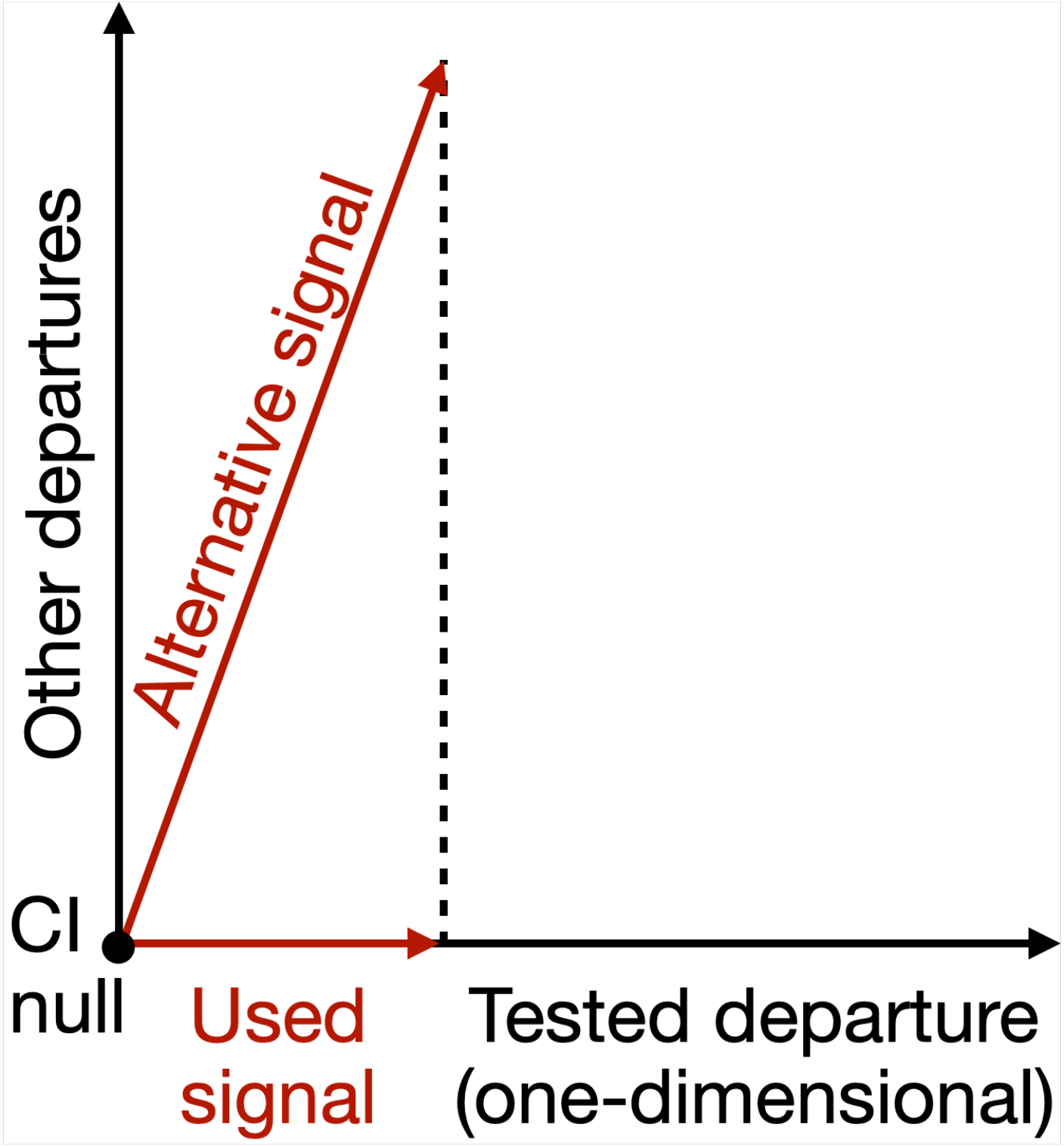}
  	\caption{Testing against 1D alternatives.}
  	\label{fig:1d-alternatives}
	\end{wrapfigure}

	\paragraph{Power against multi-dimensional alternatives.} Some CI tests detect departures from the null along a single prespecified ``direction'' in the space of alternatives. The power of such a test against any local alternative signal is a function of its projection onto this direction (Figure~\ref{fig:1d-alternatives}; see also Appendix~\ref{sec:estimands}). If the alternative is known to lie along this direction, then such tests can have optimal power \citep{Niu2022}. Otherwise, such tests can have low power if the portion of the signal that projects onto this direction is small. For this reason, it is desirable for a general-purpose variable selection method to have nontrivial power against multi-dimensional alternatives, i.e., not to detect departures from the null along only a single prespecified direction.
	
	\vspace{-0.15in}

	\paragraph{Computational speed.} The speeds of variable selection methods are often bottlenecked by running machine learning (ML) procedures to estimate functionals of $\law(\rx)$ and/or $\law(\ey \mid \rx)$ or by resampling to build null distributions. Many existing methods require either $O(p)$ ML fits or $O(p^2)$ resamples to test $p$ variables, which can be prohibitive when $p$ is large. In GWAS, for example, $(p, n) \approx (10^6, 10^5)$ is typical. Our goal is to design a method that requires only $O(1)$ ML fits and $O(p)$ resamples.
	
	\vspace{-0.15in}

\paragraph{Returning $p$-values for each variable.}
Variable selection methods that return $p$-values are preferred for two reasons. First, $p$-values permit FWER adjustments, whereas non-$p$-value methods generally lack powerful FWER control, though they may accommodate alternatives like $k$-FWER or FDR (see Section~\ref{sec:existing-approaches}). Second, practitioners expect $p$-values: they drive common visualizations (volcano, QQ, and Manhattan plots) and are routinely inspected to assess significance. GWAS, a prototypical modern application of large-scale variable selection, exemplifies these practices. The field explicitly adopts FWER control, maintaining the genome-wide $p$-value threshold of $5 \times 10^{-8}$ \citep{Risch1996b} for three decades. Furthermore, variant–disease associations must be supported with $p$-values for inclusion in the GWAS Catalog \citep{Sollis2023a}.

	\subsection{An overview of existing approaches} \label{sec:existing-approaches}
	
	We evaluate a set of leading methods for variable selection based on the above criteria (Table~\ref{tab:comparison}), deferring additional discussion of related work to Section~\ref{sec:related-work}. These methods fall into two categories: \textit{model-X} and \textit{doubly robust}. Both classes of methods include black-box ML estimates of the nuisances $\law(\rx)$ and/or $\law(\ey \mid \rx)$ (or functionals thereof). 

	\setlength{\tabcolsep}{4.5pt}
	\begin{table}[h]
		\centering
		\begin{tabular}{l|cc|cc|c}
		& \multicolumn{2}{c|}{Model-X} & \multicolumn{2}{c|}{Doubly robust} & Best of both \\ 
		\cline{2-6}
		& knockoffs & HRT & GCM & PCM & tPCM \\ 
		\hline
		Type-I control with nuisances & \checkmark  & \checkmark (new)  & \checkmark & \checkmark & \checkmark \\
		Power vs multi-dim alternatives & \checkmark & \checkmark &  & \checkmark & \checkmark \\ 
		$O(1)$ ML fits and $O(p)$ resamples & \checkmark &  &  &  & \checkmark \\ 
		Produces $p$-values for each variable &  & \checkmark & \checkmark & \checkmark & \checkmark \\ 
		\end{tabular}
		\caption{Comparison of four existing variable selection methods and the proposed method (tPCM) based on four statistical and computational criteria. The Type-I error control of HRT with nuisance estimation is established in this paper (Corollary~\ref{cor:hrt-type-I-error-control}).}
		\label{tab:comparison}
	\end{table}

	\vspace{-0.25in}

	\paragraph{Model-X methods.} This class of methods is based on the model-X framework \citep{CetL16}, where $\law(\rx)$ is assumed known but no assumptions are made about $\law(\ey \mid \rx)$. Such methods include model-X knockoffs \citep{CetL16} and the holdout randomization test (HRT; \cite{Tansey2018}). These methods are commonly deployed in practice by fitting $\law(\rx)$ in-sample. For knockoffs, recent works \citep{Fan2018a, Fan2023, Fan2025} have provided conditions under which asymptotic FDR control is maintained in this regime. Such a result was not available for HRT, but we provide one in this paper (see Section~\ref{sec:our-contributions}). We now describe each method. Knockoffs involves constructing negative control knockoff variables $\rxk$ that mimic the dependence structure of the original predictors $\rx$, and using test statistics that contrast the importance of the original and knockoff variables. This method does not provide $p$-values for each variable, which makes it incompatible with powerful FWER control at levels $\alpha < 0.5$, but still allows control of the FDR and $k$-FWER \citep{Janson2016,CetL16}. On the other hand, HRT learns $\hat m(\rx) \approx \E[\ey \mid \rx]$ on $n$ samples and quantifies the significance of the $j$th variable by comparing the prediction error $\sum_{i = 1}^n (\ey{i} - \hat m(\rx{i}{\bullet}))^2 $ on the remaining $n$ samples to its distribution under resampling $\cx{\bullet}{j} \mid \mx{\bullet}{\mj}$. HRT requires up to $O(p^2)$ resamples to test all $p$ variables, which can be expensive.

	\vspace{-0.15in}

	\paragraph{Doubly robust methods.} These methods test the CI hypothesis~\eqref{eq:conditional-independence} for a single variable $j$ based on asymptotically normal estimates of the functional 
	\begin{equation}
	\psi_j(g) = \E[(g(\rx) - \E[g(\rx) \mid \rx{\mj}])(\ey - \E[\ey \mid \rx{\mj}])],
	\label{eq:dr-functional}
	\end{equation}
	which vanishes under CI for any fixed function $g$. The generalized covariance measure (GCM) test \citep{Shah2018} employs $g_{\text{GCM}}(\rx) = \ex{j}$, which leads to power against a one-dimensional set of alternatives (see Appendix~\ref{sec:estimands}). The projected covariance measure (PCM; \cite{Lundborg2022a}) extends this idea to obtain power against multi-dimensional alternatives by setting $g_{\text{PCM}}(\rx) = \E[\ey \mid \rx]$,
	where $\E[\ey \mid \rx]$ is learned on a portion of the data. For testing CI using functionals of the form \eqref{eq:dr-functional}, taking $g_{\text{PCM}}(\rx) = \E[\ey \mid \rx]$ is in fact optimal in a precise sense \citep{Lundborg2022a}. Moreover,
	\begin{equation}
	\text{CI} \quad \Longrightarrow \quad \psi_j^{\text{PCM}} = \psi_j(g_{\text{PCM}}) = 0 \quad \Longrightarrow \quad \psi_j^{\text{GCM}} = \psi_j(g_{\text{GCM}}) = 0,
	\end{equation}
	\begin{wrapfigure}[8]{r}{0.325\textwidth}
		\vspace{-0.24in}
  	\centering
	  \includegraphics[width=0.325\textwidth]{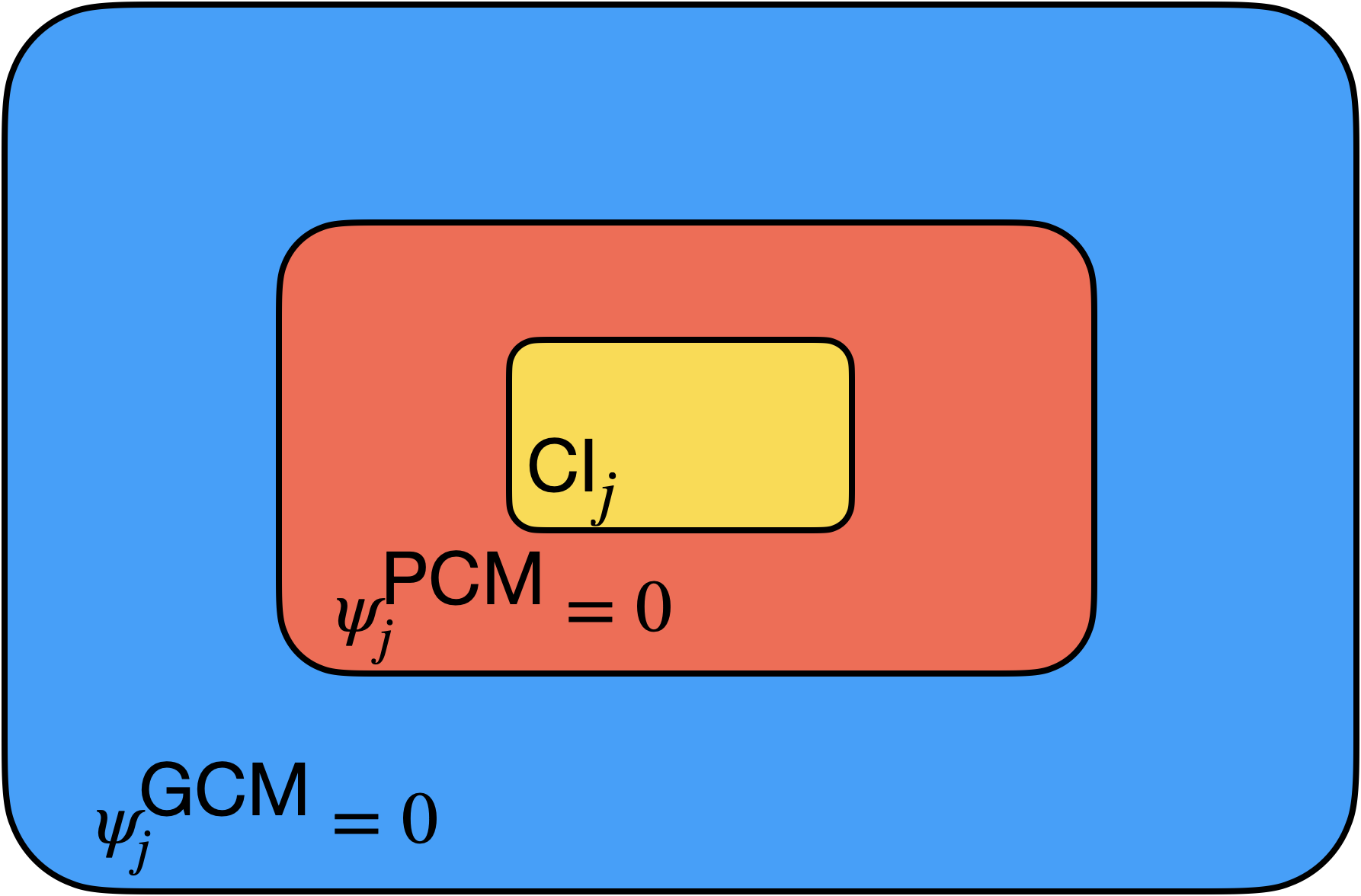}
	  \captionsetup{skip=5pt} 
  	\caption{Nested nulls.}
  	\label{fig:nested-nulls}
	\end{wrapfigure}
	so PCM is sensitive to a broader range of departures from CI than GCM (Figure~\ref{fig:nested-nulls}). Both methods involve ML steps to estimate the quantities $\E[g(\rx) \mid \rx{\mj}]$ and $\E[\ey \mid \rx{\mj}]$, and have the double robustness property \citep{Smucler2019a} that Type-I error is controlled asymptotically if the two estimation errors converge to zero and their product decays at the rate of $o(n^{-1/2})$. Since these methods were designed for a single CI test, their direct application to the high-dimensional variable selection problem is computationally challenging because it would involve multiple ML fits for each variable $j$. For GWAS, imagine fitting a million ML models based on a hundred thousand observations each!
			
	\subsection{Our contributions} \label{sec:our-contributions}

	Given the challenges inherent in variable selection, the application of any existing method involves sacrificing at least one of the desirable statistical or computational properties (Table~\ref{tab:comparison}). \textit{By leveraging ideas from both model-X and doubly robust literatures, we develop a new method, tower PCM (tPCM), that satisfies all four criteria simultaneously.}
	
	Our approach is to resolve the computational challenge of deploying doubly robust tests in the variable selection setting by taking a model-X perspective. If we had an approximation to $\law(\rx)$, could we circumvent the need to fit $\E[\ey \mid \rx{\mj}]$ for each $j$? The tower property suggests a way to proceed:
	\begin{equation}
	\E[\ey \mid \rx{\mj}] = \E[\E[\ey \mid \rx] \mid \rx{\mj}].
	\label{eq:tower-property}
	\end{equation}
	The inner expectation $\E[\ey \mid \rx]$ can be estimated using a single ML fit involving all predictors. The outer expectation can be evaluated using our approximation to $\law(\rx)$, which implies a conditional distribution $\law(\ex{j} \mid \rx{\mj})$. If these conditional distributions can be computed efficiently, then so can the outer expectation. Efficiently computable conditionals are often available in applications of model-X methods, like GWAS \citep{SetC17}, where the hidden Markov model (HMM) is the commonly accepted model for the joint distribution of genetic variants \citep{scheet2006fast} and admits efficient conditional sampling \citep{Rabiner1989}. Applying this idea to accelerate the PCM test leads to the proposed method, tPCM. 
	
	While the tower property idea~\eqref{eq:tower-property} is straightforward, the verification of tPCM's asymptotic Type-I error control is technically challenging due to the interplay between errors in the estimation of $\law(\rx)$ and $\E[\ey \mid \rx]$. \textit{We overcome these challenges to prove asymptotic uniform Type-I error control under doubly robust type conditions on the estimation errors (Theorem~\ref{thm:tower-pcm-type-I-error}).} tPCM also produces $p$-values for each variable by construction, and inherits power against multi-dimensional alternatives from PCM. Finally, tPCM satisfies the desired computational constraint by requiring only two ML fits (one for $\law(\rx)$ and one for $\E[\ey \mid \rx]$) and $O(p)$ resamples, the latter to approximate the outer expectation~\eqref{eq:tower-property} using a constant number of resamples per variable. We note that simply counting the number of ML fits and resamples is only a rough proxy for computational cost, as these computational units can vary in their expense. The computational advantage of tPCM, while not universal, is most pronounced where $\law(\rx)$ has structure (like the HMM structure ubiquitous in GWAS) that can be exploited for efficient fitting and conditional sampling. Our numerical simulations illustrate this advantage; see Figure~\ref{fig:comp_stat_efficiency} for a preview and Section~\ref{sec:sim-study} for details. Among the methods with competitive power (tPCM, PCM, and HRT), tPCM is faster than PCM and HRT by factors of 10 and 30, respectively. \textit{In larger experiments (Figure~\ref{fig:computation-rf}, right), the advantage over these methods grows to factors of 130 and 140.} 

	\begin{figure}[h!]
		\centering
		\includegraphics[width = 0.7\textwidth]{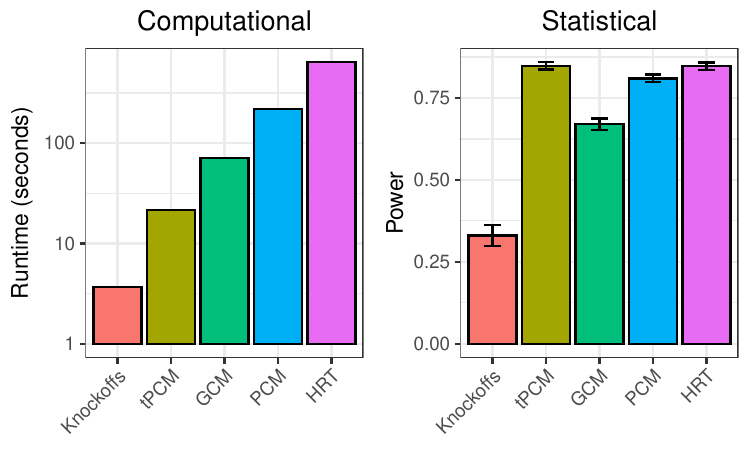}
		\caption{A comparison of the computational and statistical performance of our method, tPCM, with state-of-the-art competitors controlling FDR at level $q = 0.1$ on a problem instance of size $(2n, p) = (3000, 50)$. All methods fit $\E[\ey \mid \rx]$ via a random forest and knockoffs, tPCM, and HRT fit $\law(\rx)$ using an HMM.}
		\label{fig:comp_stat_efficiency}
	\end{figure}	
		
	tPCM is a hybrid of model-X and doubly robust methods that combines the strengths of both approaches to overcome their individual limitations. While we motivated tPCM as a computational acceleration of PCM using model-X ideas, it can also be viewed as an acceleration of HRT using doubly robust ideas, replacing the resampling-based null distribution by an asymptotic one. In fact, \textit{we prove that the tPCM is asymptotically equivalent to both PCM and HRT (Theorem~\ref{thm:HRT-PCM-equivalence}), revealing that the HRT and PCM themselves are asymptotically equivalent and that HRT is doubly robust (adding the new checkmark in Table~\ref{tab:comparison}).} We have thus demonstrated that the model-X and doubly robust literatures independently arrived at asymptotically equivalent tests, and have proposed a new method at the intersection that improves computationally upon both. We had begun probing this fruitful intersection in the context of the GCM test \citep{Niu2022}, and this work further establishes it as a source of new insights and methodologies.

	\subsection{Paper outline}

	In Section~\ref{sec:background}, we review PCM and HRT. In Section~\ref{sec:doubly robust-tower-PCM}, we introduce tPCM, compare its computational complexity to those of existing methods, and state our result on its asymptotic Type-I error control. In Section~\ref{sec:asym-equiv-tower-PCM}, we present our results on the asymptotic equivalence of tPCM, PCM, and HRT. In Section~\ref{sec:sim-study}, we present a simulation study comparing tPCM to existing methods. In Section~\ref{sec:data-analysis}, we apply tPCM to a breast cancer dataset. We conclude in Section~\ref{sec:discussion}. Proofs and additional numerical results are deferred to the Appendix.

	\section{Background: The PCM test and the HRT} \label{sec:background}

In this section, we define the PCM test and the HRT. In preparation for this, we introduce some notation. Let
\begin{equation}
m(\rx) \equiv \E_\law [\ey \mid \rx] \quad \textnormal{and} \quad m_j(\rx{\mj}) \equiv \E_\law [\ey \mid \rx{\mj}].
\end{equation}
For a fixed function $\widehat f(\rx)$, we will denote
\begin{equation}
m_{\widehat f}(\rx{\mj}) \equiv \E_\law [\widehat f(\rx) \mid \rx{\mj}].
\end{equation}
Many of the quantities we introduce will be indexed by $j$, though at times, we omit this index to lighten notation. We do not assume the model-X setting, so we treat $\law(\rx)$ as unknown. Finally, the set of null laws $\law$ for predictor $j$ is explicitly given by
\begin{equation}
	\nulllaws_{n,j} \equiv \{\law: \law(\ex{j}, \ey \mid \rx{\mj}) = \law(\ex{j}\mid\rx{\mj}) \times \law(\ey\mid\rx{\mj})\}.
\end{equation}

The algorithms reviewed in this section and the proposed tPCM test in Section~\ref{sec:doubly robust-tower-PCM} involve arbitrary black-box estimators of nuisance functions (e.g., high-dimensional regressions or machine learning methods). The theoretical guarantees for these methods require these estimators to achieve a certain level of accuracy. For more concrete examples of such nuisance estimators, see Sections~\ref{sec:sim-study}, \ref{sec:data-analysis}, and~\ref{sec:examples-type-1-error}. Additionally, the algorithms below involve sample splitting, which for simplicity we present as involving two equal-sized splits. However, our theory can accommodate splitting proportions besides 0.5, and we use different split proportions in the simulations and data analysis.

\subsection{Projected covariance measure}\label{sec:vPCM}
In this section, we describe a ``vanilla'' version of the PCM methodology proposed in \citet{Lundborg2022a}, which we shall refer to as vPCM (Algorithm \ref{alg: pcm}). vPCM is a special case of the slightly more involved PCM, which retains its essential ingredients but omits some steps that do not affect the asymptotic statistical performance. Explicitly, we omit steps 1 (iv) and 2 of Algorithm 1 in \citet{Lundborg2022a}. 

We begin by splitting our data into $D_1 \cup D_2$, with $D_1$ and $D_2$ containing $n$ samples each.  We estimate $\widehat m(\rx) \equiv \widehat{\E}[\ey \mid \rx]$ on $D_2$, and then we regress it onto $\rx{\mj}$ using $D_2$ to obtain $\widecheck m_j(\rx{\mj})$. We denote the difference of the two quantities $\widehat f_j(\rx) \equiv \widehat m(\rx) - \widecheck m_j(\rx{\mj})$. The quantity $\widehat f_j(\rx)$ is then tested for association with $\ey$, conditionally on $\rx{\mj}$, using $D_1$. To this end, we regress $\ey$ on $\rx{\mj}$ using $D_1$ to obtain an estimate of $\E[\ey|\rx{\mj}]$, which we call $ \widetilde{m}_j(\rx{\mj})$. We also regress $\widehat f_j(\rx)$ on $\rx{\mj}$ using $D_1$ to obtain $\widehat m_{\widehat f_j} (\rx{\mj})$.  We define the product of residuals stemming from the two regressions as 
\begin{equation}
L_{ij} \equiv (\ey{i} - \widetilde{m}_j(\rx{i}{\mj}))(\widehat f_j(\rx{i}{\bullet}) - \widehat m_{\widehat f_j} (\rx{i}{\mj})) 
\label{eq:l-i-j}
\end{equation}
and define the vanilla PCM statistic for predictor $j$ as:
\begin{equation}
	T^{\textnormal{vPCM}}_j \equiv \frac{\frac{1}{\sqrt n} \sum_{i=1}^n L_{ij}}{\sqrt{\frac{1}{n} \sum_{i=1}^n L_{ij}^2 - \left(\frac{1}{n} \sum_{i=1}^n L_{ij}\right)^2 }}.
\label{eq:vpcm-test-stat}
\end{equation}
Under the null hypothesis, $T_j^{\textnormal{vPCM}}$ is a sum of random quantities and for sufficiently large $n$ and under appropriate conditions, the central limit theorem (CLT) is expected to apply. Hence, we can compare our statistic to the quantiles of the normal distribution and reject for large values.
Our test is defined as
$$
\phi_j^{\textnormal{vPCM}}(\mx, \cy) \equiv \mathbbm{1}\left(T^{\textnormal{vPCM}}_j(\mx, \cy) > z_{1-\alpha}\right).
$$
Aside from the fitting of $\widehat m(\rx)$, the steps are repeated for each predictor $j = 1,\ldots,p$.

\begin{algorithm}[h]
	\SetAlgoLined 
	\KwIn{Data $\{(\rx{i}{\bullet}, \ey{i})\}_{i = 1, \dots, 2n}$}
	Split the data into $D_1 \cup D_2$, with $D_1$ and $D_2$ containing $n$ samples each.
	
	Estimate $\E[\ey \mid \rx]$ on $D_2$, call it $\widehat m(\rx)$.
	
	\For{$j \gets 1$ \KwTo $p$}{
		Regress $\widehat m(\rx)$ on $\rx{\mj}$ using $D_2$ to obtain $\widecheck m_j(\rx{\mj})$ and define $\widehat f_j(\rx)  \equiv \widehat m(\rx) - \widecheck m_j(\rx{\mj})$.

		Using $D_1$, regress $\cy$ on $\rx{\mj}$ to obtain an estimate $\widetilde{m}_j(\rx{\mj})$ of $\E[\ey|\rx{\mj}]$.

		Also on $D_1$, regress $\widehat f_j(\rx)$ on $\rx{\mj}$ to obtain $\widehat m_{\widehat f_j} (\rx{\mj})$.

		Compute $T^{\textnormal{vPCM}}_j$ based on equations~\eqref{eq:l-i-j} and~\eqref{eq:vpcm-test-stat}.

		Set $p_j \equiv 1 - \Phi(T^{\textnormal{vPCM}}_j)$.
	}
	\Return $\{p_j\}_{j = 1, \dots, p}$.
	
	\caption{Vanilla PCM}
	\label{alg: pcm}
\end{algorithm}

The primary disadvantage of Algorithm \ref{alg: pcm} is that it requires 3$p$+1 machine learning fits, which we would expect to be computationally difficult when $p$ is large. On the other hand, Algorithm \ref{alg: pcm} does not require any resampling.

\subsection{Holdout Randomization Test} 
In this section, we describe the HRT (Algorithm~\ref{alg: hrt}), which is identical to Algorithm 2 of \citet{Tansey2018} except the estimation of $\law(\rx)$, which the latter authors assumed known. As before, we divide our data into two halves, $D_1$ and $D_2$. On $D_2$, we learn the function $\widehat{m}(\rx) \equiv \widehat{\E}[\ey \mid \rx]$ and the law $\lawhat(\rx)$. On $D_1$, we compute the mean-squared error (MSE) test statistic
\begin{equation}\label{eq:HRT-test}
	T^{\textnormal{HRT}}(\mx, \cy) \equiv \frac{1}{n} \sum_{i=1}^n (\ey{i} - \widehat{m}(\rx{i}{\bullet}))^2.
\end{equation}
Next, we exploit the fact that under the null hypothesis, the conditional distribution $\law(\ex{j} \mid \rx{\mj}, \ey)$ is the same as $\law(\ex{j} \mid \rx{\mj})$, for which we have the estimate $\lawhat(\ex{j}|\rx{\mj})$. Therefore, we can approximate the distribution of $T^{\textnormal{HRT}}(\mx, \cy)$ conditional on $\cy, \mx{\bullet}{\mj},D_2$ by resampling $\exk{i}{j} \indsim \lawhat(\ex{i}{j}|\rx{i}{\mj})$ $B_{\textnormal{HRT}}$ times for each $i=1,\ldots,n$. In particular, we can approximate the following conditional quantile:
\begin{equation*}
	C_j(\cy,\mx{\bullet}{\mj}) \equiv \Q_{\alpha}\left[\frac{1}{n} \sum_{i=1}^n (\ey{i} - \widehat m(\exk{i}{j},\rx{i}{\mj}))^2 \mid \cy, \mx{\bullet}{\mj},D_2\right].
\end{equation*}
Here, and in line 8 of Algorithm~\ref{alg: hrt}, $(\exk{i}{j}, \rx{i}{\mj})$ represents the vector obtained from $\rx{i}{\bullet}$ by replacing the $j$th element with $\exk{i}{j}$. The HRT for predictor $j$ is then defined as
\begin{equation*}
\phi_j^{\textnormal{HRT}}(\mx, \cy) \equiv \mathbbm{1}\left(T^{\textnormal{HRT}}(\mx, \cy) \leq C_j(\cy,\mx{\bullet}{\mj}) \right).
\end{equation*}
The steps are than repeated for each predictor $j = 1,\ldots,p$. Algorithm~\ref{alg: hrt} describes how to compute the HRT $p$-values for each variable. 

\begin{algorithm}[h]
	\SetAlgoLined 
	\KwIn{Data $\{(\rx{i}{\bullet}, \ey{i})\}_{i = 1, \dots, 2n}$, number of resamples $B_{\textnormal{HRT}}$.}
	Split the data into $D_1 \cup D_2$, with $D_1$ and $D_2$ containing $n$ samples each.
	
	Estimate $\E[\ey \mid \rx]$ on $D_2$, call it $\widehat m(\rx)$.

	Estimate $\law(\rx)$ on $D_2$, call it $\lawhat(\rx)$.

	Compute test statistic $T^{\textnormal{HRT}}$ as in equation~\eqref{eq:HRT-test}.
	
	\For{$j \gets 1$ \KwTo $p$}{
		\For{$b \gets 1$ \KwTo $B_{\textnormal{HRT}}$}{
			Sample $\exk{i}{j} \sim \lawhat(\ex{j} \mid \rx{\mj} = \rx{i}{\mj})$ for all $i \in D_1$.

			Compute $\widetilde T_j^b \equiv \frac{1}{n} \sum_{i=1}^n (\ey{i} - \widehat{m}(\exk{i}{j}, \rx{i}{\mj}))^2$.
		}

		Set $p_j \equiv \frac{1}{B_{\textnormal{HRT}} + 1} \left(1 + \sum_{b = 1}^{B_{\textnormal{HRT}}} \mathbbm{1}\left[ T^{\textnormal{HRT}} \leq \widetilde T_j^b \right] \right)$.
	}
	\Return $\{p_j\}_{j = 1, \dots, p}$.
	\caption{Holdout Randomization Test}
	\label{alg: hrt}
\end{algorithm}

The primary disadvantage of Algorithm \ref{alg: hrt} is that it requires $p \times B_{\textnormal{HRT}}$ resamples. When using the Bonferroni correction to control the FWER, $B_{\textnormal{HRT}}$ must be at least $p/\alpha$ for nontrivial power. On the other hand, an attractive property of Algorithm \ref{alg: hrt} is that it requires only two machine learning fits.

\section{Best of both worlds: Tower PCM}\label{sec:doubly robust-tower-PCM}

In this section, we introduce the tower PCM method (Section~\ref{sec:tPCM}), followed by a discussion of its computational and statistical properties (Sections~\ref{sec:tpcm-computation} and~\ref{sec:tpcm-type-i-error-control}, respectively).

\subsection{The tower PCM algorithm} \label{sec:tPCM}

The computational bottleneck in the application of the PCM test (Algorithm~\ref{alg: pcm}) is the repeated application of regressions to obtain $\E[\ey \mid \rx{\mj}]$ for each $j$. Our key observation is that if we compute estimates $\lawhat(\rx)$ and $\widehat m(\rx) \equiv \widehat{\E}[\ey \mid \rx]$ (as in the first two steps of the HRT), then we can construct estimates of $\E[\ey \mid \rx{\mj}]$ for each $j$ without doing any additional regressions. Indeed, note that by the tower property of expectation, we have
\begin{equation*}
\E[\ey \mid \rx{\mj}] \equiv \E_{\law}[m(\rx) \mid \rx{\mj}] \approx \E_{\law}[\widehat m(\rx) \mid \rx{\mj}, D_2] \approx \E_{\lawhat}[\widehat m(\rx) \mid \rx{\mj}, D_2] \equiv \widehat m_j(\rx{\mj}).
\end{equation*}
To compute the quantity $\widehat m_j$, we can use conditional resampling based on $\lawhat(\ex{j} \mid \rx{\mj})$. Unlike the HRT, however, the goal of conditional resampling is to compute expectations rather than tail probabilities, and therefore, much fewer conditional resamples are required. Note that in the case where the $\ex{j}$ are discrete, no resampling is required at all. Equipped with $\widehat m_j$, we can proceed as in the PCM test by computing products of residuals
\begin{equation}
	R_{ij} \equiv (\ey{i}  - \widehat m_j(\rx{i}{\mj}))(\widehat m(\rx{i}{\bullet}) - \widehat{m}_j(\rx{i}{\mj})),
\end{equation}
and constructing the test statistic	
\begin{align}
	T^{\textnormal{tPCM}}_j \equiv \frac{\frac{1}{\sqrt n} \sum_{i=1}^n R_{ij}}{\sqrt{ \frac{1}{n}\sum_{i=1}^n R_{ij}^2 - \left(\frac{1}{n} \sum_{i=1}^n R_{ij} \right)^2}} \equiv \frac{\frac{1}{\sqrt n} \sum_{i=1}^n R_{ij}}{\widehat \sigma_n},
\end{align}
which we expect is asymptotically normal under the null hypothesis. This yields the test
\begin{equation}
\phi_j^{\textnormal{tPCM}}(\mx, \cy) \equiv \mathbbm{1}\left(T^{\textnormal{tPCM}}_j(\mx, \cy) > z_{1-\alpha}\right).
\end{equation}
These steps lead to Algorithm~\ref{alg: tpcm}.

\begin{algorithm}[h]
	\SetAlgoLined 
	\KwIn{Data $\{(\rx{i}{\bullet}, \ey{i})\}_{i = 1, \dots, 2n}$, number of resamples $B_{\textnormal{tPCM}}$.}
	Split the data into $D_1 \cup D_2$, with $D_1$ and $D_2$ containing $n$ samples each.
	
	Estimate $\E[\ey \mid \rx]$ on $D_2$, call it $\widehat m(\rx)$.

	Estimate $\law(\rx)$ on $D_2$, call it $\lawhat(\rx)$.
	
	\For{$j \gets 1$ \KwTo $p$}{
		\eIf{$\ex{j}$ discrete}{Compute $\widehat{m}_j (\rx{i}{\mj}) \equiv \sum_{\exs{j} \in \mathcal{X}_j} \widehat{m} (\exk{i}{j} = \exs{j}, \rx{i}{\mj}) \lawhat(\ex{j} = \exs{j} \mid \rx{i}{\mj})$ for all $i \in D_1$.}{\For{$k \gets 1$ \KwTo $B_{\textnormal{tPCM}}$}{
			Sample $\exk{i}{j} \sim \lawhat(\ex{j} \mid \rx{i}{\mj})$ for all $i \in D_1$.
		}
		Compute $\widehat{m}_j (\rx{i}{\mj}) \equiv \frac{1}{B_{\textnormal{tPCM}}}\sum_{k=1}^{B_{\textnormal{tPCM}}} \widehat{m} (\exk{i}{j}, \rx{i}{\mj}) $ for all $i \in D_1$.}

		Define $R_{ij} \equiv (\ey{i}  - \widehat{m}_j (\rx{i}{\mj}))(\widehat m(\rx{i}{\bullet}) - \widehat{m}_j (\rx{i}{\mj}))$ for $i$ in $D_1$.

		Compute $T^{\textnormal{tPCM}}_j \equiv \frac{\frac{1}{\sqrt n} \sum_{i=1}^n R_{ij}}{\sqrt{ \frac{1}{n}\sum_{i=1}^n R_{ij}^2 - \left(\frac{1}{n} \sum_{i=1}^n R_{ij}\right)^2}}$.

		Set $p_j \equiv 1 - \Phi(T^{\textnormal{tPCM}}_j)$.
	}
	\Return $\{p_j\}_{j = 1, \dots, p}$.
	\caption{Tower PCM}
	\label{alg: tpcm}
\end{algorithm}

\subsection{Computational cost comparison} \label{sec:tpcm-computation}
In this subsection, we compare the computational cost of tPCM to that of PCM and HRT. To this end, we consider the following units of computation, which compose the methods considered (except model-X knockoffs, which involves less standard components):
\begin{enumerate}
	\item \texttt{ML(n$\times$1|n$\times$p)}: Training an ML model to predict a one-dimensional quantity from a $p$-dimensional quantity based on $n$ observations
	\item \texttt{ML(n$\times$p)}: Training an ML model to learn the joint distribution of a $p$-dimensional quantity based on $n$ observations.
	\item \texttt{sample(n$\times$1|n$\times$p)}: Sampling from the conditional distribution of a one-dimensional quantity given a $p$-dimensional quantity for each of $n$ observations, or in the case when the one-dimensional quantity is discrete, computing the conditional probabilities for each of $n$ observations.
	\item \texttt{predict(n$\times$1|n$\times$p)}: Predicting a one-dimensional quantity from a $p$-dimensional quantity for $n$ new data points using a fitted ML model.
\end{enumerate}
For simplicity, we ignore distinctions between $p$- and $(p-1)$-dimensional quantities, and $n$- and $2n$-dimensional quantities. The above quantities are loose proxies for computational cost, but there may be variability in each unit both within and across methods (e.g., fitting $\ey|\rx$ and $\ex{j}|\rx{\mj}$ are both captured by the symbol \texttt{ML(n$\times$1|n$\times$p)}). Table~\ref{tab: computational work} summarizes the units of computation required by each method, taking the special case of binary $\ex{j}$ for simplicity and excluding model-X knockoffs, whose computational cost is harder to quantify in general. For continuous $\ex{j}$, the computational costs stay the same except that for tPCM, \texttt{sample(n$\times$1|n$\times$p)} must be repeated $B_{\text{tPCM}} \cdot p$ times and \texttt{predict(n$\times$1|n$\times$p)} must be repeated $(1 + B_{\text{tPCM}}) \cdot p$ times. These modifications do not change the order of the total computational cost from the binary case. 
\begin{table}[h!]
	\centering
	\begin{tabular}{|l||c|c|c|c|}
		\hline
		& \texttt{ML(n$\times$1|n$\times$p)} & \texttt{ML(n$\times$p)} & \texttt{sample(n$\times$1|n$\times$p)} & \texttt{predict(n$\times$1|n$\times$p)} \\
		\hline
		tPCM & 1 & 1 & $p$ & $3p$ \\
		PCM & $3p + 1$ & 0 & 0 & $2p$ \\
		GCM & $2p$ & 0 & 0 & 0 \\
		HRT & 1 & 1 & $B_{\textnormal{HRT}} \cdot p$ & $B_{\textnormal{HRT}} \cdot p$ \\
		\hline
	\end{tabular}
	\caption{Computational work required by the methods considered, for binary $X_j$.}
	\label{tab: computational work}
\end{table}

Given Table~\ref{tab: computational work}, it is apparent that tPCM has a computational advantage over PCM and GCM in cases where (a) \texttt{ML(n$\times$1|n$\times$p)} is more expensive than \texttt{predict(n$\times$1|n$\times$p)}, (b) \texttt{ML(n$\times$1|n$\times$p)} is more expensive than \texttt{sample(n$\times$1|n$\times$p)}, and (c) \texttt{ML(n$\times$p)} is less expensive than running \texttt{ML(n$\times$1|n$\times$p)} $p$ times. Condition (a) is often true, while conditions (b) and (c) depend on the ML methods and distributions involved, but are often satisfied when $\law(\rx)$ has structure that can be exploited. We provide a concrete setting where tPCM is computationally advantageous in Example~\ref{ex:lasso-hmm}. However, we acknowledge that the advantage is not universal: In general settings where fitting $\law(\rx)$ and/or sampling from $\law(\ex{j} \mid \rx{\mj})$ is computationally intensive, PCM may outperform tPCM. On the other hand, the tPCM is generally less computationally expensive than the HRT, with the difference more pronounced to the extent that the $B_{\text{HRT}} \cdot p$ sampling and prediction steps are a significant portion of the HRT's total computation.

\begin{example} \label{ex:lasso-hmm}
Table~\ref{tab: computational work} provides only a rough accounting of the computational cost of each method, excluding knockoffs. Here, we provide a finer-grained analysis, including knockoffs, in the GWAS-inspired problem setting where $\ey \mid \rx$ is a sparse linear model and $\rx$ is an HMM with binary emissions and $K = O(1)$ hidden states \citep{SetC17}. We consider lasso regressions for all \texttt{ML(n$\times$1|n$\times$p)} steps for all methods via $O(1)$ iterations of coordinate descent \citep{friedman2010regularization} and assume for simplicity that all fitted models have $O(s)$ nonzero coefficients. To fit $\law(\rx)$, we employ $O(1)$ iterations of the Baum-Welch algorithm with forward–backward message passing, and to fit all conditionals, we use the proposal of \citet{Perduca2013}. With these choices, the computational costs of the methods compared are given in Table~\ref{tab:gwas-example-computations} (see Appendix~\ref{sec:computational-cost-appendix} for justifications). We find that tPCM is faster than PCM and GCM by a factor of $p/s$, which can be either a large constant or grow with $p$ depending on the growth of $s$. tPCM is also faster than HRT by a factor of $p$. Finally, knockoffs is the fastest method in this setting. 
\end{example}

\begin{table}[h!]
	\centering
	\begin{tabular}{l|c|c|c|c|c}
		& tPCM & PCM & GCM & HRT & knockoffs \\
		\hline
		Cost & $O(nps)$ & $O(np^2)$ & $O(np^2)$ & $O(np^2s)$ & $O(np)$ \\
	\end{tabular}
	\caption{Computational work required by the methods considered, for binary $X_j$.}
	\label{tab:gwas-example-computations}
\end{table}

\subsection{Type-I error control and equivalence to the oracle test} \label{sec:tpcm-type-i-error-control}

In this section, we establish the Type-I error control of the tPCM test. To this end, we show that the tPCM test is asymptotically equivalent to an oracle test. For the remainder of this section, we focus on the test of $H_{0j}$ for a single predictor $j$, and sometimes omit the index $j$ to lighten the notation. We denote a sequence of null distribution by $\law_n \in \nulllaws_{n,j}$.

To define the oracle test, we begin by defining the residuals 
\begin{equation}
\varepsilon_i \equiv \ey{i} - m(\rx{i}{\mj}) \quad \text{and} \quad \xi_i = \widehat{m}(\rx{i}{\bullet}) - \E_{\law_n} [\widehat{m}(\rx{i}{\bullet}) |\rx{i}{\mj},D_2],
\end{equation}
Note that $\xi_i$ is defined in terms of the estimated $\widehat m$ rather than the true $m$. The ``oracle'' portion consists of access to the true $\law(\rx)$ to compute the conditional expectation term. Letting 
\begin{equation}
\sigma_n^2 \equiv \mathrm{Var}_{\law_n}[\bm{\varepsilon}\bm{\xi} | D_2], 
\end{equation}
the oracle test is defined as
\begin{equation}
\phi_j^{\textnormal{oracle}}(\mx, \cy) \equiv \mathbbm{1}\left(T^{\textnormal{oracle}}_j(\mx, \cy) > z_{1-\alpha}\right), \quad \text{where} \quad T_j^{\textnormal{oracle}} \equiv \frac{1}{\sqrt n \sigma_n} \sum_{i=1}^n\varepsilon_i\xi_i.
\label{eq:oracle-test}
\end{equation}

Next, we define the asymptotic equivalence of two tests $\phi^{(1)}_n, \phi^{(2)}_n: (\mx, \cy) \mapsto [0,1]$ as the statement
$$
\lim_{n\to \infty} \P_{\law_n}[\phi_n^{(1)}(\mx, \cy) \neq \phi_n^{(2)}(\mx, \cy)] = 0.
$$
The following set of properties will ensure the equivalence of $\phi_j^{\text{tPCM}}$ and $\phi_{j}^{\text{oracle}}$. The first condition bounds the conditional variance of the error $\varepsilon_i$.
\paragraph{Bounded Conditional Variance:} 
\begin{equation} \label{eq:var-bounded}
	\begin{aligned}
		\exists c_1 > 0, \quad & \P_{\law_n}\left[ \max_{i \in [n]}\ \mathrm{Var}_{\law_n}(\varepsilon_i | \rx{i}{\mj},D_2) \leq c_1 \right] \to 1.
	\end{aligned}
\end{equation}
The next condition is written in terms of the conditional chi-square divergence
\begin{equation}
\chi^2(P, Q \mid \mathcal F) \equiv \E_{Q}\left[\left(\frac{dP}{dQ}-1\right)^2\mid \mathcal F\right],
\end{equation}
defined for measures $P$ and $Q$ and a $\sigma$-algebra $\mathcal F$. Using the conditional chi-square divergence to measure the error in the conditional distribution $\law_{\ex{j} \mid \rx{\mj}}$, we assume this conditional distribution is consistently estimated in the following sense,
\paragraph{Consistency  of $\lawhat_{\ex{i}{j}|\rx{i}{\mj}}$:}
\begin{gather}
	\label{eq:uniform-bound-div} \P_{\law_n}\left( \max_{i \in [n]}\chi^2\left(\lawhat_{\ex{i}{j}|\rx{i}{\mj}}, \law_{\ex{i}{j}|\rx{i}{\mj}} | D_2\right) < c_3\right) \to 1,\\ 
	\label{eq:consistency-of-chi-square-div}	E^2_{\lawhat,n} \equiv \frac{1}{n\sigma_n^2} \sum_{i=1}^n \chi^2\left(\lawhat_{\ex{i}{j}|\rx{i}{\mj}}, \law_{\ex{i}{j}|\rx{i}{\mj}} | D_2\right) \E_{\law_n}[ \xi_i^2 | \rx{i}{\mj}, D_2] \overset{p}{\to} 0.
\end{gather}
Note that these assumptions are on the entire fitted conditional distribution $\lawhat(\ex{j}|\rx{\mj})$ rather than on its functionals, making them stronger than those required to justify the PCM test \citep{Lundborg2022a}. We conjecture that these assumptions can be weakened, in the spirit of \citet{Katsevich2020a}, but leave this direction to future work. 

Similarly, we assume a consistent estimate of $m(\rx) = m_j(\rx{\mj})$ (this equality holding because we are under the null):
\paragraph{Consistency of $\widehat m$:}
\begin{equation}\label{eq:consitency-of-reg-func}
	(E'_{\widehat m,n})^2 \equiv \frac{1}{n\sigma_n^2}\sum_{i=1}^n\E_{\law_n}[( \widehat{m}(\rx{i}{\bullet}) - m_j(\rx{i}{\mj}))^2  \mid D_2,\rx{i}{\mj}]\E_{\law_n} [\xi_i^2|\rx{i}{\mj},D_2] \overset{p}{\to} 0.
\end{equation}
Also, we define the MSE for $\widehat m$ as follows:
\begin{equation*}
	E^2_{\widehat m,n} = \frac{1}{n\sigma_n^2}\sum_{i=1}^n\E_{\law_n}[( \widehat{m}(\rx{i}{\bullet}) - m_j(\rx{i}{\mj}))^2  \mid D_2,\rx{i}{\mj}], 
\end{equation*}
and assume a doubly robust type assumption which states 
\paragraph{Double Robustness condition:} 
\begin{align}\label{eq:doubly robust-rate}
	E_{\lawhat,n} \cdot E_{\widehat m,n} = o_p(n^{-1/2}).
\end{align}  
Finally, we assume the following Lyapunov-type condition,
\paragraph{Moment Condition:}
\begin{gather}\label{eq:CLT-condition}
	\frac{1}{\sigma_n^{2+\delta}}\E_{\law_n}\left[|\varepsilon\xi|^{2+\delta} \mid D_2\right] = o_P(n^{\delta/2}),
\end{gather}
The following theorem establishes the asymptotic validity of our proposed test under the aforementioned assumptions:
\begin{theorem}\label{thm:tower-pcm-type-I-error}
Let $\mathscr{R}_n$ be a regularity class such that the assumptions~\eqref{eq:var-bounded}, \eqref{eq:uniform-bound-div}-\eqref{eq:CLT-condition} are satisfied for any sequence $\law_n \in \nulllaws_n \cap \mathscr R_n$. Then, $\phi_j^{\textnormal{tPCM}}$ is asymptotically equivalent to $\phi_j^{\textnormal{oracle}}$. Additionally, tPCM is asymptotically uniformly size $\alpha$:
	$$
	\underset{n\to \infty}{\lim \sup}\sup_{\law_n \in \nulllaws_{n,j} \cap \mathscr{R}_n}\E_{\law_n}\left[\phi_j^{\textnormal{tPCM}}(\mx, \cy)\right] \to \alpha.
	$$
	
\end{theorem}

Two examples of settings where the assumptions of Theorem~\ref{thm:tower-pcm-type-I-error} are satisfied are given in Appendix~\ref{sec:examples-type-1-error}. All proofs are deferred to Appendix~\ref{sec:proofs}.

\section{Equivalence of tPCM with existing methods}\label{sec:asym-equiv-tower-PCM}

In this section, we will show that tPCM is asymptotically equivalent to vPCM (Section~\ref{sec:asymp-equivalence-vPCM-tPCM}) and HRT (Section~\ref{sec:asymp-equiv-HRT-PCM}).

\subsection{Asymptotic equivalence of vPCM and tPCM} \label{sec:asymp-equivalence-vPCM-tPCM}

To show the equivalence of tPCM and vPCM, we will show that the latter method is equivalent to the oracle test $\phi_j^{\textnormal{oracle}}$ defined in equation~\eqref{eq:oracle-test}, which we have shown is equivalent to tPCM (Theorem~\ref{thm:tower-pcm-type-I-error}). The conditions under which vPCM is equivalent to the oracle test echo those under which \citet{Lundborg2022a} showed that PCM controls type-I error. Define the in-sample MSE for the two regressions $\widetilde m_j$ and $\widehat{m}_{\widehat f_j}$ as follows:
\begin{gather*}
	\mathcal{E}_{\widetilde m} = \frac{1}{n} \sum_{i=1}^n (\widetilde m_j(\rx{i}{\mj}) - m_j(\rx{i}{\mj}))^2,\quad
	\mathcal{E}_{\widehat m_{\widehat f}} = \frac{1}{n\sigma_n^2} \sum_{i=1}^n (\widehat m_{\widehat f_j}(\rx{i}{\mj}) - m_{\widehat f_j}(\rx{i}{\mj}))^2.
\end{gather*}
We assume the following consistency conditions for the regression functions $\widetilde m_j$ and $\widehat m_{\widehat f_j}$:
\begin{gather}\label{eq:consistency-of-reg-of-Y-on-Z}
	\frac{1}{n\sigma_n^2}\sum_{i=1}^n (\widetilde m_j(\rx{i}{\mj}) - m_j(\rx{i}{\mj}))^2 \E[\xi_i^2\mid \rx{i}{\mj}] \overset{p}{\to} 0.\\
	\label{eq:consistency-of-reg-of-f-hat-on-Z}
	\frac{1}{n\sigma_n^2} \sum_{i=1}^n (\widehat m_{\widehat f_j}(\rx{i}{\mj}) - m_{\widehat f_j}(\rx{i}{\mj}))^2\E[\varepsilon^2_i \mid \rx{i}{\mj}] \overset{p}{\to} 0.
\end{gather} 	
We also assume a doubly robust condition on the product of MSEs:
\begin{equation}\label{eq:doubly robust-condition-vPCM}
	\mathcal{E}_{\widetilde m} \cdot \mathcal{E}_{\widehat m_{\widehat f}}  = o_p(n^{-1})
\end{equation}

\begin{theorem}\label{thm:equivalence-of-vPCM-oracle}
	Suppose $\law_n \in \nulllaws_n$ is a sequence of laws satisfying~\eqref{eq:CLT-condition}, \eqref{eq:consistency-of-reg-of-Y-on-Z}, \eqref{eq:consistency-of-reg-of-f-hat-on-Z}, and~\eqref{eq:doubly robust-condition-vPCM}. Then the test $\phi_j^{\textnormal{vPCM}}$ is asymptotically equivalent to the oracle test  $\phi_j^{\textnormal{oracle}}$.
\end{theorem}	

Combining this result with that of Theorem~\ref{thm:tower-pcm-type-I-error}, we obtain the following corollary.

\begin{corollary}
	Let $\law_n \in \nulllaws_n$ be a sequence of laws satisfying~\eqref{eq:var-bounded}, \eqref{eq:uniform-bound-div}, \eqref{eq:consistency-of-chi-square-div}, \eqref{eq:consitency-of-reg-func}, \eqref{eq:doubly robust-rate}, \eqref{eq:CLT-condition}, \eqref{eq:consistency-of-reg-of-Y-on-Z}, \eqref{eq:consistency-of-reg-of-f-hat-on-Z}, and \eqref{eq:doubly robust-condition-vPCM}. For any sequence $\law_n'$ of alternative distributions contiguous to the sequence $\law_n$, we have that $\phi_n^{\textnormal{vPCM}}$ is equivalent to $\phi_n^{\textnormal{tPCM}}$ against $\law_n'$ i.e.
	$$
	\lim _{n \rightarrow \infty} \mathbb{P}_{\law_n'}\left[\phi_j^{\textnormal{vPCM}}(\mx, \cy)=\phi_j^{\textnormal{tPCM}}(\mx, \cy)\right]=1 .
	$$
	In particular, these two tests have the same limiting power:
	$$
	\lim _{n \rightarrow \infty} \left\{\mathbb{E}_{\law_n'}\left[\phi_j^{\textnormal{vPCM}}(\mx, \cy)\right] - \mathbb{E}_{\law_n'}\left[\phi_j^{\textnormal{tPCM}}(\mx, \cy)\right]\right\} = 0.
	$$
\end{corollary} 	

Despite equivalence of vPCM and tPCM, we highlight an important distinction between these two methods. tPCM exclusively employs out-of-sample regressions, where the regressions are conducted on a different dataset from which the test statistic is evaluated. In contrast, vPCM utilizes both in-sample and out-of-sample regressions. As was pointed out by \citet{Lundborg2022a}, relying on in-sample regressions can be advantageous in finite samples. Nevertheless, the effects of this distinction vanish asymptotically. 

\subsection{Asymptotic Equivalence of HRT and tPCM}\label{sec:asymp-equiv-HRT-PCM}

We now show that the HRT is asymptotically equivalent to tPCM. While tPCM relies on a central limit theorem for a test statistic with an explicit normalizing factor, HRT uses a resampling-based null distribution. Despite their conceptual differences, we prove that their test statistics and rejection thresholds converge to the same limit under mild regularity conditions. The key insight is that the HRT test statistic can be decomposed into a leading term that matches the tPCM statistic, plus remainder terms that vanish asymptotically. Furthermore, the HRT's resampling-based cutoff converges to the standard normal quantile used by tPCM. See Section~\ref{sec:hrt-pcm-details}. We now list the technical assumptions required to control the remainder terms and establish equivalence.

\vspace{1mm}
\noindent \textbf{Assumptions controlling higher-order terms:}
\begin{gather}
	\label{eq:variance-of-m-hat-given-Z}
	\frac{1}{\sigma_n^2} \E\left[ \V\left(\widehat \xi^2 \mid \rx{\mj}, D_2\right)\mid D_2\right] = o_p(1).\\
	\label{eq:tower-regression-rate-assumption}
	\frac{1}{\sqrt n\sigma_n}\sum_{i=1}^n (\widehat m_j(\rx{i}{\mj}) - \E\left[\widehat m(\rx{i}{\bullet}) \mid \rx{i}{\mj}, D_2\right])^2 \overset{p}{\to} 0.\\
	\label{eq:variance-of-xi-rate-assumption}
	\frac{1}{\sqrt n \sigma_n} \sum_{i=1}^n\left(\V_{\lawhat}[\sxi_i \mid \rx{i}{\mj}, D_2] - \V_{\law}[\sxi_i \mid \rx{i}{\mj}, D_2]\right)\overset{p}{\to} 0.
\end{gather}

\vspace{-2mm}
\noindent \textbf{Assumptions for HRT cutoff convergence:}
\begin{gather}
	\label{eq:regression-of-Y-on-Z-is-consistent-wrt-L-hat} \frac{1}{n \sigma^2_n} \sum_{i=1}^n(\widehat m_j(\rx{i}{\mj})-m_j(\rx{i}{\mj}))^2 \E( \widetilde\sxi^2_i \mid \rx{i}{\mj},D_2)  \overset{p}{\to} 0.\\
	\label{eq:variance-consistency-condition}
	\frac{1}{n\sigma_n^2} \sum_{i=1}^n \left(\V_{\lawhat}[\sxi_i \mid \rx{i}{\mj}, D_2] - \V_{\law}[\sxi_i \mid \rx{i}{\mj}, D_2] \right)\E(\seps_i^2 \mid \rx{i}{\mj}) \overset{p}{\to} 0.\\
	\label{eq:variance-of-xi-hat-goes-to-zero}\frac{1}{{\sigma_n^2}}\E\left[\V(\widetilde{\pxi}^2 \mid \rx{\mj}, D_2)\mid D_2\right] \overset{p}{\to} 0.	
\end{gather}
\vspace{-4mm}
\noindent \textbf{Moment condition:}
\begin{gather}
	\label{eq:moment-condition-on-product-of-residuals-hat} 
	\frac{1}{\sigma_n^{2+\delta}}\E( |\peps\widetilde{\pxi}|^{2+\delta} \mid D_2) = o_p(n^\delta).
\end{gather}

These assumptions ensure that the remainder terms in the HRT statistic and the randomness in its resampling-based cutoff vanish asymptotically.
Assumptions~\eqref{eq:variance-of-m-hat-given-Z}--\eqref{eq:variance-of-xi-rate-assumption} control the quality of the estimated resampling distribution. Specifically, \eqref{eq:tower-regression-rate-assumption} ensures consistent estimation of the tower regression, and \eqref{eq:variance-of-xi-rate-assumption} guarantees stable variance estimation under the learned distribution.
The cutoff-related conditions~\eqref{eq:regression-of-Y-on-Z-is-consistent-wrt-L-hat}--\eqref{eq:variance-of-xi-hat-goes-to-zero} ensure that the HRT quantile converges to the standard normal cutoff. Finally, the moment condition~\eqref{eq:moment-condition-on-product-of-residuals-hat} ensures that fluctuations from heavy-tailed residuals remain controlled, enabling a valid CLT.
Together, these conditions imply that the HRT and tPCM tests rely on asymptotically equivalent statistics and thresholds.

\vspace{1mm}
\noindent We now state the main result:

\begin{theorem}\label{thm:HRT-PCM-equivalence} Suppose $\law_n \in \nulllaws_n$ is a sequence of laws satisfying the assumptions of Theorem~\ref{thm:tower-pcm-type-I-error}, as well as conditions \eqref{eq:variance-of-m-hat-given-Z}--\eqref{eq:moment-condition-on-product-of-residuals-hat}. Then, the HRT test is equivalent to the tPCM  test against $\law_n$.
\end{theorem}

\noindent Two examples of settings where the assumptions of Theorem~\ref{thm:HRT-PCM-equivalence} are satisfied are given in Section~\ref{sec:examples-HRT-PCM-equivalence}. One consequence of this theorem is the Type-I error control of the HRT beyond the model-X assumption. 

\begin{corollary} \label{cor:hrt-type-I-error-control}
	For a sequence of null laws $\law_n \in \nulllaws_n$ satisfying the assumptions of Theorem~\ref{thm:HRT-PCM-equivalence}, the HRT is asymptotically size $\alpha$.
\end{corollary}

\noindent Another consequence of Theorem~\ref{thm:HRT-PCM-equivalence} is that HRT and tPCM are equivalent under contiguous alternatives, and therefore have equal asymptotic power against contiguous alternatives. 

\begin{corollary}
	If $\law_n'$ is a sequence of alternative distributions contiguous to a sequence \( \law_n \) in \( \nulllaws_n \) satisfying the assumptions of Theorem \ref{thm:HRT-PCM-equivalence}, then the HRT and tPCM tests are asymptotically equivalent against \( \law_n' \). Furthermore, they have equal asymptotic power against $\law_n'$:
	\begin{equation}
		\lim _{n \rightarrow \infty}\left\{\E_{\law_n'}[\phi_j^{\textnormal{HRT}}(\mx, \cy)] - \E_{\law_n'}[\phi_j^{\textnormal{tPCM}}(\mx, \cy)]\right\} = 0.
	\end{equation}
\end{corollary}

By constructing a null distribution through resampling, the HRT accommodates arbitrarily complex machine learning methods for constructing test statistics, whose asymptotic distributions may not be known. However, we find that after appropriate scaling and centering, the resampling-based null distribution essentially replicates the asymptotic normal distribution utilized by the PCM test. Therefore, when testing a single hypothesis in large samples, the additional computational burden of resampling is unnecessary, as the equivalent PCM test can be applied instead. When dealing with a large number of samples and multiple hypotheses, the tPCM test becomes the natural candidate, combining the best aspects of the existing methodologies. For a small number of samples, the HRT remains an attractive option, as it does not rely on asymptotic normality.

	\section{Finite-sample assessment}\label{sec:sim-study}
	In this section, we investigate the finite-sample performance of tPCM with a simulation-based assessment of Type-I error,
	power, and computation time. We deployed all methods to control the FDR at level $\alpha = 0.1$. We consider a nonlinear, interacted model specification for the distribution of $\ey \mid \rx$, and an HMM specification for the distribution of $\rx$. The goal of the simulation is to corroborate the findings of the previous sections: (1) tPCM is computationally efficient, (2) tPCM controls the Type-I error, and (3) tPCM is as powerful as HRT and PCM. To highlight tPCM's versatility, we complement this simulation with another (Appendix \ref{sec:gam-banded-sim-results}) in which $\ey \mid \rx$ follows a generalized additive model, $\rx$ is drawn from a multivariate normal with a banded precision matrix, and the type-I error metric is the FWER. Code to reproduce the simulations in this section and the real data analysis in the next is available at \url{https://github.com/Katsevich-Lab/symcrt2-manuscript}.
	
	\subsection{Data-generating model}

	We pick $s$ of the $p$ variables to be nonnull at random. Let $\mathcal{S}$ denote the set of nonnulls. The data-generating model for $\ey \mid \rx$ is as follows:
	\begin{equation*}
	 \law_n(\ey \mid \rx) = N\left(\theta \cos\left(\sum_{j \in \text{nonnulls}} \ex{j}  + \sum_{j \neq k \in \text{nonnulls}} 0.2 \ex{j} \ex{k}\right),1 \right).
	\end{equation*}
	Meanwhile, the data-generating model for $\rx$ follows an HMM with binary observations and 5 hidden states. The transition probabilities are defined as follows: the last hidden state is absorbing, and the rest have probability stay\_prob of staying in the same state, and $1 - \text{stay\_prob}$ of moving up one state. The emission probabilities are defined as follows: the first state emits 0 or 1 with equal probability, while the rest emit zero with probability 0.9. Only the stay\_prob parameter is varied. Therefore, the entire data-generating process is parameterized by the five parameters $(n, p, s, \text{stay\_prob}, \theta)$; see Table \ref{tab:rf sim parameters}. In this section, we let $n$ denote the \emph{total} sample size, i.e. the combined size of $D_1$ and $D_2$. We vary each of the five parameters across five values each, setting the remaining to the default values (in bold).

	\begin{table}[h!]
		\centering
		\begin{tabular}{ccccc}
			\hline
			$n$ & $p$ & $s$ & stay\_prob & $\theta$ \\
			\hline
			2000 & 30 & 12 & 0.35 & 0.7 \\
			2250 & 40 & 16 & 0.425 & 0.75 \\
			\textbf{2500} & \textbf{50} & \textbf{20} & \textbf{0.5} & \textbf{0.8} \\
			2750 & 60 & 24 & 0.575 & 0.85 \\
			3000 & 70 & 28 & 0.65 & 0.9 \\
			\hline
		\end{tabular}
		\caption{The values of the sample size $n$, covariate dimension $p$, sparsity $s$, stay probability $\rho$, and signal strength $\theta$ used for the
		simulation study. Each of the parameters $n$, $p$, $s$, $\rho$, $\theta$ was varied among the values displayed in the
		table while keeping the other four at their default values, indicated in bold. For example, $p = 50$, $s = 20$, $\text{stay\_prob}= 0.5$, $\theta = 0.8$ were kept fixed while varying $n \in \{2000, 2250, 2500, 2750, 3000\}$.}
		\label{tab:rf sim parameters}
	\end{table}

	\subsection{Methodologies compared}

	We applied the five methods tPCM, HRT, vPCM (henceforth ``PCM''), oracle GCM, and model-X knockoffs. The first four were paired with the Benjamini-Hochberg (BH) procedure at level $\alpha = 0.1$ to control the FDR. For all methods except oracle GCM, quantities such as $\E[\ey \mid \rx]$ and $\E[\ey \mid \rx{\mj}]$ were fit using random forests; for oracle GCM, the true quantities were used for maximum power. tPCM and HRT exploited knowledge of the HMM structure and so $\law(\rx)$ was fit using a method designed for HMMs. The knockoffs implementation used the default random forest variable importance statistic built into \verb|knockoff| and HMM sampler from \verb|SNPknock|. We defer the remaining details to Appendix~\ref{sec:simulation-study-details}, and justify the omission of two additional methods in Section~\ref{sec:omission-justification}.

	\subsection{Simulation results}

	Power, runtime and FDR results for the primary simulation are presented in Figure \ref{fig:power-rf}, \ref{fig:computation-rf} (left), and \ref{fig:fdr-rf}, respectively. For the additional simulation, see Figure~\ref{fig:computation-rf} (right) as well as Figures~\ref{fig:fwer-gam} and~\ref{fig:power-gam} in Appendix~\ref{sec:gam-banded-sim-results}. Note that knockoffs was excluded from the additional simulation because it is based on FWER control, which knockoffs is not equipped to control. Based on these two simulation studies, we make the following observations:

	\begin{itemize}

		\item \textbf{Type-I error:} All methods control the Type-I error rates, indicating that in these settings, $\law(\ey \mid \rx)$ and $\law(\rx)$ are learned sufficiently well.
		\item \textbf{Power:}  In both settings, PCM, tPCM, and HRT are roughly tied for the highest power, although the power of PCM is slightly lower in the primary simulation, likely due to its not fully exploiting the HMM structure in learning the nuisances (it is difficult to do so within PCM). Oracle GCM has significantly lower power because the projection of the true alternative onto the alternative direction it is powerful against is small (recall Section~\ref{sec:desirable-properties} and Figure~\ref{fig:1d-alternatives}). In the primary simulation, knockoffs also has noticeably lower power than the top three methods, likely due to its test statistic (the \verb|knockoff| package default for random forests) not being as sensitive as those used by tPCM, PCM, and HRT.
		\item \textbf{Computation:} Among the three most powerful methods, tPCM is the fastest across both simulations. The acceleration relative to the next closest method, PCM, ranges from about 10$\times$ (Figure~\ref{fig:computation-rf}, left) to about 130$\times$ (Figure~\ref{fig:computation-rf}, right). The acceleration relative to HRT is even more dramatic, ranging from about 30$\times$ to about 140$\times$. Considering tPCM and PCM, we find in the primary simulation that the scaling of logarithmic runtime with $p$ occurs at approximately the same rate across the two methods, suggesting that tPCM offers a constant factor speedup. In the additional simulation, the logarithmic runtime of tPCM scales less steeply with $p$ than that of PCM, suggesting computational savings as a power of $p$. 
 	\end{itemize}

	\begin{figure}[H]
		\centering
		\includegraphics[width = 0.95\textwidth]{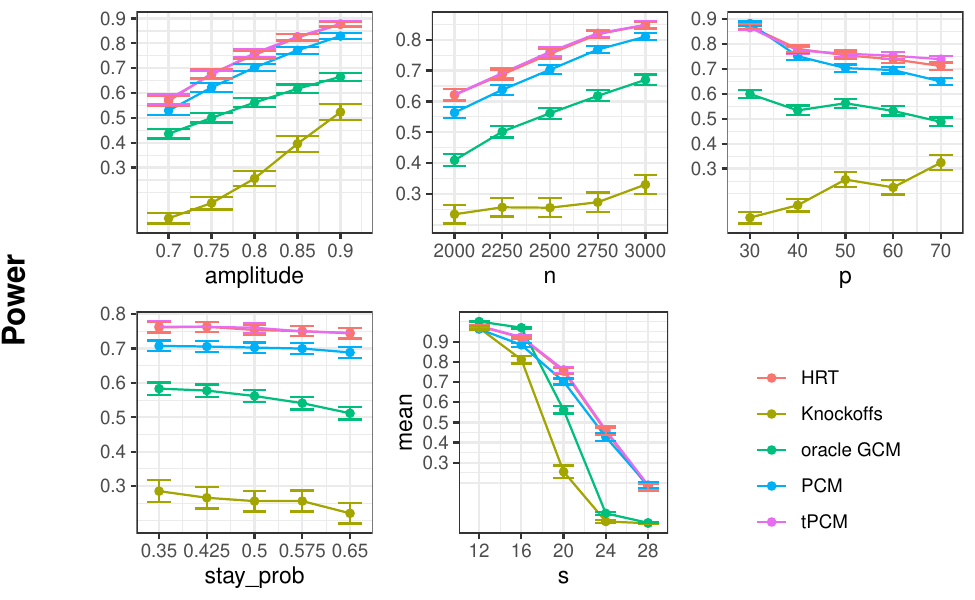}
		\caption{Power: in each plot, we vary one parameter. Each point is the average of 400 Monte Carlo replicates, and the error bars are the average $\pm 2 \times \widehat{\sigma}_p$, where $\widehat{\sigma}_p$ is the Monte Carlo standard deviation divided by $\sqrt{400}$.}
		\label{fig:power-rf}
	\end{figure}

	\begin{figure}[H]
		\centering
		\includegraphics{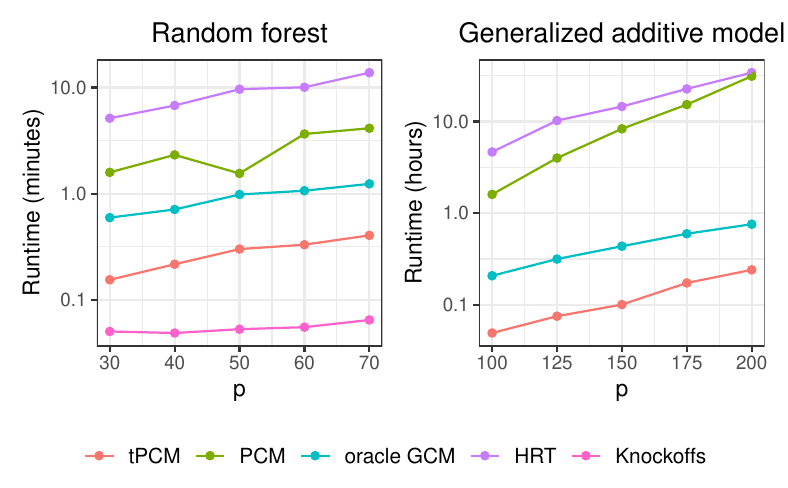}
		\caption{Runtime as a function of dimension in the primary simulation (left) and the additional simulation (right). Each point is the average of 400 Monte Carlo replicates.}
		\label{fig:computation-rf}
	\end{figure}

\section{Application to breast cancer dataset} \label{sec:data-analysis}
\subsection{Overview of the data}
As a final illustration of our method, we apply tPCM to a breast cancer dataset from \citet{Curtis2012}, which has been previously analyzed in the statistical literature by \citet{Liu2020} and \citet{Li2021c}. The data consist of $n = 1396$ positive cases of breast cancer categorized by stage (the outcome variable) and $p = 164$ genes, for which the expression level (mRNA) and copy number aberration
(CNA) are measured. We seek to discover genes that are associated with stage of breast cancer, conditional on the remaining genes. Statistically, we set the false discovery rate to be $\alpha = 0.1$. The data is preprocessed using the same steps as in \citet{Liu2020}; we refer the reader to Appendix E of \citet{Liu2020} for more details. The stage of cancer outcome is binary, and the gene expression predictors are continuous.

\subsection{Methods and their implementations}

As in the simulation study, we applied five methods to the data, which were HRT, tPCM, PCM, tower GCM (tGCM), and knockoffs. The methods are similar to those from the simulations, and we again fit $\E[\ey \mid \rx]$ and $\E[f_j(\rx) \mid \rx{\mj}]$ using a (classification) random forest. One major distinction, however, was that the predictors $\rx$ were not discrete as in the simulation, so we fit $\mathcal{L}(\rx)$ using the graphical lasso. This also had implications for the $\E[f_j(\rx) \mid \rx{\mj}]$ fits in PCM. We leave the details of the implementations for each method and their specific hyperparameters to Appendix~\ref{sec:data-analysis-details}, as well as an explanation of the tGCM procedure, which is similar to the oracle GCM procedure from the simulation.

\subsection{Results}
Since all five methods we considered in the simulation are inherently stochastic due to sample splitting, cross-fitting, or knockoff sampling, we report results over 25 replications. The results include number of rejections with a target FDR level of 0.1, number of rejections with a target FWER level of 0.1, and computation time (Figure \ref{fig:data-analysis-results}). Notably, in contrast to the simulation study, knockoffs produces the highest average number of rejections, though its performance is volatile, as its median number of rejections was 0. HRT made slightly more rejections than tPCM for both FDR and FWER. PCM made no rejections whatsoever, which is surprising given our theoretical result on the equivalence between tPCM and PCM.

This may be due to the discrepancy in estimating quantities like $\E[f_j(\rx) \mid \rx{\mj}]$, for which PCM uses the lasso, while tPCM and HRT fit a graphical lasso for the entire distribution of $\rx$. Finally, tGCM makes fewer rejections than HRT and tPCM. Recall that tGCM uses cross-fitting and thus does not discard any data when testing, while tPCM and HRT use just 70\% of the data for testing. That tPCM makes more rejections than tGCM despite the difference in effective sample size suggests the functional used by the former may be better suited for detecting the types of alternatives present in this particular dataset. In terms of computation, knockoffs was noticeably fastest, tPCM was slightly faster than tGCM and PCM, and HRT was the slowest.

\begin{figure}[H]
		\centering
		\includegraphics[width = \textwidth]{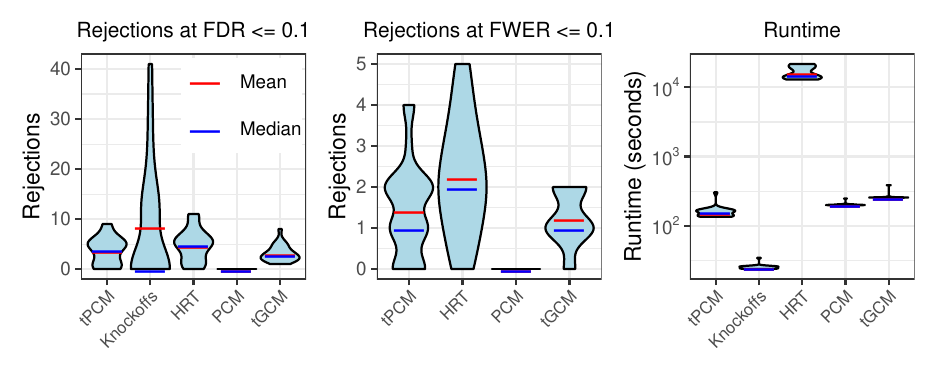}
		\caption{Results from the analysis of the breast cancer data, including rejections when controlling FDR and FWER (left and middle) and runtime (right). Variation is shown over 25 runs of each method. Knockoffs was omitted from the FWER comparison as it is not designed to control this error rate.}
		\label{fig:data-analysis-results}
	\end{figure}

\section{Discussion} \label{sec:discussion}

In this paper, we approached the variable selection problem from the dual perspectives of model-X and doubly robust methodologies, focusing on methods with power against broad classes of alternatives. We proved the equivalence of the model-X HRT and the doubly robust PCM, extending the bridge between model-X and doubly robust methodologies we established in \citet{Niu2022}. This equivalence showed the doubly robust nature of the HRT test, which had not been established before. Going beyond drawing connections between these two classes of methodologies, we borrowed ideas from both to propose the significantly faster and equally powerful tPCM test. 

The primary limitation of the tPCM test, as well as of the PCM test and HRT, is that all of these methodologies rely on sample splitting. We are not aware of any method that can achieve all four of the properties in Table~\ref{tab:comparison} without sample splitting. Unfortunately, cross-fitting cannot be used in conjunction with sample splitting to boost power in this context, since it leads to dependencies between test statistics from different folds. These dependencies can be captured and accounted for by employing the recently proposed rank-transformed subsampling method \citep{Guo2023}, though this method is computationally expensive. Sample splitting can reduce the power of these methods compared to model-X knockoffs (recall the data analysis in Section~\ref{sec:data-analysis}), which does not require sample splitting. If $p$-values for each variable are not required, knockoffs may be preferable to sample-splitting methods. We leave it to future work to explore whether there is a method that can achieve all four properties in Table~\ref{tab:comparison} without sample splitting.

\section*{Acknowledgments}

We acknowledge the Wharton research computing team for their help with our use of the Wharton high-performance computing cluster for the numerical simulations in this paper. This work was partially supported by NSF DMS-2113072 and NSF DMS-2310654.

\printbibliography


\appendix

	\section{On the estimands and power of GCM and PCM} \label{sec:estimands}

	Recalling Section~\ref{sec:existing-approaches}, GCM and PCM can be viewed as testing 
	\begin{equation}
		H_0: \psi_j = \E[(g(\rx) - \E[g(\rx) \mid \rx{\mj}])(\ey - \E[\ey \mid \rx{\mj}])] = 0,
	\end{equation}
	where $g(\rx) = \ex{j}$ for GCM and $g(\rx) = \E[\ey \mid \rx]$ for PCM. Plugging in these choices for $g$, we arrive at the expected conditional covariance (ECC) functional
	\begin{equation}
	\psi_j^{\text{ECC}}(\law) \equiv \E_{\law}[\textnormal{Cov}_{\law}[\ex{j}, \ey \mid \rx{\mj}]]
	\label{eq:expected-conditional-covariance}
	\end{equation}
	for GCM and the minimum mean squared error gap \citep{Zhang2020,Williamson2021a} 
	\begin{equation}
	\psi^{\text{mMSE}}_j(\law) \equiv \E_{\law}[(\ey - \E_{\law}[\ey \mid \rx{\mj}])^2] - \E_{\law}[(\ey - \E_{\law}[\ey \mid \rx])^2]
	\label{eq:mmse-gap}
	\end{equation}
	for PCM. Tests of $\psi_j^{\text{ECC}}(\law) = 0$ and $\psi_j^{\text{mMSE}}(\law) = 0$ are also tests of CI because both functionals vanish under CI (Figure~\ref{fig:nested-nulls}). However, the mMSE gap is nonzero for a much broader set of laws than is the ECC. In particular, any law for which $\E[\ey \mid \rx] \neq \E[\ey \mid \rx{\mj}]$ has $\psi_j^{\text{mMSE}}(\law) > 0$, while the ECC is only sensitive to departures of $\E[\ey \mid \rx]$ from $\E[\ey \mid \rx{\mj}]$ that are linear in $\ex{j}$. 
	
	In the language of nonparametrics, this manifests itself in that $\psi_j^{\text{ECC}}$ has a nonzero influence function at the null, 
	whereas $\psi_j^{\text{mMSE}}$ has a vanishing influence function under the null \citep{Williamson2021a}. This means that, locally around any null distribution $\law$, the variation in the functional $\psi_j^{\text{ECC}}$ in any direction (quantified by a pathwise derivative) is the projection of that direction onto the single direction given by the influence function. Thus, the optimal power of a test of $\psi_j^{\text{ECC}} = 0$ against any local alternative is a function of the projection of that alternative onto the influence function \citep[Lemma 25.45]{VDV1998}. This falls within our definition of a test having power against one-dimensional alternatives (Figure~\ref{fig:1d-alternatives}). By contrast, the vanishing influence function of $\psi_j^{\text{mMSE}}$ implies that to first order this functional is flat in every direction. Informally, this corresponds to a test that can, at least in principle, be sensitive to a broader set of departures, though at the cost of weaker (second-order) local signal strength.

	\section{Additional related work} \label{sec:related-work}

	Here, we expand on different strands of related work.  
	
	\paragraph{Model-X methods.} There have been several works focusing on Type-I error control for model-X methodologies without requiring the model-X assumption, in addition to those noted above on knockoffs with an in-sample estimate of $\law(\rx)$ \citep{Fan2018a, Fan2023, Fan2025}. This question has also been studied when $\law(\rx)$ is estimated out-of-sample \citep{Barber2018, Fan2020}. A conditional variant of model-X knockoffs that allows $\law(\rx)$ to follow a parametric model with unknown parameters was proposed by \citet{Huang2019}. In addition to model-X knockoffs, \citet{CetL16} also proposed the conditional randomization test (CRT) for conditional independence testing, of which the HRT is a special case. The Type-I error of the CRT when $\law(\rx)$ is estimated out-of-sample was studied by \citet{Berrett2019}. A special case of the CRT called the distilled CRT (dCRT; \cite{Liu2020}) was shown to be doubly robust by \citet{Niu2022}. Other variants of the CRT have also been proposed for their improved robustness properties \citep{Berrett2019,Li2022b,Barber2020,Zhu2023}. Other variants of the CRT have also been proposed for improved computational performance, including the HRT and several others \citep{Tansey2018,Zhong2021,Li2021c,Liu2020}. In the latter category, tests either are not suited for producing fine-grained $p$-values for each variable or require up to $O(p^2)$ resamples to get them. 
 	
	\paragraph{Doubly robust methods.}

	Another related strand of literature focuses on doubly robust testing and estimation. The GCM test \citep{Shah2018} uses a product-of-residuals statistic to test conditional independence against alternatives where the expected conditional covariance~\eqref{eq:expected-conditional-covariance} is nonzero. Minimax estimation of the expected conditional covariance has also been extensively studied; see for example \citet{Robins2008,Robins2009}. The weighted GCM test \citep{Scheidegger2022a} extends the GCM test for power against broader classes of alternatives. For sensitivity against even more general departures from the null, estimation and testing of functionals related to the mMSE gap~\eqref{eq:mmse-gap} have been considered \citep{Zhang2020,Williamson2021,Williamson2021a,Dai2022,Lundborg2022a,Hudson2023,Verdinelli2024}, including the PCM test. This functional's efficient influence function vanishes at the null, which allows power against broader alternatives but invalidates standard inferential techniques. Different methods have different approaches to mitigating this issue. However, all of these methods were designed to examine a single variable at a time, so naive application of these approaches to each of the predictor variables is computationally expensive when the number of predictors is large.

	\paragraph{Work at the intersection.}

	In a previous work \citep{Niu2022}, we established an initial bridge between the model-X and doubly robust literatures by proving the asymptotic equivalence between two conditional independence tests with power against partially linear alternatives: the dCRT \citep{Liu2020} and the GCM test \citep{Shah2018}. In this work, we strengthen this bridge by proving the asymptotic equivalence between the HRT and the PCM test, which have power against more general classes of alternatives.

	\paragraph{Other work on CI testing.} The literature on CI testing extends far beyond model-X and doubly robust methods; \citet{Pogodin2024} contains a more comprehensive review.

\section{Examples of theoretical results} \label{sec:examples-theory}

\subsection{Examples satisfying Theorem~\ref{thm:tower-pcm-type-I-error}} \label{sec:examples-type-1-error}

\paragraph{Linear Model:} For a fixed $p$ we have $(\ex{i}{j},\ey{i},\rx{i}{\mj}) \in \R \times \R \times \R^p$ for $i=1,\ldots,2n$ i.i.d samples arising out of the linear model:
\begin{align}\label{eq:linear-model}
	\begin{split}
		\ey{i} = \beta \ex{i}{j} + \rx{i}{\mj}^T \gamma + \epsilon_i \\
		\ex{i}{j} = \rx{i}{\mj}^T\eta + \delta_i, \rx{i}{\mj}\sim P_{\rx{\mj}} 
	\end{split}	
\end{align}
where $\epsilon_i,\delta_i \sim N(0,1)$ and $P_{\rx{\mj}}$ has bounded support, i.e. $\exists\, c_{\rx{\mj}} >0$ such that $\|\rx{\mj}\|_2 \leq c_{\rx{\mj}}$. We split our data into two halves and estimate all of the unknown parameters using the least squares estimates; this yields estimates $\widehat m(\rx)$ and $\lawhat_{\ex{j}|\rx{\mj}}$.
\begin{lemma}\label{lemma:linear-model-type-I-error}
	For the linear model described in \eqref{eq:linear-model}, under the null (i.e. $\beta = 0$), the assumptions \eqref{eq:var-bounded}, \eqref{eq:uniform-bound-div}-\eqref{eq:CLT-condition} hold true. Therefore, by Theorem~\ref{thm:tower-pcm-type-I-error} we conclude that $\phi_j^{\textnormal{tPCM}}$ is an asymptotically level $\alpha$ test.
\end{lemma}

Next we consider a non-parametric example which we borrow from \citet{Lundborg2022a}, namely spline estimators, which are used to fit a nonlinear model for $\ey$ on $\rx$. Our primary interest is to identify conditions under which our test is asymptotically valid in a nonparametric setting.

\paragraph{Spline Estimators and Basis Functions:}
We assume that $\rx \equiv (\ex{j},\rx{\mj}) \in [0,1] \times [0,1]^{p-1}$.
Let $\cS^{p-1}_{r,N}$ denote the space of $p-1$-tensor splines on $[0,1]^{p-1}$, where $N$ represents the number of equi-spaced interior knots in each dimension and $r$ is the order of the splines. We denote the $(p-1)$-tensor B-spline basis for $\cS^{p-1}_{r,N}$ as $\boldsymbol{\phi}^{\rx{\mj}}$, which consists of $K_{\rx{\mj}} := (N+r)^{p-1}$ basis functions. Similarly, writing $\cS^{1}_{r,N}$ for the corresponding spline space on $[0, 1]$ with $1$-tensor B-spline basis $\boldsymbol{\phi}^{\ex{j}}$, having $K_{\ex{j}} := (N + r)$ basis functions, we define the $p$-tensor product basis:
$\boldsymbol{\phi}(\rxs) := \boldsymbol{\phi}^{\ex{j}}(\exs{j}) \otimes \boldsymbol{\phi}^{\rx{\mj}}(\rxs{\mj})$ for $\cS^{p}_{r,N}$, where $u \otimes v := \text{vec}(uv^T)$, resulting in $K_{\rx}:= K_{\ex{j}} \times K_{\rx{\mj}}$ basis functions.

Let us denote the estimate of $\E(\ey|\rx)$ by $\hat m(\rx)$, which is obtained by an Ordinary Least Squares (OLS) regression of $\ey$ on the spline basis $\boldsymbol{\phi}(\rx)$ on the second half of the data ($D_2$). Hence $\hat m(\rxs) = \boldsymbol{\phi}(\rxs)^T \hat \beta_{\rx}$ where

\begin{equation}\label{eq:def-Sigma-hat}
 \hat \beta_{\rx} = \hat{\Sigma}_{\rx}^{-1} \left(\frac{1}{n}\sum_{i=n+1}^{2n} \ey{i} \boldsymbol{\phi}(\rx{i}{\bullet})\right) \text { and }	\hat{\Sigma}_{\rx} =\frac{1}{n}\sum_{i=n+1}^{2n} \boldsymbol{\phi}(\rx{i}{\bullet}) \boldsymbol{\phi}(\rx{i}{\bullet})^T.
\end{equation}
We assume we have some consistent estimate on the distribution of $\law_{\ex{j} \mid \rx{\mj}}$ which we denote by $\lawhat_{\ex{j} \mid \rx{\mj}}$ (obtained using $D_2$). 

In order to state our Type I error control result for spline regressions, it will be convenient to define the projection $\boldsymbol{\Pi}: \mathbb{R}^{K_{\rx}} \rightarrow \mathbb{R}^{K_{\rx}}$ by $\boldsymbol{\Pi}(\boldsymbol{u}) \equiv \boldsymbol{\Pi}\left(u_1, \ldots, u_{K_{X }}\right):=\boldsymbol{u}-\mathbf{1} \otimes \overline{\boldsymbol{u}}$, with $\overline{\boldsymbol{u}}=\left(\bar{u}_1, \ldots, \bar{u}_{K_{\rx{-\mj}}}\right)$ given by $\bar{u}_k:=K_{\ex{j}}^{-1} \sum_{\ell=1}^{K_{\ex{j}}} u_{(k-1) K_{\ex{j}}+\ell}$ for $k \in\left[K_{\rx{\mj}}\right]$. We denote the projection matrix corresponding to this to be $\Pi$, this projection removes the mean from each block of size $K_{X_j}$ effectively centering the vector $u$ within each group.
 
\subsubsection*{Assumptions} 
We make the following assumptions for our theoretical results:

\begin{enumerate}
	\item \textbf{Approximation Error:} Define \( m^+(\rxs) := \boldsymbol{\phi}(\rxs)^\top \beta_{\rx} \), where $\beta_{\rx} $ is the population-level best fit OLS coefficient vector. Then the approximant $m^+$ of $m$ in $\cS^p_{r,N}$ satisfies
	\begin{equation}\label{ass:bias-control}
		\|m^+ -m\|_\infty \leq K_{\rx}^{-s/p}.
	\end{equation}
	
	\item \textbf{Density Bounds of $\rx{\mj}$:} The density of $\rx{\mj}$, $p(\rxs{\mj})$, is absolutely continuous with respect to Lebesgue measure on $[0,1]^p$ and satisfies
	\begin{equation}\label{ass:strong-density-assumption}
		\sup_{\rxs{\mj} \in [0,1]^{p-1}} p(\rxs{\mj}) :=C < \infty \quad \text{ and } \quad \inf_{\rxs{\mj} \in [0,1]^{p-1}} p(\rxs{\mj}) :=c > 0.
	\end{equation}
	
	\item \textbf{Moment Conditions for Error Term:} There exist constants $c,C>0$ such that
	\begin{equation}\label{ass:bounded-moment-eps}
		\E(\varepsilon^2\mid \rx{\mj}) \geq c \quad \text{ and } \quad \E(|\varepsilon|^{2+\delta} \mid \rx{\mj}) \leq C.
	\end{equation}	
	
	\item \textbf{Restricted Eigenvalue Condition:} The matrix $\Lambda = \E[\mathrm{Cov}(\boldsymbol{\phi}(\rx)\mid \rx{\mj})]$ satisfies
	\begin{equation}\label{ass:restricted-eigen-value-lower-bound}
		\tilde \lambda_{\min}(\Lambda) := \min_{\rxs \in \mathbb{R}^{K_{\rx}}: \Pi \rxs = \rxs, \|\rxs\|_2 =1} \rxs^T \Lambda \rxs \geq \frac{c}{K_{\rx}}
	\end{equation}
	for some $c \in(0,1]$.
	\item \textbf{Non-degenerate statistic with high probability:}  It is assumed that 
	\begin{equation}\label{ass:non-degenerate}
		\P(\|\Pi \hat{\beta}_{\rx}\|_\infty = 0) = o(1).
	\end{equation}
	
	\item \textbf{Convergence Rate for Divergence:}
	\begin{equation}\label{ass:rate-cond-chi-square-div}
		\frac{1}{n} \sum \chi^2\left(\widehat\law_{\srx_j|\srx_{-j}}, \law_{\ex{j}|\rx{\mj}} | D_2\right) = o_p(n^{-\frac{2p}{2s +p}}) .
	\end{equation}

\end{enumerate}

\paragraph{Discussion of Assumptions:}
The assumptions presented are standard in nonparametric regression and share similarities with those found in \cite{Lundborg2022a}, specifically assumptions \ref{ass:bias-control}-\ref{ass:non-degenerate} are exactly the same as theirs. Assumption \ref{ass:bias-control} quantifies the approximation properties of the spline basis, which typically hold for functions exhibiting sufficient smoothness (e.g., Hölder smooth functions, as discussed in Lemma 38 of \cite{Lundborg2022a}). Assumptions \ref{ass:strong-density-assumption} and \ref{ass:bounded-moment-eps} impose regularity conditions on the data generating process, ensuring well-behaved covariates and error terms, which are fundamental for establishing asymptotic results.
Conditions \ref{ass:restricted-eigen-value-lower-bound} is a structural assumptions on the design matrix. These are often related to restricted strong convexity properties and are crucial for ensuring the stability and consistency of the OLS estimator, especially in high-dimensional or nonparametric settings where the number of basis functions $K_{X}$ can be large. Assumption \ref{ass:non-degenerate} ensures that the test statistic is non-degenerate.  Finally, Assumption \ref{ass:rate-cond-chi-square-div} specifies a required convergence rate for the estimation of conditional distributions, which is vital for controlling the Type-I error of our proposed test procedure.

\begin{theorem}\label{thm:spline-type-I-error}
	Let $K_{\rx} = n^{\frac{p}{2s+p}}$, where $s$ is the smoothness parameter from Assumption~\ref{ass:bias-control}. We also assume $s/p > \max(1/\delta, 1/2)$, where $\delta$ is from Assumption~\ref{ass:bounded-moment-eps}.
	Under Assumptions~\eqref{ass:bias-control}-\eqref{ass:rate-cond-chi-square-div}, and conditions~\eqref{eq:uniform-bound-div} and~\eqref{eq:CLT-condition}, tPCM controls type-I error.
\end{theorem}

The chosen rate for $K_{\rx}$ results from a standard bias-variance trade-off in nonparametric estimation.
Importantly, the rate condition for the chi-square divergence in Assumption~\ref{ass:rate-cond-chi-square-div} is strictly slower than the parametric rate $O_p(n^{-1})$ under the assumptions of Theorem~\ref{thm:spline-type-I-error}.

\subsection{Examples satisfying Theorem~\ref{thm:HRT-PCM-equivalence}} \label{sec:examples-HRT-PCM-equivalence}

We claim that the assumptions of Theorem~\ref{thm:HRT-PCM-equivalence} are satisfied in the linear model~\eqref{eq:linear-model}.
\begin{lemma}\label{lemma:linear-model-equivalence-to-HRT}
	For the linear model described in \eqref{eq:linear-model} with $\beta = 0$, the assumptions of Theorem \ref{thm:HRT-PCM-equivalence} are satisfied, which implies that $\phi_j^{\textnormal{tPCM}}$ is an asymptotically equivalent to $\phi_j^{\textnormal{HRT}}$.
\end{lemma}

Next, we turn our attention to the splines example. For ease of analysis, we restrict ourselves to the Model-X setting, i.e., we assume that the distribution \( \law_{\rx} \) is known, so that we may plug in \( \lawhat_{\ex{j} \mid \rx{\mj}} = \law_{\ex{j} \mid \rx{\mj}} \).
\begin{lemma}\label{lemma:spline-equiv}
	For  the spline model described in section \ref{sec:examples-type-1-error}, under the Model-X setting and the assumptions of Theorem~\ref{thm:spline-type-I-error}, if the fitted spline coefficients satisfy
	\[
	\|\Pi(\widehat \beta_{\rx})\|_\infty^2 = o_p\left(K^{-1}_{\bm X}\right),
	\]
	where \( \widehat \beta_{\exa} \) denotes the fitted spline coefficients and \( \Pi \) is the spline basis matrix introduced in Section~\ref{sec:examples-type-1-error},
	then the HRT and tPCM tests are asymptotically equivalent.
\end{lemma}

\section{Proofs} \label{sec:proofs}
Since all of our theoretical results focus on the hypothesis level, where the $j$th hypothesis to be tested is defined in~\eqref{eq:conditional-independence}, and since $j$ is fixed for the given hypothesis test, we will simplify our notation for clarity. We denote $\exa = \ex{j}$ (the $j$th predictor) and $\rz = \rx{\mj}$ (all other predictors), and in this notation, we are interested in testing the hypothesis:
\begin{equation}
	H_0: \exa \independent \ey \mid \rz
\end{equation}
In addition, we drop the $j$ subscripts from all quantities. For functions, instead of $m_j(\rx{\mj})$, we use $m(\rz)$ to denote $\E [\ey \mid \rz]$, instead of $\widehat m_j(\rx{\mj})$, we use $\widehat m(\rz)$, and instead of $\widehat f_j(\rx)$, we use $\widehat f(\exa, \rz)$. We replace $L_{ij}$ and $R_{ij}$ with $L_i$ and $R_i$. Moreover, instead of indexing tests and test statistics by $j$, we index by $n$. We will be using this notation in all of the subsequent sections.

We also define concretely here certain notions of conditional convergence. The first definition is about conditional convergence in distribution.

\begin{definition}
	\label{defn:d-p-convg}
	 For each $n$, let $W_n$ be a random variable and let $\mathcal{F}_n$ be a $\sigma$-algebra. Then, we say $W_n$ converges in distribution to a random variable $W$ conditionally on $\mathcal{F}_n$ if
	$\mathbb{P}\left[W_n \leq t \mid \mathcal{F}_n\right] \xrightarrow{p} \mathbb{P}[W \leq t]$ for each $t \in \mathbb{R}$ at which $t \mapsto \mathbb{P}[W \leq t]$ is continuous.
	We denote this relation via $W_n \mid \mathcal{F}_n \xrightarrow{d, p} W$.
\end{definition}

The next definition is about conditional convergence in probability.

\begin{definition}
		\label{defn:p-p-convg}
	For each $n$, let $W_n$ be a random variable and let $\mathcal{F}_n$ be a $\sigma$-algebra. Then, we say $W_n$ converges in probability to a constant $c$ conditionally on $\mathcal{F}_n$ if $W_n$ converges in distribution to the delta mass at $c$ conditionally on $\mathcal{F}_n$ (recall Definition 1). We denote this convergence by $W_n \mid \mathcal{F}_n \xrightarrow{p, p} c$. In symbols,
	$$
	W_n \mid \mathcal{F}_n \xrightarrow{p, p} c \text { if } \quad W_n \mid \mathcal{F}_n \xrightarrow{d, p} \delta_c \text {. }
	$$
\end{definition}
\subsection{Proof of results in Section \ref{sec:doubly robust-tower-PCM}}
\subsubsection{Auxiliary Lemmas}
\begin{lemma}[Lemma S8 from \citep{Lundborg2022a}]\label{lemma:S8}
	Let $\left(X_{n, i}\right)_{n \in \mathbb{N}, i \in[n]}$ be a triangular array of real-valued random variables and let $\left(\mathcal{F}_n\right)_{n \in \mathbb{N}}$ be a filtration on $\mathcal{F}$. Assume that
	\begin{itemize}
		\item[(i)] $X_{n, 1}, \ldots, X_{n, n}$ are conditionally independent given $\mathcal{F}_n$, for each $n \in \mathbb{N}$;
		\item[(ii)] $\mathbb{E}_P\left(X_{n, i} \mid \mathcal{F}_n\right)=0$ for all $n \in \mathbb{N}, i \in[n]$;
		\item[(iii)] $\left|n^{-1} \sum_{i=1}^n \mathbb{E}_P\left(X_{n, i}^2 \mid \mathcal{F}_n\right)-1\right|=o_{\mathcal{P}}(1)$;
		\item[(iv)] there exists $\delta>0$ such that
		$$
		\frac{1}{n} \sum_{i=1}^n \mathbb{E}_P\left(\left|X_{n, i}\right|^{2+\delta} \mid \mathcal{F}_n\right)=o_{\mathcal{P}}\left(n^{\delta / 2}\right)
		$$
	\end{itemize}
	Then $S_n \equiv n^{-1 / 2} \sum_{m=1}^n X_{n, m}$ converges uniformly in distribution to $N(0,1)$, i.e.
	$$
	\lim _{n \rightarrow \infty} \sup _{P \in \mathcal{P}} \sup _{x \in \mathbb{R}}\left|\mathbb{P}_P\left(S_n \leq x\right)-\Phi(x)\right|=0
	$$
\end{lemma}
\begin{lemma}[ Lemma S9 from \citep{Lundborg2022a}] \label{lemma:S9}
	Let $\left(X_{n, i}\right)_{n \in \mathbb{N}, i \in[n]}$ be a triangular array of real-valued random variables and let $\left(\mathcal{F}_n\right)_{n \in \mathbb{N}}$ be a filtration on $\mathcal{F}$. Assume that
	\begin{itemize}
		\item[(i)]$X_{n, 1}, \ldots, X_{n, n}$ are conditionally independent given $\mathcal{F}_n$ for all $n \in \mathbb{N}$;
		\item[(ii)] there exists $\delta \in(0,1]$ such that
		$$
		\sum_{i=1}^n \mathbb{E}_P\left(\left|X_{n, i}\right|^{1+\delta} \mid \mathcal{F}_n\right)=o_{\mathcal{P}}\left(n^{1+\delta}\right) .
		$$  	
	\end{itemize}
	Then $S_n \equiv n^{-1} \sum_{i=1}^n X_{n, i}$ and $\mu_{P, n} \equiv n^{-1} \sum_{i=1}^n \mathbb{E}_P\left(X_{n, i} \mid \mathcal{F}_n\right)$ satisfy $\left|S_n-\mu_{P, n}\right|=o_{\mathcal{P}}(1)$; i.e., for any $\epsilon>0$
	$$
	\lim _{n \rightarrow \infty} \sup _{P \in \mathcal{P}} \mathbb{P}_P\left(\left|S_n-\mu_{P, n}\right|>\epsilon\right)=0 .
	$$
\end{lemma}	
\begin{lemma}[Lemma 2 of \citet{Niu2022}]\label{lemma:conditional-expectation-covergence-to-unconditional-probability-convergence}
	Let $W_n$ be a sequence of nonnegative random variables and let $\mathcal{F}_n$ be a sequence of $\sigma$-algebras. If $\mathbb{E}\left[W_n \mid \mathcal{F}_n\right] \stackrel{p}{\rightarrow} 0$, then $W_n \stackrel{p}{\rightarrow} 0$.
\end{lemma}
\begin{lemma}[Asymptotic equivalence of tests]\label{lemma:asymptotic-equivalence-of-tests}
	Consider two hypothesis tests based on the same test statistic $T_n(\cxa, \cy, \mz)$ but different critical values:
	$$
	\phi_n^1(\cxa, \cy, \mz) \equiv \mathbbm{1}\left(T_n(\cxa, \cy, \mz)>C_n(\cxa, \cy, \mz)\right) ; \quad \phi_n^2(\cxa, \cy, \mz) \equiv \mathbbm{1}\left(T_n(\cxa, \cy, \mz)>z_{1-\alpha}\right) .
	$$	
	If the critical value of the first converges in probability to that of the second:
	$$
	C_n(\cxa, \cy, \mz) \stackrel{p}{\rightarrow} z_{1-\alpha}
	$$
	and the test statistic does not accumulate near the limiting critical value:
	\begin{equation}
	\lim _{\delta \rightarrow 0} \limsup _{n \rightarrow \infty} \mathbb{P}_{\mathcal{L}_n}\left[\left|T_n(\cxa, \cy, \mz)-z_{1-\alpha}\right| \leq \delta\right]=0,
	\label{eq:non-accumulation}
	\end{equation}
	then the two tests are asymptotically equivalent:
	$$
	\lim _{n \rightarrow \infty} \mathbb{P}_{\mathcal{L}_n}\left[\phi_n^1(\cxa, \cy, \mz)=\phi_n^2(\cxa, \cy, \mz)\right]=1 .
	$$
\end{lemma}

\begin{lemma} \label{lemma:estimation-error-to-chi-square-div}
	We have that
	$$
	(\widehat m(\rz{i}) - \E_{\law} (\widehat m(\exa{i},\rz{i}) |\rz{i},D_2) )^2 \leq \chi^2\left(\lawhat_{\exa{i}|\rz{i}}, \law_{\exa{i}|\rz{i}} | D_2\right) \E_\law[\xi_i^2|\rz{i},D_2]
	$$
	we can show that this implies
	\begin{equation*}
		(\widehat m(\rz{i}) -m(\rz{i}) )^2 \leq 2\left( 1 + \chi^2\left(\lawhat_{\exa{i}|\rz{i}}, \law_{\exa{i}|\rz{i}} | D_2\right) \right) \E_\law[ (\widehat m (\exa{i},\rz{i}) - m(\rz{i}))^2|\rz{i}, D_2]
	\end{equation*}
\end{lemma}
\begin{proof}[Proof of Lemma \ref{lemma:estimation-error-to-chi-square-div}]
	Using the variational representation of chi-squared divergence (see for example equation (7.91) in \citet{Polyanskiy2023})
	\begin{equation}\label{eq:variational-representation-of-chi-square-div}
		\chi^2(P,Q) = \sup_g \frac{(\E_P(g) - \E_Q(g))^2}{\mathrm{Var}_Q(g)} .
	\end{equation}  
	For our purposes we will condition throughout on $D_2$. Fix an $i\in[n]$ and additionally condition on $\rz{i}$, set $P = \widehat{\law}_{\exa{i} \mid \rz{i}}$ and $Q = \law_{\exa{i} \mid \rz{i}}$. Next we look at a particular $g \equiv \widehat m(\exa{i},\rz{i})$, which implies $\E_Q(g) = \E_{\law_{\exa{i} \mid \rz{i}}}[\widehat m(\exa{i}, \rz{i})\mid D_2] = \E_\law[\widehat m(\exa{i}, \rz{i})\mid \rz{i}, D_2]$ similarly $\E_P(g) = \widehat m(\rz{i})$. Observe that $\mathrm{Var}_Q(g) = \mathrm{Var}_{\law_{\exa{i} \mid \rz{i}}}(\widehat m(\exa{i}, \rz{i})\mid D_2) = \E_\law (\xi_i^2\mid \rz{i},D_2)$. We denote the conditional chi-squared divergence by $\chi^2\left(\lawhat_{\exa{i}|\rz{i}}, \law_{\exa{i}|\rz{i}} | D_2\right)$ which then implies by~\eqref{eq:variational-representation-of-chi-square-div} that
	$$
	(\widehat m(\rz{i}) - \E_{\law} (\widehat m(\exa{i},\rz{i}) |\rz{i},D_2) )^2 \leq \chi^2\left(\lawhat_{\exa{i}|\rz{i}}, \law_{\exa{i}|\rz{i}} | D_2\right) \E_\law(\xi_i^2|\rz{i},D_2),
	$$ 
	which verifies the first claim.
	
	We can bound $	(\widehat m(\rz{i}) -m(\rz{i}) )^2$ as follows:
	\begin{equation}\label{eq:decomposing-the-error-in-regression-of-Y-given-Z}
		(\widehat m(\rz{i}) -m(\rz{i}) )^2 \leq 2(\widehat m(\rz{i}) - \E_{\law} (\widehat m(\exa{i},\rz{i}) |\rz{i},D_2) )^2 + 2(\E_{\law} (\widehat m(\exa{i},\rz{i}) |\rz{i},D_2) - m(\rz{i}) )^2
	\end{equation}
	We have already upper bounded the first term. 
	
	Observe that using the fact that $(\E \exa)^2 \leq \E \exa^2$ 	we have that
	\begin{equation}\label{eq:bias-squared-is-less-than-mse}
		(\E_{\law} (\widehat m(\exa{i},\rz{i}) |\rz{i},D_2) - m(\rz{i}) )^2 \leq  \E_\law[ (\widehat m (\exa{i},\rz{i}) - m(\rz{i}))^2|\rz{i}, D_2].
	\end{equation} 
	Also using bias variance decomposition inequality we have 
	\begin{equation}\label{eq:variance-is-less-than-mse}
		\E_\law (\xi_i^2\mid \rz{i},D_2) = \mathrm{Var}_{\law}(\widehat m(\exa{i}, \rz{i})\mid \rz{i}, D_2)  \leq \E_\law[ (\widehat m (\exa{i},\rz{i}) - m(\rz{i}))^2|\rz{i}, D_2].
	\end{equation}
	Combining~\eqref{eq:bias-squared-is-less-than-mse} and~\eqref{eq:variance-is-less-than-mse} with~\eqref{eq:decomposing-the-error-in-regression-of-Y-given-Z} the result follows.
\end{proof}
\subsubsection{Proof of main results} \label{sec:proofs-of-main-results-sec-3}
The proof of the next result borrows some crucial ideas from \citet{Lundborg2022a} and builds on them.
\begin{proof}[Proof of Theorem \ref{thm:tower-pcm-type-I-error}]
	$T_n^{\textnormal{tPCM}}$ can be written as $T_N/T_D$ where $T_N = \frac{1}{\sqrt n \sigma_n} \sum_{i=1}^n R_i$ and $T_D = \widehat \sigma_n/\sigma_n$ where $\widehat \sigma_n^2 = \frac{1}{n}\sum_{i=1}^n R_i^2 - \left(\frac{1}{n} \sum_{i=1}^n R_i\right)^2$.
	We would show that $T_N \overset{d}{\to} N(0,1)$ and $T_D \overset{p}{\to} 1$. The first of these results is stated as Lemma~\ref{lem:numerator} below. The second is stated as Lemma~\ref{lem:sigma-hat-n-convergence} in Section~\ref{sec:asymp-equiv-HRT-PCM} and is proved below.
	
	The equivalence follows from the fact that $T_n - G_n = o_p(1)$ (as shown in Lemma \ref{lem:numerator}).  We next prove the uniform type-1 error control.
	
	We have already shown that for any sequence $\law_n \in\nulllaws_{n,j}\cap\mathscr{R}_n$, $\lim \sup_{n\to \infty} \E_{\law_n} [\phi_n^{\textnormal{tPCM}}] = \alpha$.
	Fix $\epsilon >0$, and for each $n$ let $\law_n^* \in \nulllaws_{n,j}\cap \mathscr{R}_n$ be such that
	$$
	\E_{\law^*_n} [\phi_n^{\textnormal{tPCM}}] \geq \sup_{\law_n \in\nulllaws_{n,j}\cap\mathscr{R}_n}\E_{\law_n} [\phi_n^{\textnormal{tPCM}}] -\epsilon
	$$ 
	Now we use the fact that $\law^*_n \in\nulllaws_{n,j}\cap\mathscr{R}_n$ to conclude that $\lim \sup_{n\to \infty} \E_{\law^*_n} [\phi_n^{\textnormal{tPCM}}] = \alpha$ which implies
	$$
	\lim \sup_{n\to \infty} \sup_{\law_n \in\nulllaws_{n,j}\cap\mathscr{R}_n}\E_{\law_n} [\phi_n^{\textnormal{tPCM}}]  \leq \alpha  + \epsilon.
	$$
	Since $\epsilon>0$ is arbitrary we have that $\lim \sup_{n\to \infty} \sup_{\law_n \in\nulllaws_{n,j}\cap\mathscr{R}_n}\E_{\law_n} [\phi_n^{\textnormal{tPCM}}]  \leq \alpha  $ i.e uniform type-I error control.
\end{proof}

\begin{lemma} \label{lem:numerator}
Under the assumptions of Theorem \ref{thm:tower-pcm-type-I-error}, we have that $T_n \overset{d}{\to} N(0,1)$.
\end{lemma}
\begin{proof}

	First we analyze $T_N$ for that we decompose $T_N$ into four terms as follows:
	\begin{align*}
		T_N &= \underbrace{\frac{1}{\sqrt n \sigma_n} \sum \varepsilon_i\xi_i}_{G_n}\
		- \underbrace{\frac{1}{\sqrt n \sigma_n} \sum \varepsilon_i(\widehat m(\rz{i}) - \E_{\law_n} [\widehat m(\exa{i},\rz{i}) |\rz{i},D_2])}_{A_n} 
		- \underbrace{\frac{1}{\sqrt n\sigma_n} \sum\xi_i(\widehat m(\rz{i}) - m(\rz{i}))}_{B_n} \\
		&+ \underbrace{\frac{1}{\sqrt n \sigma_n} \sum (\widehat m(\rz{i}) - m(\rz{i}))(\widehat m(\rz{i}) - \E_{\law_n} [\widehat m(\exa{i},\rz{i}) |\rz{i},D_2])}_{C_n}
	\end{align*}
	\paragraph{Term $G_n$}
	We use Lemma \ref{lemma:S8}, $\varepsilon_i\xi_i$ are conditionally independent given $\mathcal F_n \equiv \sigma(D_2)$. Also note that under the null conditional on $\mathcal{F}_n$, $\varepsilon_i\xi_i/\sigma_n$ are identically distributed random variables with mean zero and unit variance. Hence if we assume (assumption~\eqref{eq:CLT-condition}) that 
	$$
	\frac{1}{\sigma_n^{2+\delta}}\E_{\law_n}\left[|\varepsilon\xi|^{2+\delta} \mid D_2\right] = o_p(n^{\delta/2})
	$$
	we have that  $G_n \overset{d}{\to} N(0,1)$.
	\paragraph{Term $A_n$}
	By Lemma \ref{lemma:conditional-expectation-covergence-to-unconditional-probability-convergence} it is enough to show $\E [A_n^2\, |\, \mz,D_2] \overset{p}{\to} 0$. Using the fact  that conditionally on $ \mz,D_2$ the summands of $A_n$ are mean zero and independent we have that it is sufficient to show
	$$
	\E_{\law_n} [A_n^2 | \mz,D_2] \overset{p}{\to} 0 \iff \frac{1}{n \sigma^2_n} \sum_{i=1}^n \E_{\law_n} [\varepsilon^2_i| \rz{i},D_2](\widehat m(\rz{i}) - \E_{\law_n} [\widehat m(\exa{i},\rz{i}) |\rz{i},D_2])^2 \overset{p}{\to} 0,
	$$
	Using Lemma \ref{lemma:estimation-error-to-chi-square-div} we have that the above display is implied by
	$$
	\frac{1}{n\sigma_n^2} \sum_{i=1}^n \chi^2\left(\widehat\law_{\exa{i}|\rz{i}}, \law_{\exa{i}|\rz{i}} | D_2\right) \E_{\law_n}[ \xi_i^2 | \rz{i}, D_2] \E_{\law_n}[ \varepsilon_i^2 | \rz{i}]\overset{p}{\to} 0.
	$$
	Next we use assumption~\eqref{eq:var-bounded} to conclude that it is sufficient to have
	$$
	\frac{1}{n\sigma_n^2} \sum_{i=1}^n \chi^2\left(\widehat\law_{\exa{i}|\rz{i}}, \law_{\exa{i}|\rz{i}} | D_2\right) \E_{\law_n}[\xi_i^2 | \rz{i}, D_2] \overset{p}{\to} 0.
	$$
	which is our assumption~\eqref{eq:consistency-of-chi-square-div}.
	\paragraph{Term $B_n$}
	Again by Lemma \ref{lemma:conditional-expectation-covergence-to-unconditional-probability-convergence} it is enough to show $\E [B_n^2 | \mz,D_2] \overset{p}{\to} 0$. Using the fact  that under the null conditionally on $ \mz,D_2$ the summands of $B_n$ are mean zero and independent we have that it is sufficient to show
	\begin{equation}\label{eq:expectation-B_n-squared}
		\frac{1}{ n\sigma^2_n} \sum \E[ \xi^2_i |\rz{i}, D_2] (\widehat m(\rz{i}) - m(\rz{i}))^2 \overset{p}{\to} 0
	\end{equation}
	Using the Lemma \ref{lemma:estimation-error-to-chi-square-div} we have
	$$
	(\widehat m(\rz{i}) -m(\rz{i}) )^2 \leq 2\left( 1 + \chi^2\left(\lawhat_{\exa{i}|\rz{i}}, \law_{\exa{i}|\rz{i}} | D_2\right) \right) \E[ (\widehat m (\exa{i},\rz{i}) - m(\rz{i}))^2|\rz{i}, D_2]
	$$
	using assumption~\eqref{eq:uniform-bound-div} we have that~\eqref{eq:expectation-B_n-squared} is implied by 
	\begin{align}\label{eq:B_n-squared-term1}
		\frac{1}{n\sigma_n^2}\sum \E[ (\widehat m (\exa{i},\rz{i}) - m(\rz{i}))^2|\rz{i}, D_2] \E [\xi_i^2|\rz{i},D_2] \overset{p}{\to} 0.
	\end{align}
	which is our assumption~\eqref{eq:consitency-of-reg-func}.
	\paragraph{Term $C_n$} By Cauchy-Schwartz inequality we can upper bound $C_n$ by
	\begin{align*}
		C_n &\leq \frac{1}{\sqrt n \sigma_n} \left(\sum_{i=1}^n (\widehat m(\rz{i}) - m(\rz{i}))^2\right)^{1/2} \left(\sum_{i=1}^n(\widehat m(\rz{i}) - \E_{\law_n} [\widehat m(\exa{i},\rz{i}) |\rz{i},D_2])^2\right)^{1/2}\\
		&= \sqrt n \left(\frac{1}{n}\sum_{i=1}^n (\widehat m(\rz{i}) - m(\rz{i}))^2\right)^{1/2} \left(\frac{1}{ n \sigma^2_n}\sum_{i=1}^n(\widehat m(\rz{i}) - \E_{\law_n} [\widehat m(\exa{i},\rz{i}) |\rz{i},D_2])^2\right)^{1/2}.
	\end{align*}
	Hence it is enough to show that
	\begin{equation}\label{eq:upper-bd-C_n}
		  n \left(\frac{1}{n}\sum_{i=1}^n (\widehat m(\rz{i}) - m(\rz{i}))^2\right) \left(\frac{1}{ n \sigma^2_n}\sum_{i=1}^n(\widehat m(\rz{i}) - \E_{\law_n} [\widehat m(\exa{i},\rz{i}) |\rz{i},D_2])^2\right) = o_p(1)
	\end{equation}
	Using Lemma \ref{lemma:estimation-error-to-chi-square-div} we conclude that the above display is implied by
	\begin{align*}
		&\left(\frac{1}{n} \sum_{i=1}^n \left(1 +  \chi^2\left(\lawhat_{\exa{i}|\rz{i}}, \law_{\exa{i}|\rz{i}} | D_2\right)\right)  \E_{\law_n}[ (\widehat m(\exa{i},\rz{i}) - m(\rz{i}))^2|\rz{i},D_2]\right)	\\
		&\times\left(\frac{1}{n \sigma_n^2} \sum_{i=1}^n \chi^2\left(\lawhat_{\exa{i}|\rz{i}}, \law_{\exa{i}|\rz{i}} | D_2\right) \E_{\law_n} [\xi_i^2 |\rz{i},D_2]\right) = o_p(n^{-1})
	\end{align*}
	Under our assumption~\eqref{eq:uniform-bound-div} it is sufficient to have 
	$$
	\left(\frac{1}{n} \sum   \E_{\law_n}[ (\widehat m(\exa{i},\rz{i}) - m(\rz{i}))^2|\rz{i},D_2]\right)
	\times\left(\frac{1}{n \sigma_n^2} \sum_{i=1}^n \chi^2\left(\lawhat_{\exa{i}|\rz{i}}, \law_{\exa{i}|\rz{i}} | D_2\right) \E_{\law_n} [\xi_i^2 |\rz{i},D_2]\right) = o_p(n^{-1})
	$$
	which is our assumption~\eqref{eq:doubly robust-rate}.
	
	Combining the convergence properties of the four terms,  $T_N \convd N(0,1)$ by Slutsky's theorem. 
\end{proof}
	
\begin{proof}[Proof of Lemma~\ref{lem:sigma-hat-n-convergence}]
	
	Let us denote $u_i = m(\rz{i}) - \widehat m(\rz{i})$ and $v_i =  \E_{\law_n}( \widehat m(\exa{i},\rz{i}) |\rz{i},D_2) -  \widehat m(\rz{i})$. Then we have that $R_i = (\varepsilon_i + u_i)(\xi_i + v_i)$. We have shown that $\frac{1}{\sqrt n \sigma_n} \sum_{i=1}^n R_i \overset{d}{\to} N(0,1)$ this implies $\frac{1}{n \sigma_n} \sum_{i=1}^n R_i \overset{p}{\to} 0$. Hence it is enough to show that $\frac{1}{n \sigma^2_n} \sum_{i=1}^n R^2_i \overset{p}{\to} 1$,which would imply $T_D \overset{p}{\to} 1$. We decompose the term as
	\begin{gather*}
		\frac{1}{n \sigma^2_n} \sum_{i=1}^n R^2_i = \overbrace{\frac{1}{n\sigma_n^2} \sum_{i=1}^n \varepsilon_i^2\xi^2_i}^{S_1} +  \overbrace{\frac{1}{n\sigma_n^2} \sum_{i=1}^n v_i^2\varepsilon_i^2}^{S_2} + \overbrace{ \frac{1}{n\sigma_n^2} \sum_{i=1}^n u_i^2\xi^2_i}^{S_3} + \overbrace{ \frac{1}{n\sigma_n^2} \sum_{i=1}^n u_i^2v_i^2}^{S_4} \\
		+ 2\underbrace{\frac{1}{n\sigma_n^2} \sum_{i=1}^n v_i\varepsilon_i^2\xi_i}_{C_1} + 2\underbrace{ \frac{1}{n\sigma_n^2} \sum_{i=1}^n u_i \varepsilon_i\xi^2_i}_{C_2} + 4\underbrace{ \frac{1}{n\sigma_n^2} \sum_{i=1}^n \varepsilon_i\xi_iu_iv_i}_{C_3}	\\
		2\underbrace{\frac{1}{n\sigma_n^2} \sum_{i=1}^n u_iv^2_i\varepsilon_i}_{C_4} + 2\underbrace{ \frac{1}{n\sigma_n^2} \sum_{i=1}^n u^2_iv_i \xi_i}_{C_5} 
	\end{gather*}
	Let us look at one term at a time. We would show that all the terms except $S_1$ are $o_P(1)$ terms and $S_1 \overset{p}{\to}  1$. For showing $S_1 \overset{p}{\to} 1$ we invoke Lemma \ref{lemma:S9}.

	 Observe that $\frac{1}{\sigma^2_n} \varepsilon_i^2\xi_i^2$ is an i.i.d sequence conditional on $D_2$ which mean $1$. Hence if we assume $\sigma_n^{-{(1+\delta)}} \E\left(|\varepsilon\xi|^{1+\delta} \mid \mathcal{F}_n\right) = o_P(n^{\delta})$ (which is implied by the moment conditions needed for CLT a.k.a~\eqref{eq:CLT-condition}) then we have that $S_1$ converges to $1$ in probability.
	
	We have that $S_2,S_3 =o_P(1)$ because $\E(S_2| \mz,D_2) =\E(A_n^2| \mz,D_2)= o_p(1)$ and $\E(S_3| \mz,D_2) = \E(B_n^2| \mz,D_2)= o_p(1)$ as already shown. For $S_4$ observe that 
	$$
	S_4 \leq \frac{1}{n \sigma_n^2} \sum_{i=1}^n u_i^2 \sum_{i=1}^n v_i^2 = o_p(1)
	$$
	which is implied by~\eqref{eq:upper-bd-C_n}, which we have already proved using
	~\eqref{eq:uniform-bound-div} and~\eqref{eq:doubly robust-rate}.
	Next observe that 
	$$
	C_1 \leq \left(\frac{1}{n \sigma_n^2} \sum_{i=1}^n \varepsilon_i^2\xi_i^2 \right)^{1/2}\left(\frac{1}{n \sigma_n^2} \sum_{i=1}^n  v^2_i \varepsilon^2_i \right)^{1/2} = S_1^{1/2} S_2^{1/2} = o_p(1)
	$$
	$$
	\quad C_2 \leq S_1^{1/2} S_3^{1/2} \quad C_3 \leq S_3^{1/2}S_4^{1/2}
	$$
	$$
	C_4 \leq S_4^{1/2}S_2^{1/2} \quad C_5 \leq S_4^{1/2} S_3^{1/2}
	$$
	Since we have that $S_1 =O_p(1)$ and $S_i = o_p(1)$ for $i=1,2,3$ we have that $C_k = o_p(1)$ for $k = 1,\ldots,5$.
	
	Combining everything so far we have that $\phi^{\textnormal{tPCM}}_n$ is equivalent to the test: reject $H_0$ if
	$$
	\frac{1}{\sqrt n\sigma_n} \sum_{i=1}^n \varepsilon_i\xi_i \geq \frac{\widehat \sigma_n}{\sigma_n}z_{1-\alpha} -A_n -B_n - C_n
	$$
	Now note that the RHS converges in probability to $z_{1-\alpha}$ and the oracle test statistic converges to $N(0,1)$ (hence does not accumulate near $z_{1-\alpha}$), hence by Lemma \ref{lemma:asymptotic-equivalence-of-tests} we have that $\phi^{\textnormal{vPCM}}_n$ is equivalent to $\phi^{\textnormal{oracle}}_n$. 
\end{proof}

\begin{proof}[Proof of Lemma \ref{lemma:linear-model-type-I-error}]
	For our problem $\law(\exa|\rz) \sim N(\rz^T\eta,1)  $ and $ \lawhat(\exa|\rz) = N(\rz^T \hat \eta, 1)$. We also have that $m(\exa,\rz) = \beta \exa + \rz^T \gamma$ and $\widehat m(\exa,\rz) = \hat \beta \exa + \rz^T\hat \gamma$. Observe that $\xi_i = \widehat m(\exa{i},\rz{i}) - \E_\law (\widehat m(\exa{i},\rz{i}) | \rz{i},D_2)  = \hat\beta (\exa{i} - \E (\exa{i}|\rz{i})) = \hat \beta \delta_i $.
	
	Let us verify~\eqref{eq:var-bounded}. Observe that $\mathrm{Var}(\varepsilon_i | \rz{i}, D_2) = 1$ and hence the required condition holds.
	
	Next, we compute $\sigma^2_n$ as $\mathrm{Var}_\law[\xi_i | \rz{i}, D_2] = \hat{\beta}^2$, implying $\sigma_n^2 = \hat{\beta}^2$. We also evaluate the $\chi^2$ divergence between $\law_{\exa | \rz}$ and $\lawhat_{\exa | \rz}$ using the identity (the identity can be verified by directly evaluating the divergence):
	$$
	\chi^2(N(\mu,\sigma^2), N(\nu, \sigma^2)) = \exp\left(\frac{1}{\sigma^2}(\mu-\nu)^2\right)-1 
	$$
	which yields $\chi^2(\law_{\exa{i} | \rz{i}}, \lawhat_{\exa{i} | \rz{i}} \,|\, D_2) = \exp\left(\frac{1}{\sigma^2}[\rz{i} ^T(\hat\eta-\eta)]^2\right)-1$.
	
	We observe that $\max_{i \in [n]} |\rz{i}^T(\hat \eta - \eta) |\leq \|\rz{i}\|_2 \|\widehat\eta - \eta\|_2 \leq c_{\rz}\|\widehat\eta - \eta\|_2 \leq 1$ with high probability (since $\|\widehat\eta -\eta\|_2 \convp 0$), hence on a high probability set:
	\begin{equation}\label{eq:bound-on-chi-square-for-LM}
		\left(\exp\left(\frac{1}{\sigma^2}[\rz{i}^T(\hat\eta-\eta)]^2\right) - 1\right) \leq 2\frac{1}{\sigma^2}[\rz{i}^T(\hat\eta-\eta)]^2,
	\end{equation}
	(where we used the fact that $e^x -1 \leq 2x\,\forall\, 0\leq x\leq 1$), this allows us to   show~\eqref{eq:uniform-bound-div} which follows using the property that $\max_{i \in [n]} |\rz{i}^T(\hat \eta - \eta) |\leq 1$ with high probability.
	
	Let us look at the relevant error~\eqref{eq:consistency-of-chi-square-div} which simplifies to
	$$
	\frac{1}{n}\sum  \chi^2( \law_{\exa{i} | \rz{i}} , \lawhat_{\exa{i}|\rz{i}}) = \frac{1}{n} \sum  \left(\exp\left(\frac{1}{\sigma^2}[\rz{i}^T(\hat\eta-\eta)]^2\right) - 1\right).
	$$
	Using~\eqref{eq:bound-on-chi-square-for-LM} it implies that it is enough to show that	
	$$
	\frac{1}{n} \sum_{i=1}^n (\rz{i}^T\hat \eta - \rz{i}^T\eta)^2 \overset{p}{\to} 0
	$$
	which is clearly true because LHS is upper bounded by $c_{\rz}^2\|\widehat\eta - \eta\|_2^2 $ which goes to zero in probability at a rate $\frac{1}{n}$.
	Let us look at the estimation error~\eqref{eq:consitency-of-reg-func}
	$$
	\E_\law[(\widehat m(\exa{i},\rz{i})-m(\rz{i}))^2| \rz{i},D_2]  = \E_\law[(\exa{i}\hat \beta + \rz{i}(\hat \gamma - \gamma))^2 | \rz{i},D_2] .
	$$
	It is enough to show that
	$$
	\frac{1}{n} \sum_{i=1}^n \E_\law[(\exa{i}\hat \beta + \rz{i}(\hat \gamma - \gamma))^2 | \rz{i}] \overset{p}{\to} 0
	$$
	By further analysis, we obtain:
	\begin{align*}
		\E_\law[(\exa{i}\hat \beta + \rz{i}(\hat \gamma - \gamma))^2 | \rz{i}] &=  \E_\law[(\rz{i}^T\eta + \delta_i)\hat \beta + \rz{i}(\hat \gamma - \gamma))^2 | \rz{i}]\\
		&\leq 3 \left(\widehat\beta^2 \E [\delta_i^2|\rz{i}] + \widehat\beta^2 (\rz{i}^T\eta)^2 + (\rz{i}^T(\widehat\gamma - \gamma))^2 \right)	\\	
		&\leq 3\left(\widehat\beta^2 + \widehat\beta^2 c^2_{\rz} \|\eta\|_2^2 + c^2_{\rz} \|\widehat \gamma - \gamma\|_2^2\right) = O_p\left(\frac{1}{n}\right) 
	\end{align*}
	where have used the fact that $\sqrt n\widehat\beta = O_p(1)$ and $\sqrt n\|\widehat \gamma - \gamma\|_2 = O_p(1)$.
	The last criterion~\eqref{eq:doubly robust-rate} is product of the two rates going to zero at rate $1/n$ which  is satisfed trivially because both the rates go to zero at rate $1/n$.
\end{proof}
\subsubsection{Proof of Theorem \ref{thm:spline-type-I-error}}
\begin{proof}
	Specifically for the proof this theorem and the corresponding lemmas (pertaining to the spline example) we use $d$ to denote the dimension of the covariates instead of $p$.
	
	Let us denote the estimate of $\E(\ey|\exa,\rz)$ by $\hat m(\exa,\rz)$, which is obtained by an Ordinary Least Squares (OLS) regression of $\cy$ on the spline basis $\boldsymbol{\phi}(\cxa,\mz)$ on the dataset $D_2$. Hence,
	\[
	\hat m(x,z) = \boldsymbol{\phi}(x,z)^T \hat \beta_{XZ} \quad \text{ where } \quad \hat \beta_{XZ} = \hat{\Sigma}_{XZ}^{-1} \left(\frac{1}{n}\sum_{i=n+1}^{2n} \ey{i} \boldsymbol{\phi}(\exa{i},\rz{i})\right).
	\]
	where $\hat \Sigma_{XZ} = \frac{1}{n} \sum_{i=n+1}^{2n} \bm \phi(\exa{i},\rz{i})\bm \phi(\exa{i},\rz{i})^\top$.
	Let us verify the conditions for maintaining valid type-1 error.
	\paragraph{Verifying~\eqref{eq:CLT-condition}:} Note that $\E(|\varepsilon \xi|^{2+\delta}) = \E(\E(|\varepsilon \xi|^{2+\delta}\mid \rz,D_2)) = \E [\E(|\varepsilon|^{2+\delta} \mid Z)\E(|\xi|^{2+\delta} \mid \rz,D_2)] \leq C\E [\E(|\xi|^{2+\delta} \mid \rz,D_2)]$ where we used~\eqref{ass:bounded-moment-eps} for the last inequality. Finally using~\eqref{eq:upper-bound-xi} we have that $\E(|\varepsilon \xi|^{2+\delta}) \leq 2^{2+\delta}C \|\Pi \hat\beta_{XZ}\|_\infty^{2+\delta}\|$.
	
	Next we use~\eqref{ass:restricted-eigen-value-lower-bound} which implies~\eqref{eq:lower-bound-sigma} to obtain $\sigma_n^{2+\delta} \geq c^{1+\delta/2} K_{XZ}^{-1-\delta/2}\|\Pi\hat \beta_{XZ}\|_\infty^{2+\delta}$. Putting everything together we have
	$$
	\frac{1}{\sigma_n^{2+\delta}}\E\left(|\varepsilon\xi|^{2+\delta} \mid D_2\right) \leq \frac{2^{2+\delta}C \|\Pi \hat\beta_{XZ}\|_\infty^{2+\delta}}{c K_{XZ}^{-1-\delta/2}\|\Pi\hat \beta_{XZ}\|_2^{2+\delta}} \leq \frac{2^{2+\delta}}{c^{1+\delta/2}} K_{XZ}^{1+\delta/2} = o(n^{\delta/2})
	$$
	The last inequality follows from the fact that $K_{XZ} = n^{\frac{1}{2s/d+1}}$ and $s/d > 1/\delta$.

	\paragraph{Verifying~\eqref{eq:var-bounded}} By our assumption~\eqref{ass:bounded-moment-eps} we have that $\E(\varepsilon^2\mid \rz) \leq \E\E(\varepsilon^{2+\delta}\mid \rz)^{2/(2+\delta)} \leq C^{2/(2+\delta)}$.  Using~\eqref{eq:upper-bound-xi} we have that 
	$$
	\E (\xi^2 \mid \rz,D_2) \leq 2\|\Pi \hat \beta_{XZ}\|_\infty^2  
	$$
	Then by (36) in proof of Theorem 6 in \cite{Lundborg2022a} we have 
	$ \|\Pi \hat \beta_{XZ}\|_\infty =O_p(K_{XZ}n^{-1/2} +1) = O_p(1)$ the last equality follows from the fact that $s/d>1/2$.
	
	\paragraph{Verifying~\eqref{eq:consistency-of-chi-square-div}} Using~\eqref{eq:upper-bound-xi} and~\eqref{eq:lower-bound-sigma} we have
	\begin{align*}
		\frac{1}{n\sigma_n^2} \sum_{i=1}^n \chi^2\left(\lawhat_{\exa{i}|\rz{i}}, \law_{X|Z_i} | D_2\right) \E[ \xi_i^2 | Z_i, D_2]&\leq \frac{ 2\|\Pi \hat\beta_{XZ}\|_\infty^2}{n c K_{XZ}^{-1} \|\Pi \hat \beta\|_2^2} \sum_{i=1}^n \chi^2\left(\lawhat_{\exa{i}|\rz{i}}, \law_{\exa{i}|\rz{i}} | D_2\right)\\
		&\lesssim K_{XZ} \frac{1}{n}  \sum_{i=1}^n \chi^2\left(\lawhat_{\exa{i}|\rz{i}}, \law_{\exa{i}|\rz{i}} | D_2\right)\\
		&\lesssim n^{1/(2s/d+1)} o_p( n^{-2/(2s/d+1)}) =o_p(n^{-1/(2s/d+1)}) =o_p(1)
	\end{align*}
	where we used~\eqref{ass:rate-cond-chi-square-div} in the last line. 
	
	\paragraph{Verifying~\eqref{eq:consitency-of-reg-func}} Using assumptions~\eqref{ass:bias-control},~\eqref{ass:strong-density-assumption},~\eqref{ass:bounded-moment-eps} and~\eqref{ass:restricted-eigen-value-lower-bound} in conjunction with Lemma \ref{lemma:consistency-of-reg-func} directly gives you the result.
	
	\paragraph{Verifying~\eqref{eq:uniform-bound-div}} We just assume it to be true.
	
	\paragraph{Verifying~\eqref{eq:doubly robust-rate}} Using Lemma \ref{lemma:mse-splines-bound} and looking back at the verification of~\eqref{eq:consistency-of-chi-square-div} we have that we need he following condition to hold
	\[
	O_p\left(K_{XZ}^{-2s/d} + \frac{K_{XZ}}{n}\right)\cdot o_p(n^{-1/(2s/d+1)}) = o_p(n^{-1})
	\]
	Note that $O_p\left(K_{XZ}^{-2s/d} + \frac{K_{XZ}}{n}\right) = O_p(n^{-2s/(2s+d)})$ (by our choice of $K_{XZ}$) which implies that the LHS is given by $O_p(n^{-2s/(2s+d)})  o_p(n^{-d/(2s+d)}) = o_p(n^{-1})$.
\end{proof}

\begin{lemma}\label{lemma:mse-splines-bound}
	Assume that $K_{XZ} = O(n^{1-\omega})$ for some $\omega > 0$. Under Assumptions~\eqref{ass:bias-control},~\eqref{ass:strong-density-assumption}, and~\eqref{ass:bounded-moment-eps}, the expected mean squared error (MSE) satisfies
	\[
	\frac{1}{n} \sum_{i=1}^n \E\left((\hat{m}(\exa{i},\rz{i}) - m(\rz{i}))^2 \mid \rz{i}, D_2\right) =O_p\left(K_{XZ}^{-2s/d} + \frac{K_{XZ}}{n}\right).
	\]
\end{lemma}

\begin{proof}
	We aim to control the expected MSE, given by
	\[
	\frac{1}{n} \sum_{i=1}^n \E\left((\hat{m}(\exa{i},\rz{i}) - m(\rz{i}))^2 \mid D_2\right).
	\]
	Note that the conditioning on $\rz{i}, D_2$ in the lemma statement means we are considering the expectation conditional on the observed covariates for a specific sample. We proceed by analyzing the conditional expectation (on just $D_2$) and then discuss how it implies the lemma later.
	
	Using the inequality $(a+b)^2 \leq 2a^2 + 2b^2$, we decompose the squared error:
	\begin{align*}
		\E(\hat{m}(\exa{i},\rz{i}) - m(\rz{i}) \mid D_2)^2
		&\leq 2\E(\hat{m}(\exa{i},\rz{i}) - \bm \phi(\exa{i},\rz{i})^\top \beta_{XZ} \mid D_2)^2 \\
		&\quad + 2 \E(\bm \phi(\exa{i},\rz{i})^\top \beta_{XZ} - m(\rz{i}) \mid D_2)^2.
	\end{align*}
	Thus, the expected MSE can be bounded above by
	\begin{align*}
		& \frac{1}{n} \sum_{i=1}^n \E (\hat{m}(\exa{i},\rz{i}) - m(\rz{i}))^2 \\
		&\leq 2\frac{1}{n} \sum_{i=1}^n \E (\hat{m}(\exa{i},\rz{i}) - \bm \phi(\exa{i},\rz{i})^\top \beta_{XZ} \mid D_2)^2 \quad (\textbf{Term I}) \\
		&\quad + 2\frac{1}{n} \sum_{i=1}^n \E(\bm\phi(\exa{i},\rz{i})^\top \beta_{XZ} - m(\rz{i}))^2 \quad (\textbf{Term II}).
	\end{align*}
	
	\paragraph{Analysis of Term II:}
	Term II represents the \textit{squared bias} due to approximating $m(\rz{i})$ with the spline basis expansion $\bm \phi(\exa{i},\rz{i})^\top \beta_{XZ}$. Let $m^+(\exa{i}, \rz{i}) = \bm \phi(\exa{i},\rz{i})^\top \beta_{XZ}$ denote the best spline approximation of $m(\exa{i}, \rz{i})$.
	\begin{align*}
		\textbf{Term II} &= \frac{1}{n} \sum_{i=1}^n\E (\bm \phi(\exa{i},\rz{i})^\top \beta_{XZ} - m(\rz{i}))^2 \\
		&= \frac{1}{n} \sum_{i=1}^n \E(m^+(\exa{i},\rz{i}) - m(\rz{i}))^2.
	\end{align*}
	By Assumption~\eqref{ass:bias-control}, which presumably bounds the approximation error of the spline basis, we have
	\[
	\textbf{Term II} \leq \|m^+ -m\|_\infty^2 = O(K_{XZ}^{-2s/d}).
	\]
	
	\paragraph{Analysis of Term I:}
	Term I represents the \textit{variance component} of the estimator. We can rewrite it using the definition of $\hat \beta_{XZ}$:
	\begin{align*}
		\textbf{Term I} &= \frac{1}{n} \sum_{i=1}^n \E\left( \bm \phi(\exa{i},\rz{i})^\top (\hat{\beta}_{XZ} - \beta_{XZ}) \mid D_2 \right)^2 \\
		&= (\hat{\beta}_{XZ} - \beta_{XZ})^\top \E\left[\left( \frac{1}{n} \sum_{i=1}^n \bm \phi(\exa{i},\rz{i})\bm \phi(\exa{i},\rz{i})^\top \right)\right](\hat{\beta}_{XZ} - \beta_{XZ}) .
	\end{align*}
	Let $\hat \Sigma_{XZ} = \frac{1}{n} \sum_{i=n+1}^{2n} \bm \phi(\exa{i},\rz{i})\bm \phi(\exa{i},\rz{i})^\top$ and $\E (\hat \Sigma_{XZ}) = \Sigma_{XZ}$. From the normal equations for $\hat{\beta}_{XZ}$, we have:
	\begin{align*}
		\hat \Sigma_{XZ} (\hat{\beta}_{XZ} - \beta_{XZ})
		&= \frac{1}{n} \sum_{i=n+1}^{2n} \ey{i} \bm\phi(\exa{i},\rz{i}) - \hat \Sigma_{XZ} \beta_{XZ} \\
		&= \frac{1}{n} \sum_{i=n+1}^{2n} (\ey{i} - \bm \phi(\exa{i},\rz{i})^\top\beta_{XZ}) \bm\phi(\exa{i},\rz{i}).
	\end{align*}
	Let $h_i = m(\exa{i},\rz{i}) - m^+(\exa{i},\rz{i})$ be the approximation error, and $\varepsilon_i = \ey{i} - m(\exa{i},\rz{i})$ be the noise term. Then $\ey{i} - \bm \phi(\exa{i},\rz{i})^\top\beta_{XZ} = \ey{i} - m^+(\exa{i},\rz{i}) = (\ey{i} - m(\exa{i},\rz{i})) + (m(\exa{i},\rz{i}) - m^+(\exa{i},\rz{i})) = \varepsilon_i + h_i$.
	So,
	\[
	\hat \Sigma_{XZ} (\hat{\beta}_{XZ} - \beta_{XZ}) = \frac{1}{n} \sum_{i=n+1}^{2n} (h_i + \varepsilon_i) \bm \phi(\exa{i},\rz{i}).
	\]
	Substituting this back into the expression for Term I:
	\begin{align*}
		\textbf{Term I} &= \left( \frac{1}{n} \sum_{i=n+1}^{2n} (h_i + \varepsilon_i)\bm \phi(\exa{i},\rz{i}) \right)^\top \hat{\Sigma}_{XZ}^{-1}\Sigma_{XZ}\hat{\Sigma}_{XZ}^{-1}
		\left( \frac{1}{n} \sum_{i=n+1}^{2n} (h_i + \varepsilon_i)\bm \phi(\exa{i},\rz{i}) \right)\\
		&\leq \| \Sigma_{XZ}\|_{op} \| \hat \Sigma^{-1}_{XZ}\|_{op} \left\|\hat \Sigma_{XZ}^{-1/2}\left(\frac{1}{n} \sum_{i=n+1}^{2n} (h_i + \varepsilon_i)\bm \phi(\exa{i},\rz{i})\right)\right\|_2^2.
	\end{align*}

	Using assumption~\eqref{ass:strong-density-assumption} and the fact that $ K_{XZ} = O(n^{1-\omega})$ for $\omega >0$, along with Proposition S23 (d) from \cite{Lundborg2022a} we have that $ \| \Sigma_{XZ}\|_{op}   =O(K_{XZ}^{-1})$. Under the same conditions using Proposition S28 from \cite{Lundborg2022a} we have $\| \hat \Sigma^{-1}_{XZ}\|_{op} =O_p(K_{XZ})$. Finally using assumptions \eqref{ass:bias-control}, \eqref{ass:strong-density-assumption} and \eqref{ass:bounded-moment-eps}, along with following the proof of Proposition S29 (specifically equation (S49) and (S50) ) we have that 
   $$
   \left\|\hat \Sigma_{XZ}^{-1/2}\left(\frac{1}{n} \sum_{i=n+1}^{2n} (h_i + \varepsilon_i)\bm \phi(\exa{i},\rz{i})\right)\right\|^2_2 = O_p\left(\frac{K^{-(2s/d -1)}_{XZ}}{n} + \frac{K_{XZ}}{n}\right)
   $$
Putting everything together we have (since $s/d >1/2$):
	\[
	\textbf{Term I} = O_p\left(\frac{K_{XZ}}{n}\right).
	\]
	Combining the bounds for Term I and Term II, we get:
	\[
	\E(\text{MSE}\mid D_2) = O_p\left(K_{XZ}^{-2s/d} + \frac{K_{XZ}}{n}\right).
	\]
	The convergence in probability ($O_p$) for the overall MSE is implies by the fact that convergence in expectation implies in probability convergence.
\end{proof}

\begin{lemma}\label{lemma:consistency-of-reg-func}
	Assume~\eqref{ass:bounded-moment-eps}, ~\eqref{ass:restricted-eigen-value-lower-bound}, and all the conditions from Lemma \ref{lemma:mse-splines-bound}. Then, the following rate of convergence holds for the weighted expected mean squared error:
	\[
	\frac{1}{n \sigma_n^2} \sum_{i=1}^n \E\left[ (\hat{m}(\exa{i},\rz{i}) - m(\exa{i}, \rz{i}))^2 \mid \rz{i},D_2\right] \E(\xi_i^2\mid \rz{i},D_2) =O_p\left(K_{XZ}^{-2s/d+1} + \frac{K^2_{XZ}}{n}\right).
	\]
	Here, $\xi_i = \hat{m}(\exa{i},\rz{i}) - \E(\hat{m}(\exa{i},\rz{i}) \mid \rz{i}, \hat{m})$, and $\sigma_n^2 = \E(\varepsilon^2\xi^2\mid D_2)$.
\end{lemma}

\begin{proof}
	The estimated regression function $\hat{m}(x,z)$  by Proposition 36 of \cite{Lundborg2022a} can be written as
	\[
	\hat{m}(x,z) = \bm \phi(x,z)^\top \Pi \hat{\beta}_{XZ} + \bm \phi^Z(z)^\top \hat{\bar{\beta}}_{XZ}.
	\]
	Recall the definition of $\xi_i$ :
	\begin{align*}
		\xi_i &= \hat{m}(\exa{i},\rz{i}) - \E(\hat{m}(\exa{i},\rz{i}) \mid \rz{i}, \hat{m}) \\
		&= \bm \phi(\exa{i},\rz{i})^\top \Pi \hat{\beta}_{XZ} + \bm \phi^Z(\rz{i})^\top \hat{\bar{\beta}}_{XZ} - \E\left(\bm \phi(\exa{i},\rz{i})^\top \Pi \hat{\beta}_{XZ} + \bm \phi^Z(\rz{i})^\top \hat{\bar{\beta}}_{XZ} \mid \rz{i},\hat m\right) \\
		&= (\Pi \hat{\beta}_{XZ})^\top \left(\bm \phi(\exa{i},\rz{i}) - \E(\bm \phi(\exa{i},\rz{i}) \mid \rz{i})\right).
	\end{align*}
	This implies that
	\begin{align}\notag
		\xi_i^2 &= \left|(\Pi \hat{\beta}_{XZ})^\top \left(\bm \phi(\exa{i},\rz{i}) - \E(\bm \phi(\exa{i},\rz{i}) \mid \rz{i})\right)\right|^2 \\\notag
		&\leq \norm{\Pi \hat{\beta}_{XZ}}_\infty^2 \norm{\bm \phi(\exa{i},\rz{i}) - \E(\bm \phi(\exa{i},\rz{i}) \mid \rz{i})}_1^2 \\ \notag
		&\leq \norm{\Pi \hat{\beta}_{XZ}}_\infty^2 \left(\norm{\bm \phi(\exa{i},\rz{i})}_1 + \norm{\E(\bm\phi(\exa{i},\rz{i}) \mid \rz{i})}_1\right)^2 \\
		&= \norm{\Pi \hat{\beta}_{XZ}}_\infty^2 \left(\norm{\bm\phi(\exa{i},\rz{i})}_1 + \E(\norm{\bm\phi(\exa{i},\rz{i})}_1 \mid \rz{i})\right)^2. \label{eq:upper-bound-xi}
	\end{align}
	By Proposition 28 of \cite{Lundborg2022a}, the basis functions $\{\phi_k\}$ are non-negative and form a partition of unity. This means $\sum_k \phi_k(x,z) = 1$ for all $(x,z)$. Consequently, $\norm{\bm\phi(x,z)}_1 = \sum_k |\phi_k(x,z)| = \sum_k \phi_k(x,z) = 1$.
	Therefore, \eqref{eq:upper-bound-xi} simplifies to:
	\begin{equation}\label{eq:upper-bound-xi-simplified}
	\xi_i^2 \leq \norm{\Pi \hat{\beta}_{XZ}}_\infty^2 (1 + 1)^2 = 4\norm{\Pi \hat{\beta}_{XZ}}_\infty^2.
	\end{equation}
	
	Next, we analyze the term $\sigma_n^2 = \E (\varepsilon^2\xi^2\mid D_2)$.
	Using the fact that $\varepsilon_i$ is independent of $\xi_i$ conditional on $D_2$ (assuming $\xi_i$ depends on $\hat m$ which is determined by $D_2$), we have
	\[
	\sigma_n^2 = \E(\varepsilon^2 \mid D_2) \E(\xi^2 \mid D_2).
	\]
	From \eqref{ass:bounded-moment-eps}, we know $\E(\varepsilon_i^2 \mid \exa{i}, \rz{i}) \geq c > 0$. Therefore, $\E(\varepsilon^2 \mid D_2) \geq c$. This implies
	\begin{equation}\label{eq:lower-bound-sigma-partial}
		\sigma_n^2 \geq c \E(\xi^2 \mid D_2).
	\end{equation}
	Furthermore, we can write $\E(\xi^2 \mid D_2)$ as:
	\[
	\E(\xi^2 \mid D_2) = \E\left[ (\Pi \hat{\beta}_{XZ})^\top \left(\bm\phi(\exa,\rz) - \E(\bm \phi(\exa,\rz) \mid \rz)\right)\left(\bm\phi(\exa,\rz) - \E(\bm\phi(\exa,\rz) \mid \rz)\right)^\top (\Pi \hat{\beta}_{XZ}) \mid D_2 \right].
	\]
	Let $\Lambda = \E\left[\left(\bm\phi(\exa,\rz) - \E(\bm\phi(\exa,\rz) \mid \rz)\right)\left(\bm\phi(\exa,\rz) - \E(\bm\phi(\rz,\exa) \mid \rz)\right)^\top \mid \rz\right]$.
	Then, assuming $\Lambda$ is positive definite, we have
	\[
	\E(\xi^2 \mid D_2) = (\Pi \hat{\beta}_{XZ})^\top \E(\Lambda \mid D_2) (\Pi \hat{\beta}_{XZ}).
	\]
	By~\eqref{ass:restricted-eigen-value-lower-bound}, which states $\lambda_{\min}(\E(\Lambda \mid D_2)) \geq c K_{XZ}^{-1}$, we get:
	\begin{equation}\label{eq:lower-bound-sigma}
		\sigma_n^2 \geq c \E(\xi^2 \mid D_2) \geq c \lambda_{\min}(\E(\Lambda \mid D_2)) \norm{\Pi \hat{\beta}_{XZ}}_2^2 \geq c' K_{XZ}^{-1} \norm{\Pi \hat{\beta}_{XZ}}_2^2.
	\end{equation}
	Now, let's combine these bounds to evaluate the main expression:
	\begin{align*}
		& \frac{1}{n \sigma_n^2} \sum_{i=1}^n \E\left[(\hat{m}(\exa{i},\rz{i}) - m(\exa{i},\rz{i}))^2 \mid \rz{i},D_2\right] \mathbb{E}\left(\xi_i^2 \Big| \rz{i},D_2 \right) \\
		&\leq \frac{1}{n \cdot c' K_{XZ}^{-1} \norm{\Pi \hat{\beta}_{XZ}}_2^2} \sum_{i=1}^n \E\left[(\hat{m}(\exa{i},\rz{i}) - m(\exa{i},\rz{i}))^2 \mid \rz{i},D_2\right] \cdot 4\norm{\Pi \hat{\beta}_{XZ}}_\infty^2 \\
		&\lesssim \frac{K_{XZ}}{n} \frac{\norm{\Pi \hat{\beta}_{XZ}}_\infty^2}{\norm{\Pi \hat{\beta}_{XZ}}_2^2} \sum_{i=1}^n \E\left[(\hat{m}(\exa{i},\rz{i}) - m(\exa{i},\rz{i}))^2 \mid \rz{i},D_2\right].
	\end{align*}
	Since $\norm{\Pi \hat{\beta}_{XZ}}_\infty^2 / \norm{\Pi \hat{\beta}_{XZ}}_2^2$ is bounded by 1 (because $\|\cdot\|_\infty \leq \|\cdot\|_2$), the expression becomes:
	\begin{align*}
		&\lesssim K_{XZ} \cdot \frac{1}{n}\sum_{i=1}^n \E\left[(\hat{m}(\exa{i},\rz{i}) - m(\exa{i},\rz{i}))^2 \mid \rz{i},D_2\right] \\
		&= O_p\left(K_{XZ} \left(K_{XZ}^{-2s/d} + \frac{K_{XZ}}{n}\right)\right) \quad (\text{by Lemma \ref{lemma:mse-splines-bound}}) \\
		&= O_p\left(K_{XZ}^{-2s/d+1} + \frac{K^2_{XZ}}{n}\right).
	\end{align*}
	This completes the proof.
\end{proof}

\subsection{Proof of Results in Section \ref{sec:asymp-equivalence-vPCM-tPCM}}

\subsubsection{Proof of main results}

\begin{proof}[Proof of Theorem \ref{thm:equivalence-of-vPCM-oracle}]
	$T_n^{\textnormal{vPCM}}$ can be written as $T_N/T_D$ where $T_N = \frac{1}{\sqrt n \sigma_n} \sum_{i=1}^n L_i$ and $T_D = \widetilde \sigma_n/\sigma_n$ where $\widetilde \sigma_n^2 = \frac{1}{n}\sum_{i=1}^n L_i^2 - \left(\frac{1}{n} \sum_{i=1}^n L_i\right)^2$.
	We would show that $T_N \overset{d}{\to} N(0,1)$ and $T_D \overset{p}{\to} 1$.
	First we make a crucial observation that
	\begin{align*}
		\widehat f(\exa{i},\rz{i}) - \E_{\law}[\widehat{f} (\exa{i},\rz{i}) | \rz{i},D_2] &= \widehat m(\exa{i},\rz{i}) - \widecheck m(\rz{i}) - \E_{\law}\left[\widehat m(\exa{i},\rz{i}) - \widecheck m(\rz{i}) | \rz{i},D_2 \right]\\
		&= \widehat m(\exa{i},\rz{i}) - \E_{\law}\left[\widehat m(\exa{i},\rz{i})  | \rz{i},D_2 \right] = \xi_i
	\end{align*}
	First we analyze $T_N$ for that we decompose $T_N$ into four terms as follows:
	\begin{align*}
		T_N &= \underbrace{\frac{1}{\sqrt n \sigma_n} \sum \varepsilon_i\xi_i}_{G'_n}\
		+ \underbrace{\frac{1}{\sqrt n \sigma_n} \sum \varepsilon_i(m_{\widehat f}(\rz{i}) -\widehat m_{\widehat f}(\rz{i}))}_{P_n} 
		+\underbrace{\frac{1}{\sqrt n\sigma_n} \sum\xi_i( m(\rz{i}) - \widetilde m(\rz{i}))}_{Q_n} \\
		&+ \underbrace{\frac{1}{\sqrt n \sigma_n} \sum (m_{\widehat f}(\rz{i}) -\widehat m_{\widehat f}(\rz{i}))( m(\rz{i}) - \widetilde m(\rz{i}))}_{R_n}
	\end{align*}
	We first focus on the term $G_n'$. 
	We use Lemma S8 from \citep{Lundborg2022a}, $\varepsilon_i\xi_i$ are conditionally independent given $\mathcal F_n \equiv \sigma(D_2)$. Also note that under the null conditional on $\mathcal{F}_n$, $\varepsilon_i\xi_i/\sigma_n$ are identically distributed random variables with mean zero and unit variance. Hence if we assume (assumption~\eqref{eq:CLT-condition}) that 
	$$
	\frac{1}{\sigma_n^{2+\delta}}\E_\law\left[|\varepsilon\xi|^{2+\delta} \mid D_2\right] = o_P(n^{\delta/2})
	$$
	we have that  $G'_n \overset{d}{\to} N(0,1)$. Next we turn our attention to the term $P_n$.  Our assumption~\eqref{eq:consistency-of-reg-of-f-hat-on-Z} is equivalent to $\E_{\law}[P_n^2 | \mz,D_2] = o_p(1)$, now by using Lemma \ref{lemma:conditional-expectation-covergence-to-unconditional-probability-convergence} we have that $P_n \overset{p}{\to} 0$. Similarly for the term $Q_n$ our assumption~\eqref{eq:consistency-of-reg-of-Y-on-Z}  is equivalent to $\E_{\law}[Q_n^2 | \mz,D_2] = o_p(1)$ which again by using Lemma \ref{lemma:conditional-expectation-covergence-to-unconditional-probability-convergence} implies that $Q_n \overset{p}{\to} 0$. Finally we look at the fourth term $R_n$, by Cauchy-Schwartz inequality we can upper bound $R_n$ by
	\begin{align*}
		R_n &\leq \frac{1}{\sqrt n \sigma_n} \left(\sum_{i=1}^n (m_{\widehat f}(\rz{i}) -\widehat m_{\widehat f}(\rz{i}))^2\right)^{1/2} \left(\sum_{i=1}^n( m(\rz{i}) - \widetilde m(\rz{i}))^2\right)^{1/2}\\
		&= \sqrt n \left(\frac{1}{n}\sum_{i=1}^n (m_{\widehat f}(\rz{i}) -\widehat m_{\widehat f}(\rz{i}))^2\right)^{1/2} \left(\frac{1}{ n \sigma^2_n}\sum_{i=1}^n( m(\rz{i}) - \widetilde m(\rz{i}))^2\right)^{1/2}.
	\end{align*}
	The RHS goes to zero in probability by our assumption~\eqref{eq:doubly robust-condition-vPCM} which implies $R_n = o_p(1)$.
	
	Combining the convergence properties of the four terms,  $T_N \convd N(0,1)$ by Slutsky's theorem. Next we analyze $T_D$ and show it is $o_p(1)$.
	
	Let us denote $u_i = m(\rz{i}) - \widetilde m(\rz{i})$ and $v_i = m_{\widehat f}(\rz{i}) -  \widehat m_{\widehat f}(\rz{i})$. Then we have that $L_i = (\varepsilon_i + u_i)(\xi_i + v_i)$. We have shown that $\frac{1}{\sqrt n \sigma_n} \sum_{i=1}^n L_i \overset{d}{\to} N(0,1)$ this implies $\frac{1}{n \sigma_n} \sum_{i=1}^n L_i \overset{p}{\to} 0$. Hence it is enough to show that $\frac{1}{n \sigma^2_n} \sum_{i=1}^n L^2_i \overset{p}{\to} 1$,which would imply $T_D \overset{p}{\to} 1$. We decompose the term as
	\begin{gather*}
		\frac{1}{n \sigma^2_n} \sum_{i=1}^n R^2_i = \overbrace{\frac{1}{n\sigma_n^2} \sum_{i=1}^n \varepsilon_i^2\xi^2_i}^{S_1} +  \overbrace{\frac{1}{n\sigma_n^2} \sum_{i=1}^n v_i^2\varepsilon_i^2}^{S_2} + \overbrace{ \frac{1}{n\sigma_n^2} \sum_{i=1}^n u_i^2\xi^2_i}^{S_3} + \overbrace{ \frac{1}{n\sigma_n^2} \sum_{i=1}^n u_i^2v_i^2}^{S_4} \\
		+ 2\underbrace{\frac{1}{n\sigma_n^2} \sum_{i=1}^n v_i\varepsilon_i^2\xi_i}_{C_1} + 2\underbrace{ \frac{1}{n\sigma_n^2} \sum_{i=1}^n u_i \varepsilon_i\xi^2_i}_{C_2} + 4\underbrace{ \frac{1}{n\sigma_n^2} \sum_{i=1}^n \varepsilon_i\xi_iu_iv_i}_{C_3}	\\
		2\underbrace{\frac{1}{n\sigma_n^2} \sum_{i=1}^n u_iv^2_i\varepsilon_i}_{C_4} + 2\underbrace{ \frac{1}{n\sigma_n^2} \sum_{i=1}^n u^2_iv_i \xi_i}_{C_5} 
	\end{gather*}
	Let us look at one term at a time. We would show that all the terms except $S_1$ are $o_P(1)$ terms and $S_1 \overset{p}{\to}  1$. For showing $S_1 \overset{p}{\to} 1$ we invoke Lemma S9 from \citep{Lundborg2022a}. Observe that $\frac{1}{\sigma^2_n} \varepsilon_i^2\xi_i^2$ is an i.i.d sequence conditional on $D_2$ which mean $1$. Hence if we assume $\sigma_n^{-{(1+\delta)}} \E\left(|\varepsilon\xi|^{1+\delta} \mid \mathcal{F}_n\right) = o_P(n^{\delta})$ (which is implied by the moment conditions needed for CLT a.k.a~\eqref{eq:CLT-condition}) then we have that $S_1$ converges to $1$ in probability.
	
	We have that $S_2,S_3 =o_P(1)$ because $\E(S_2| \mz,D_2) =\E(P_n^2| \mz,D_2)= o_p(1)$ and $\E(S_3| \mz,D_2) = \E(Q_n^2| \mz,D_2)= o_p(1)$ as already shown. For $S_4$ observe that 
	$$
	S_4 \leq \frac{1}{n \sigma_n^2} \sum_{i=1}^n u_i^2 \sum_{i=1}^n v_i^2 = o_p(1)
	$$
	which is implied by~\eqref{eq:doubly robust-condition-vPCM}.
	Next observe that 
	$$
	C_1 \leq \left(\frac{1}{n \sigma_n^2} \sum_{i=1}^n \varepsilon_i^2\xi_i^2 \right)^{1/2}\left(\frac{1}{n \sigma_n^2} \sum_{i=1}^n  v^2_i \varepsilon^2_i \right)^{1/2} = S_1^{1/2} S_2^{1/2} = o_p(1)
	$$
	$$
	\quad C_2 \leq S_1^{1/2} S_3^{1/2} \quad C_3 \leq S_3^{1/2}S_4^{1/2}
	$$
	$$
	C_4 \leq S_4^{1/2}S_2^{1/2} \quad C_5 \leq S_4^{1/2} S_3^{1/2}
	$$
	Since we have that $S_1 =O_p(1)$ and $S_i = o_p(1)$ for $i=1,2,3$ we have that $C_k = o_p(1)$ for $k = 1,\ldots,5$.
	
	Combining everything so far we have that $\phi^{\textnormal{vPCM}}_n$ is equivalent to the test: reject $H_0$ if
	$$
	\frac{1}{\sqrt n\sigma_n} \sum_{i=1}^n \varepsilon_i\xi_i \geq \frac{\widetilde \sigma_n}{\sigma_n}z_{1-\alpha} -P_n -Q_n - R_n
	$$
	Now note that the RHS converges in probability to $z_{1-\alpha}$ and the oracle test statistic converges to $N(0,1)$ (hence does not accumulate near $z_{1-\alpha}$), hence by Lemma \ref{lemma:asymptotic-equivalence-of-tests} we have that $\phi^{\textnormal{vPCM}}_n$ is equivalent to $\phi^{\textnormal{oracle}}_n$. 
\end{proof}

\subsection{Proof of Results in Section	\ref{sec:asymp-equiv-HRT-PCM}}
\subsubsection{Derivations Supporting the Equivalence of HRT and tPCM}\label{sec:hrt-pcm-details}

This section contains the technical derivations and intermediate results omitted from Section~\ref{sec:asymp-equiv-HRT-PCM} of the main text. These details formally establish the connection between the HRT and tPCM statistics and verify that their cutoffs converge under suitable regularity conditions.

\paragraph{Decomposing the HRT Statistic.}
To establish the equivalence, we first observe that the HRT test statistic can be expressed as a transformation of the tPCM statistic plus remainder terms. We find that 
\begin{equation}
	\begin{split}
		T_j^{\text{HRT}} &\equiv \frac{1}{n}\sum_{i = 1}^n (\ey{i} - \widehat m(\rx{i}{\bullet}))^2 \\
		&= \frac{1}{n}\sum_{i = 1}^n ((\ey{i} - \widehat m_j(\rx{i}{\mj})) + (\widehat m_j(\rx{i}{\mj}) - \widehat m(\rx{i}{\bullet})))^2 \\
		&= -\frac{2\widehat \sigma_n}{\sqrt{n}} T_j^{\text{tPCM}} + \frac{1}{n}\sum_{i = 1}^n (\widehat m_j(\rx{i}{\mj}) - \widehat m(\rx{i}{\bullet}))^2 + \frac{1}{n}\sum_{i = 1}^n (\ey{i} - \widehat m_j(\rx{i}{\mj}))^2.
	\end{split}
\end{equation}
Letting
\begin{equation}
	\widehat \xi_i \equiv \widehat m(\rx{i}{\bullet}) - \widehat m_j(\rx{i}{\mj}) \quad \text{and} \quad \widetilde \xi_i \equiv \widehat m(\rxka{i}{\bullet}) - \widehat m_j(\rx{i}{\mj}),
\end{equation}
we have that
\begin{equation}
	\begin{split}
		T_j^{\text{HRT}} 
		&= -\frac{2\widehat \sigma_n}{\sqrt{n}} T_j^{\text{tPCM}} + \frac{1}{n}\sum_{i = 1}^n \widehat \sxi_i^2 + \frac{1}{n}\sum_{i = 1}^n (\ey{i} - \widehat m_j(\rx{i}{\mj}))^2\\
		&=-\frac{2\widehat \sigma_n}{\sqrt{n}} T_j^{\text{tPCM}} + \frac{1}{n}\sum_{i = 1}^n (\widehat \sxi_i^2 -\E[\widetilde \xi_i^2 \mid \rx{i}{\mj},  D_2] )+ \frac{1}{n}\sum_{i = 1}^n \E[\widetilde \xi_i^2 \mid \rx{i}{\mj},  D_2] \\&+ \frac{1}{n}\sum_{i = 1}^n (\ey{i} - \widehat m_j(\rx{i}{\mj}))^2.
	\end{split}
	\label{eq:hrt-decomp}
\end{equation}
Under the assumptions of Theorem~\ref{thm:HRT-PCM-equivalence}, we will show that the second term in the expression above vanishes at a rate \( o_p(n^{-1/2}) \), and is therefore a higher-order term. The last two terms do not depend on \( \cxa_j \), the variable being tested. Consequently, any additive or multiplicative transformations involving only \( \cy \) and \( \cx{\bullet}{\minus j} \) are absorbed into the null distribution generated by resampling, and do not affect the decision rule of the HRT. Since the transformation relating \( T_j^{\text{HRT}} \) to \( T_j^{\text{tPCM}} \) is monotonic and independent of \( \cxa_j \), both tests effectively compare the same core statistic against equivalent thresholds, leading to the same asymptotic rejection behavior. We make these statements precise below.

\paragraph{Equivalence of HRT and re-scaled test.}
We define a re-scaled HRT statistic:
\[
T_j^{\text{rHRT}} \equiv \frac{\widehat \sigma_n}{\sigma_n} T_j^{\text{tPCM}} - \frac{1}{2\sqrt{n} \sigma_n} \sum_{i=1}^n (\widehat \xi_i^2 - \E[\widetilde \xi_i^2 \mid \rx{i}{\mj}, D_2]).
\]
and a corresponding test based on the quantile of the re-scaled statistic:
\[
\phi_j^{\text{rHRT}}(\mx,\cy) \equiv \mathbbm{1}\left(T^{\text{rHRT}}_j(\mx,\cy) > C'_n(\cy,\mx{\bullet}{\minus j}) \right),
\]
with cutoff
\[
C'_n(\cy, \mx{\bullet}{\minus j}) \equiv \Q_{1-\alpha}\left[T^{\text{rHRT}}_j(\cxk{\bullet}{j}, \mx{\bullet}{\minus j},\cy) \mid \cy, \mx{\bullet}{\minus j}, D_2\right].
\]

\begin{lemma}[Equivalence of HRT and rHRT]
	\label{lemma:first-reduction-of-HRT}
	We have \( \phi_j^{\textnormal{HRT}}(\mx,\cy) = \phi_j^{\textnormal{rHRT}}(\mx,\cy) \).
\end{lemma}

\paragraph{Convergence of the remainder term.}
We will provide conditions under which the multiplicative factor $\frac{\widehat \sigma_n}{\sigma_n}$ tends to one, and the additive term $\frac{1}{2\sqrt n \sigma_n}\sum_{i=1}^n(\widehat \xi_i^2 - \E[\widetilde \xi_i^2 \mid \rx{i}{\mj},  D_2])$ tends to zero. This will imply that the test statistics $T_j^{\text{rHRT}}$ and $T_j^{\text{tPCM}}$ are asymptotically equivalent. The multiplicative factor $\frac{\widehat \sigma_n}{\sigma_n}$ tends to one under the assumptions of Theorem~\ref{thm:tower-pcm-type-I-error}:
\begin{lemma} \label{lem:sigma-hat-n-convergence}
	Under the assumptions of Theorem~\ref{thm:tower-pcm-type-I-error}, we have that $\frac{\widehat \sigma_n}{\sigma_n} \convp 1$.
\end{lemma}
Note that Lemma~\ref{lem:sigma-hat-n-convergence} was proved in Section~\ref{sec:proofs-of-main-results-sec-3}.
Under Assumptions~\eqref{eq:variance-of-m-hat-given-Z}--\eqref{eq:variance-of-xi-rate-assumption}, we show that the error term in \( T_j^{\text{rHRT}} \) vanishes.

\begin{lemma}\label{lemma:extra-term-in-HRT-vanishes}
	Under the assumptions listed above, we have
	\[
	\frac{1}{\sqrt n \sigma_n}\sum_{i=1}^n(\widehat \xi_i^2 - \E[\widetilde \xi_i^2 \mid \rx{i}{\mj},  D_2]) \convp 0.
	\]
\end{lemma}

\paragraph{Convergence of the resampling-based cutoff.}
We now establish conditions under which the quantile cutoff \( C_n'(\cy, \cx{\bullet}{\minus j}) \) converges to the standard normal quantile.

\begin{lemma}
	\label{thm:conditional-convergence-of-HRT-to-normality}
	Under assumptions~\eqref{eq:CLT-condition}, \eqref{eq:regression-of-Y-on-Z-is-consistent-wrt-L-hat}--\eqref{eq:moment-condition-on-product-of-residuals-hat}, we have
	\[
	C_n'(\cy, \cx{\bullet}{\minus j}) \overset{p}{\to} z_{1-\alpha}.
	\]
\end{lemma}

\paragraph{Conclusion.}
Putting together Lemmas~\ref{lemma:first-reduction-of-HRT},~\ref{lemma:extra-term-in-HRT-vanishes}, and~\ref{thm:conditional-convergence-of-HRT-to-normality}, we obtain the asymptotic equivalence of the HRT and tPCM tests, as formally stated in Theorem~\ref{thm:HRT-PCM-equivalence} of the main text.

\subsubsection{Auxiliary Theorems and Lemmas}
In this section we state a number of auxiliary lemmas and theorems which aid us in proving the main results. Many of them are borrowed from \citet{Niu2022} such as Lemma \ref{lemma:conditional-convergence-to-quantile-convergence}, \ref{lemma:conditional-jensen} and \ref{lemma:pp-to-p-convergence}, and Theorem \ref{thm:conditional-clt}, \ref{thm:conditional-wlln} and \ref{thm:unconditional-wlln}.
\begin{lemma}[Conditional convergence implies quantile convergence] \label{lemma:conditional-convergence-to-quantile-convergence}
	Let $W_n$ be a sequence of random variables and $\alpha \in(0,1)$. If $W_n \mid \mathcal{F}_n \xrightarrow{d, p} W$ for some random variable $W$ whose CDF is continuous and strictly increasing at $\mathbb{Q}_\alpha[W]$, then
	$$
	\mathbb{Q}_\alpha\left[W_n \mid \mathcal{F}_n\right] \stackrel{p}{\rightarrow} \mathbb{Q}_\alpha[W] .
	$$
\end{lemma}

\begin{lemma}[Conditional Jensen inequality]\label{lemma:conditional-jensen} Let $W$ be $a$ random variable and let $\phi$ be a convex function, such that $W$ and $\phi(W)$ are integrable. For any $\sigma$-algebra $\mathcal{F}$, we have the inequality
	$$
	\phi(\mathbb{E}[W \mid \mathcal{F}]) \leq \mathbb{E}[\phi(W) \mid \mathcal{F}] \text { almost surely. }
	$$
\end{lemma}
\begin{lemma}\label{lemma:pp-to-p-convergence}
	Let $W_n$ be a sequence of random variables and $\mathcal{F}_n$ a sequence of $\sigma$-algebras. If $W_n \mid \mathcal{F}_n \xrightarrow{p, p} 0$, then $W_n \xrightarrow{p} 0$.
\end{lemma}
\begin{theorem}[Conditional Slutsky's theorem] \label{thm:conditional-slutskys-thoerem}
	Let $W_n$ be a sequence of random variables. Suppose $a_n$ and $b_n$ are sequences of random variables such that $a_n \stackrel{p}{\rightarrow} 1$ and $b_n \stackrel{p}{\rightarrow} 0$. If $W_n \mid \mathcal{F}_n \xrightarrow{d, p} W$ for some random variable $W$ with continuous $C D F$, then
	$$
	a_n W_n+b_n \mid \mathcal{F}_n \stackrel{d, p}{\longrightarrow} W
	$$
\end{theorem}
\begin{theorem}[Conditional central limit theorem]\label{thm:conditional-clt}
	Let $W_{\text {in }}$ be a triangular array of random variables, such that for each $n, W_{\text {in }}$ are independent conditionally on $\mathcal{F}_n$. Define
	$$
	S_n^2 \equiv \sum_{i=1}^n \operatorname{Var}\left[W_{i n} \mid \mathcal{F}_n\right],
	$$
	and assume $\operatorname{Var}\left[W_{i n} \mid \mathcal{F}_n\right]<\infty$ almost surely for all $i=1, \ldots, n$ and for all $n \in \mathbb{N}$. If for some $\delta>0$ we have
	$$
	\frac{1}{S_n^{2+\delta}} \sum_{i=1}^n \mathbb{E}\left[\left|W_{i n}-\mathbb{E}\left[W_{i n} \mid \mathcal{F}_n\right]\right|^{2+\delta} \mid \mathcal{F}_n\right] \stackrel{p}{\rightarrow} 0,
	$$
	then
	$$
	\frac{1}{S_n} \sum_{i=1}^n\left(W_{i n}-\mathbb{E}\left[W_{i n} \mid \mathcal{F}_n\right]\right) \mid \mathcal{F}_n \stackrel{d, p}{\longrightarrow} N(0,1)
	$$
\end{theorem}
\begin{theorem}[Conditional law of large numbers]\label{thm:conditional-wlln} Let $W_{\text {in }}$ be a triangular array of random variables, such that $W_{i n}$ are independent conditionally on $\mathcal{F}_n$ for each $n$. If for some $\delta>0$ we have
	$$
	\frac{1}{n^{1+\delta}} \sum_{i=1}^n \mathbb{E}\left[\left|W_{i n}\right|^{1+\delta} \mid \mathcal{F}_n\right] \xrightarrow{p} 0,
	$$
	then
	$$
	\frac{1}{n} \sum_{i=1}^n\left(W_{i n}-\mathbb{E}\left[W_{i n} \mid \mathcal{F}_n\right]\right) \mid \mathcal{F}_n \xrightarrow{p, p} 0 .
	$$
	The condition is satisfied when
	$$
	\sup _{1 \leq i \leq n} \mathbb{E}\left[\left|W_{i n}\right|^{1+\delta} \mid \mathcal{F}_n\right]=o_p\left(n^\delta\right) .
	$$
\end{theorem}
\begin{theorem}[Unconditional weak law of large numbers]\label{thm:unconditional-wlln}
	Let $W_{i n}$ be a triangular array of random variables, such that $W_{\text {in }}$ are independent for each $n$. If for some $\delta>0$ we have
	$$
	\frac{1}{n^{1+\delta}} \sum_{i=1}^n \mathbb{E}\left[\left|W_{i n}\right|^{1+\delta}\right] \rightarrow 0,
	$$
	then
	$$
	\frac{1}{n} \sum_{i=1}^n\left(W_{i n}-\mathbb{E}\left[W_{i n}\right]\right) \stackrel{p}{\rightarrow} 0 .
	$$
	The condition is satisfied when
	$$
	\sup _{1 \leq i \leq n} \mathbb{E}\left[\left|W_{i n}\right|^{1+\delta}\right]=o\left(n^\delta\right) .
	$$
\end{theorem}
\begin{lemma}\label{lemma:variance-of-the-resampled-statistic-converges}
	Under assumption~\eqref{eq:CLT-condition}
	\begin{equation}\label{eq:variance-A}
		\frac{1}{n\sigma_n^2} \sum_{i=1}^n \E_\law(\varepsilon^2_i \mid \rz{i}) \E_\law(\xi_i^2 \mid \rz{i}, D_2) \overset{p}{\to} 1 
	\end{equation}
	and under assumption~\eqref{eq:moment-condition-on-product-of-residuals-hat}
	\begin{equation}\label{eq:variance-B}
		\frac{1}{n\sigma_n^2} \sum_{i=1}^n \left(\varepsilon_i^2 - \E_\law(\varepsilon_i^2 \mid \rz{i}) \right)\E(\widetilde\xi_i^2 \mid \rz{i},D_2) \overset{p}{\to} 0
	\end{equation}
	Under the previous two assumptions~\eqref{eq:CLT-condition},~\eqref{eq:moment-condition-on-product-of-residuals-hat} and additionally~\eqref{eq:variance-consistency-condition} we have that
	\begin{equation}\label{eq:resampling-variance-is-consistent}
		\frac{1}{n\sigma_n^2} \sum_{i=1}^n \varepsilon_i^2 \E( \widetilde\xi^2_i \mid \rz{i}) \overset{p}{\to} 1 
	\end{equation}
\end{lemma}
\begin{proof}
	We first show~\eqref{eq:variance-A}, let us define $W_{in} = \left(\E(\varepsilon_i^2 \mid \rz{i}) \E(\xi_i^2 \mid \rz{i},D_2) \right)/\sigma_n^2$ and $\cF_n = \sigma(D_2)$. Observe that $\E(W_{in} \mid D_2) = 1$ since we are under the null. We will use Theorem \ref{thm:conditional-wlln} for which we need to bound the moments appropriately as follows:
	\begin{align*}
		\frac{1}{n^{1 + \delta/2}} \sum_{i=1}^n \E[|W_{in}|^{1+\delta/2} \mid \cF_n] &= \frac{1}{n^{1+\delta/2}\sigma_n^{2+\delta}}\sum_{i=1}^n  \E[| \left(\E(\varepsilon_i^2 \mid \rz{i}) \E(\xi_i^2 \mid \rz{i},D_2) \right)|^{1+\delta/2} \mid D_2]\\
		&= \frac{1}{n^{1+\delta/2}\sigma_n^{2+\delta}}\sum_{i=1}^n  \E[| \left(\E(\varepsilon_i^2 \xi_i^2 \mid \rz{i},D_2) \right)|^{1+\delta/2} \mid D_2]\\
		&\leq \frac{1}{n^{1+\delta/2}\sigma_n^{2+\delta}}\sum_{i=1}^n  \E[ \E\left(|\varepsilon_i \xi_i|^{2+\delta} \mid \rz{i},D_2) \right) \mid D_2]\\
		&= \frac{1}{n^\delta \sigma_n^{2+\delta}}\E[|\varepsilon\xi|^{2+\delta}\mid D_2] = o_p(1)
	\end{align*}
	The third line in the above display follows from Lemma \ref{lemma:conditional-jensen} and the last line follows from assumption~\eqref{eq:CLT-condition}. Hence we have that
	$$
	\frac{1}{n\sigma_n^2} \sum_{i=1}^n \E_\law(\varepsilon^2_i \mid \rz{i}) \E_\law(\xi_i^2 \mid \rz{i}, D_2) \mid \cF_n \overset{p,p}{\to} 1 
	$$
	from which~\eqref{eq:variance-A} follows by applying Lemma \ref{lemma:pp-to-p-convergence}.
	
	Next we prove~\eqref{eq:variance-B}, let us define $W_{in} = \varepsilon_i^2 \E(\widetilde \xi_i^2 \mid \rz{i},D_2)/\sigma_n^2$ and $\cF_n = \sigma(\mz,D_2)$. Observe that $\E(W_{in} \mid \cF_n) = \E(\varepsilon_i^2 \mid \rz{i})\E(\widetilde \xi_i^2 \mid \rz{i},D_2) /\sigma_n^2$. We use Theorem \ref{thm:conditional-wlln}, for which we need to check some moment conditions:
	\begin{align*}
		\frac{1}{n^{1 + \delta/2}} \sum_{i=1}^n \E[|W_{in}|^{1+\delta/2} \mid \cF_n] &= \frac{1}{n^{1+\delta/2}\sigma_n^{2+\delta}}\sum_{i=1}^n  \E\left[\left| \left(\varepsilon_i^2  \E(\widetilde\xi_i^2 \mid \rz{i},D_2) \right)\right|^{1+\delta/2} \mid \mz, D_2\right]\\
		&\leq \frac{1}{n^{1+\delta/2}\sigma_n^{2+\delta}}\sum_{i=1}^n  \E\left[ |\varepsilon_i|^{2 + \delta} \left(  \E(|\widetilde\xi_i|^{2+\delta} \mid \rz{i},D_2) \right) \mid \mz, D_2\right]\\
		&\leq \frac{1}{n^{1+\delta/2}\sigma_n^{2+\delta}}\sum_{i=1}^n  \E\left[ |\varepsilon_i|^{2 + \delta}  \mid \rz{i}, D_2\right] \E(|\widetilde\xi_i|^{2+\delta} \mid \rz{i},D_2)
	\end{align*}
	which goes to zero using our assumption~\eqref{eq:moment-condition-on-product-of-residuals-hat}. Hence we have by Theorem \ref{thm:conditional-wlln}
	$$
	\frac{1}{n\sigma_n^2} \sum_{i=1}^n \left(\varepsilon_i^2 - \E_\law(\varepsilon_i^2 \mid \rz{i}) \right)\E(\widetilde\xi_i^2 \mid \rz{i},D_2) \overset{p,p}{\to} 0.
	$$
	which implies~\eqref{eq:variance-B} by Lemma \ref{lemma:pp-to-p-convergence}.
	
	Next we will show~\eqref{eq:resampling-variance-is-consistent}:
	\begin{align*}
		\frac{1}{n\sigma_n^2} \sum_{i=1}^n \varepsilon_i^2 \E( \widetilde\xi^2_i \mid \rz{i}) &= 	\frac{1}{n\sigma_n^2} \sum_{i=1}^n \left(\varepsilon_i^2 - \E_\law(\varepsilon_i^2 \mid \rz{i}) \right)\E(\widetilde\xi_i^2 \mid \rz{i},D_2) + 	\frac{1}{n\sigma_n^2} \sum_{i=1}^n \E_\law(\varepsilon^2_i \mid \rz{i}) \E_\law(\xi_i^2 \mid \rz{i}, D_2)\\
		&+\frac{1}{n\sigma_n^2} \sum_{i=1}^n \E(\varepsilon_i^2 \mid \rz{i})\left[\E(\widetilde \xi_i^2 \mid \rz{i}) - \E( \xi_i^2 \mid \rz{i}) \right] \overset{p}{\to} 1
	\end{align*}
	where we have used~\eqref{eq:variance-A},~\eqref{eq:variance-B} and~\eqref{eq:variance-consistency-condition}. 
\end{proof}
\subsubsection{Proof of the main results}

\begin{proof}[Proof of Lemma \ref{lemma:first-reduction-of-HRT}]
	Observe that 
	$$
	\frac{1}{n}\sum_{i=1}^n (\ey{i} - \widehat m(\exa{i},\rz{i}))^2 \leq C(\cy,\mz)
	$$
	is equivalent to
	$$
	\frac{1}{n} \sum_{i=1}^n (\ey{i} - \widehat m(\exa{i},\rz{i}))^2  - \sum_{i=1}^n (\ey{i} - \E_{\lawhat}[\widehat m(\exa{i},\rz{i}) \mid \rz{i},D_2])^2 \leq \widetilde C(\cy,\mz)
	$$
	where $\widetilde C(\cy,\mz)$ is the obtained by suitably updating $C(\cy,\mz)$. Now the above display can be shown equivalent to
	\begin{align*}
		& \frac{2}{n} \sum_{i=1}^n \left(\E_{\lawhat}[\widehat m(\exa{i},\rz{i}) \mid \rz{i},D_2] - \widehat m(\exa{i},\rz{i})\right)\left(\ey{i} - \frac{\widehat m(\exa{i},\rz{i}) + \E_{\lawhat }[\widehat m(\exa{i},\rz{i}) \mid \rz{i},D_2]}{2}\right) \leq \widetilde C(\cy,\mz)\\
		&\iff \frac{-2}{n} \sum_{i=1}^n(\widehat m(\exa{i},\rz{i}) - \widehat m(\rz{i}))(\ey{i} -  \widehat m(\rz{i}))  + \frac{1}{ n}\sum_{i=1}^n (\widehat m(\exa{i},\rz{i}) - \widehat m(\rz{i}))^2 \leq \widetilde C(\cy,\mz)\\
		&\iff \frac{-2}{n} \sum_{i=1}^n (\widehat m(\exa{i},\rz{i}) - \widehat m(\rz{i}))(\ey{i} -  \widehat m(\rz{i}))  + \frac{1}{ n}\sum_{i=1}^n ((\widehat m(\exa{i},\rz{i}) - \widehat m(\rz{i}))^2 - \E (\widetilde \xi_i^2 \mid \rz{i}, D_2)) \leq \widetilde{\widetilde C}(\cy,\mz)
	\end{align*} 
	We have adjusted by $\frac{1}{n}\sum_{i=1}^n  \E (\xi_i^2 \mid  \rz{i}, D_2)$ on the last line and got the modified $\widetilde{\widetilde C}(\cy,\mz)$.
	Re-scaling by $-\frac{\sqrt n}{2\sigma_n}$ we have proved the result. 
	
\end{proof}

\begin{proof}[Proof of Lemma \ref{lemma:extra-term-in-HRT-vanishes}]
We have
\begin{align*}
&\frac{1}{\sqrt{n}\sigma_n}\sum_{i = 1}^n (\widehat \xi_i^2 - \E(\widetilde \xi_i^2 \mid \rz{i}, D_2)) \\
&\quad= \frac{1}{\sqrt{n}\sigma_n}\sum_{i = 1}^n \left(\widehat \xi_i^2 - \E(\widehat \xi_i^2 \mid \rz{i}, D_2)\right) \\
&\quad \quad + \frac{1}{\sqrt{n}\sigma_n}\sum_{i = 1}^n \left(\E(\widehat \xi_i^2 \mid \rz{i}, D_2) - \E(\widetilde \xi_i^2 \mid \rz{i}, D_2)\right)\\
&\quad=\frac{1}{\sqrt{n}\sigma_n}\sum_{i = 1}^n \left(\widehat \xi_i^2 - \E(\widehat \xi_i^2 \mid \rz{i}, D_2)\right) \\
&\quad \quad + \frac{1}{\sqrt{n}\sigma_n}\sum_{i = 1}^n (\E[\widehat m(\exa{i}, \rz{i}) \mid \rz{i},D_2] - \widehat m(\rz{i}))^2 \\
&\quad \quad + \frac{1}{\sqrt{n}\sigma_n}\sum_{i = 1}^n (\E[\xi_i^2 \mid \rz{i},D_2]-\E[\widetilde \xi_i^2 \mid \rz{i},D_2]) \\
&\quad \equiv I_n + II_n + III_n.
\end{align*}
We have $I_n \convp 0$ because $\E[I_n^2 \mid \mz,D_2]$ by assumption~\eqref{eq:variance-of-m-hat-given-Z}. Furthermore, we have $II_n \convp 0$ and $III_n \convp 0$ by assumptions~\eqref{eq:tower-regression-rate-assumption} and~\eqref{eq:variance-of-xi-rate-assumption}, respectively. 
\end{proof}

\begin{proof}[Proof of Lemma~\ref{thm:conditional-convergence-of-HRT-to-normality}]
	Observe that $T^{\textnormal{rHRT}}$ can be decomposed as
	\begin{equation}\label{eq:rHRT-decomposition}
	\begin{split}
		&T^{\textnormal{rHRT}}(\widetilde \cxa,\cy,\mz) \\
		&\quad= \frac{1}{\sqrt n \sigma_n} \sum_{i=1}^n \varepsilon_i\widetilde \xi_i + \frac{1}{\sqrt n \sigma_n} \sum_{i=1}^n (m(\rz{i}) - \widehat m(\rz{i}))\widetilde\xi_i  - \frac{1}{2\sqrt n \sigma_n}\sum_{i=1}^n (\widetilde\xi^2_i - \E(\widetilde \xi_i^2 \mid \rz{i}, D_2)) \\
		&\quad= I_n + II_n + III_n,
	\end{split}
	\end{equation}
	where $m(\rz{i}) = \E(\ey{i} \mid \rz{i})$. First, we claim that $II_n,III_n \overset{p}{\to} 0$. Let us first look at $II_n$. We calculate
	$$
	\E[II_n^2 \mid \mz, D_2] = \frac{1}{n \sigma^2_n} \sum_{i=1}^n(m(\rz{i}) - \widehat m(\rz{i}))^2 \E(\widetilde \xi_i^2 \mid \rz{i},D_2)  \overset{p}{\to} 0.
	$$
	The convergence to zero follows from our assumption~\eqref{eq:regression-of-Y-on-Z-is-consistent-wrt-L-hat}, so we conclude that $II_n \overset{p}{\to} 0$ by Lemma \ref{lemma:conditional-expectation-covergence-to-unconditional-probability-convergence}. Next we look at $III_n$ and evaluate $\E( III_n^2 \mid D_2)$,which is given by
	$$
	\E (III_n^2 \mid D_2) = \frac{\E\left[\mathrm{Var}(\widetilde{\xi}^2 \mid \bm \mz, D_2)\mid D_2\right]}{4\sigma_n^2} \overset{p}{\to} 0,
	$$ 
	which goes to zero by our assumption~\eqref{eq:variance-of-xi-hat-goes-to-zero}, from which we conclude $III_n \overset{p}{\to} 0$ by Lemma \ref{lemma:conditional-expectation-covergence-to-unconditional-probability-convergence}. Now, we turn our attention to $I_n$. Let us denote by $W_{in} \equiv \varepsilon_i\widetilde\xi_i$, $\mathcal{F}_n = \sigma(\cy,\mz,D_2)$ and invoke the conditional CLT \ref{thm:conditional-clt} to obtain that
	\begin{equation}\label{eq:conditional-clt-for-product-of-residual}
		\frac{1}{\widehat S_n} \sum_{i=1}^n W_{in} \overset{d,p}{\to} N(0,1),
	\end{equation}
	where $\widehat S^2_n = \sum_{i=1}^n\mathrm{Var}(W_{in} \mid \mathcal{F}_n) = \sum_{i=1}^n \varepsilon_i^2 \E ( \widetilde\xi_i^2 \mid \rz{i},D_2) $  if
	\begin{equation}\label{eq:clt-moment-condition}
		\frac{1}{\widehat S_n^{2+\delta}} \sum_{i=1}^n \E(|W_{in}|^{2+\delta} \mid \mathcal{F}_n) \overset{p}{\to} 0.
	\end{equation}
	Now from Lemma \ref{lemma:variance-of-the-resampled-statistic-converges} we know that $\frac{\widehat S^2_n}{n\sigma_n^2} \overset{p}{\to} 1$. Using this we know that~\eqref{eq:clt-moment-condition} is equivalent to showing
	$$
	\frac{1}{n^{1+\delta/2}\sigma_n^{2+\delta}} \sum_{i=1}^n \E(|W_{in}|^{2+\delta} \mid \mathcal{F}_n) \overset{p}{\to} 0.
	$$
	The LHS above is equal to
	\begin{align*}
		\frac{1}{n^{1+\delta/2}\sigma_n^{2+\delta}} \sum_{i=1}^n \E(|W_{in}|^{2+\delta} \mid \mathcal{F}_n)  &= 	\frac{1}{n^{1+\delta/2}\sigma_n^{2+\delta}} \sum_{i=1}^n \E(|\varepsilon_i\widetilde\xi_i|^{2+\delta} \mid \ey{i},\rz{i},D_2) \\
		&= \frac{1}{n^{1+\delta/2}\sigma_n^{2+\delta}} \sum_{i=1}^n |\varepsilon_i|^{2+\delta}\E(|\widetilde\xi_i|^{2+\delta} \mid \rz{i},D_2) \\
		&\equiv IV_n.
	\end{align*}
	Our assumption~\eqref{eq:moment-condition-on-product-of-residuals-hat} implies $\E(IV_n \mid D_2) \overset{p}{\to} 0$, which by Lemma \ref{lemma:conditional-expectation-covergence-to-unconditional-probability-convergence} implies $IV_n\overset{p}{\to} 0$. Hence the condition for the conditional CLT holds. Next let us look at the statement of conditional CLT, using the fact that  $\frac{\widehat S^2_n}{n\sigma_n^2} \overset{p}{\to} 1$ we can show that~\eqref{eq:conditional-clt-for-product-of-residual} is equivalent to (by using conditional Slutsky, Theorem \ref{thm:conditional-slutskys-thoerem})
	$$
	\frac{1}{\sqrt n \sigma_n} \sum_{i=1}^n \varepsilon_i\widetilde\xi_i \mid \cy, \mz, D_2\overset{d,p}{\to} N(0,1).
	$$
	Again using conditional Slutsky (Theorem \ref{thm:conditional-slutskys-thoerem}) we have that $ T^{\textnormal{rHRT}}(\widetilde \cxa,\cy,\mz) \mid \cy,\mz, D_2\overset{d,p}{\to} N(0,1).$ This in turn implies $C_n'(\cy,\mz) \overset{p}{\to} z_{1-\alpha}$ by Lemma \ref{lemma:conditional-convergence-to-quantile-convergence}.
\end{proof}

\begin{proof}[Proof of Theorem \ref{thm:HRT-PCM-equivalence}]
By Lemma~\ref{lemma:first-reduction-of-HRT}, we have
\begin{equation}
T^{\text{HRT}}_n \geq C_n(\cy, \mz) \quad \Longleftrightarrow \quad T^{\text{rHRT}}_n \geq C_n'(\cy, \mz).
\end{equation}
By Lemmas~\ref{lem:sigma-hat-n-convergence}, \ref{lemma:extra-term-in-HRT-vanishes}, and~\ref{thm:conditional-convergence-of-HRT-to-normality}, we have
\begin{equation}
T^{\text{rHRT}}_n \geq C_n'(\cy, \mz) \quad \Longleftrightarrow \quad T^{\text{tPCM}}_n \geq C_n''(\cy, \mz), \quad \text{where} \quad C_n''(\cy, \mz) \convp z_{1-\alpha}.
\end{equation}
By Lemmas~\ref{lem:sigma-hat-n-convergence} and~\ref{lem:numerator}, we have $T^{\text{tPCM}}_n \convd N(0,1)$, so the non-accumulation condition~\eqref{eq:non-accumulation} holds. Therefore, by Lemma~\ref{lemma:asymptotic-equivalence-of-tests}, we conclude that the HRT and tPCM tests are asymptotically equivalent.
\end{proof}

\begin{proof}[Proof of Lemma \ref{lemma:linear-model-equivalence-to-HRT}]
	In Lemma \ref{lemma:linear-model-type-I-error} we have already verified all the assumptions pertaining to tPCM for the proposed linear model, so we only need to verify the assumptions \eqref{eq:variance-of-m-hat-given-Z}, \eqref{eq:tower-regression-rate-assumption}, \eqref{eq:variance-of-xi-rate-assumption}, \eqref{eq:regression-of-Y-on-Z-is-consistent-wrt-L-hat}, \eqref{eq:variance-consistency-condition}, \eqref{eq:variance-of-xi-hat-goes-to-zero}, and~\eqref{eq:moment-condition-on-product-of-residuals-hat}.
	
	Recall from the proof of Lemma \ref{lemma:linear-model-type-I-error} that $\law(\exa|\rz) \sim N(\rz^T\eta,1)  $ and $ \lawhat(\exa|\rz) = N(\rz^T \hat \eta, 1)$. We also have that $m(\exa,\rz) = \beta \exa + \rz^T \gamma$ and $\widehat m(\exa,\rz) = \hat \beta \exa + \rz^T\hat \gamma$. Observe that $\xi_i = \hat \beta \delta_i $, $\mathrm{Var}(\varepsilon_i | \rz{i}, D_2) = 1$ and  $\mathrm{Var}_\law[\xi_i | \rz{i}, D_2] = \hat{\beta}^2$, implying $\sigma_n^2 = \hat{\beta}^2$. Now note that $\widehat m(\rz{i}) = \E_{\lawhat}(\widehat m(X,Z)) = \widehat \beta(\rz{i}^T \widehat \eta) + \rz{i}^T\widehat \gamma$ and $\widetilde \xi_i = \widehat m(\exka{i},\rz{i}) - \widehat m(\rz{i}) = \widehat \beta(\exka{i} - \rz{i}^T\widehat \eta) = \widehat \beta \delta_i'$, where $\delta_i' \overset{i.i.d}{\sim} N(0,1)$. This implies that $\E(\widetilde \xi_i^2 \mid \rz{i},D_2) = \widehat\beta^2$.
  
	First, we verify equation~\eqref{eq:variance-of-m-hat-given-Z}:
\begin{align*}
	\frac{1}{\sigma_n^2} \E&\left[ \V\left( \left(\widehat m(\exa,\rz) - \widehat m(\rz)\right)^2 \mid \rz, D_2\right)\mid D_2\right] \\ &= \frac{1}{\widehat \beta^2} \E\left[ \V\left( \widehat \beta^2\left(\exa - \rz^T\widehat \eta\right)^2 \mid \rz, D_2\right)\mid D_2\right]\\
	&= \frac{\widehat \beta^4}{\widehat \beta^2}  \E\left[ \V\left( \left(\exa - \rz^T\widehat \eta\right)^2 \mid \rz, D_2\right)\mid D_2\right]\\
	&= \widehat{\beta}^2  \E\left[ \V\left(  \left(\rz^T(\eta -\widehat \eta) + \delta\right)^2 \mid \rz, D_2\right)\mid D_2\right]\\
	&\leq \widehat{\beta}^2  \E\left[ \E\left( \left(\rz^T(\eta -\widehat \eta) + \delta\right)^4 \mid \rz, D_2\right)\mid D_2\right]\\
	&\leq c_1 \widehat{\beta}^2  \E\left[ \left(\rz^T(\eta -\widehat \eta) \right)^4+ \delta^4 \mid D_2\right]\\
	&\leq c_1 \widehat{\beta}^2 \left[ \|\widehat{\eta}-\eta\|_2^4c_{\rz}^4 + \E \delta^4\right] = O_p\left(\frac{1}{n}\right)
\end{align*}
where we have used the fact that $\widehat{\beta} =O_p\left(\frac{1}{\sqrt n}\right)$ and $\|\widehat{\eta}-\eta\|_2 = o_p(1)$. 

Next, we verify equation~\eqref{eq:tower-regression-rate-assumption}:
\begin{align*}
		\frac{1}{\sqrt n\sigma_n}\sum_{i=1}^n (\widehat m(\rz{i}) - \E\left[\widehat m(\exa{i},\rz{i}) \mid \rz{i}, D_2\right])^2 &= \frac{1}{\sqrt n \widehat \beta} \sum_{i=1}^n \left(\widehat{\beta}(\rz{i}^T\widehat \eta - \rz{i}^T \eta)\right)^2\\
		&= \widehat{\beta} \frac{1}{\sqrt n} \sum_{i=1}^n \left(\rz{i}^T\widehat \eta - \rz{i}^T \eta\right)^2\\
		&\leq |\sqrt n \widehat{\beta} | |c_{\rz}^2 \|\widehat{\eta} -\eta\|_2^2|  =O_p\left(\frac{1}{n}\right),
\end{align*}
where we have used Cauchy-Schwartz inequality at the last inequality and used the fact that $\widehat{\beta} $ and $\|\widehat{\eta}-\eta\|_2 $ are $=O_p\left(\frac{1}{\sqrt n}\right)$.

Next, we note that equation~\eqref{eq:variance-of-xi-rate-assumption} follows from the fact that $\V_{\lawhat}[\sxi_i \mid \rz{i}, D_2] = \V_{\law}[\sxi_i \mid \rz{i}, D_2] = \widehat \beta^2$, which implies the LHS of~\eqref{eq:variance-of-xi-rate-assumption} is exactly $0$.

Next, we verify equation~\eqref{eq:regression-of-Y-on-Z-is-consistent-wrt-L-hat}:
  \begin{align*}
  	\frac{1}{n\sigma_n^2} \sum_{i=1}^n (m(\rz{i}) - \widehat m(\rz{i}))^2 \E(\widetilde \xi_i^2 | \rz{i}, D_2) 
  	&= 	\frac{1}{n\widehat \beta^2} \sum_{i=1}^n \left[\rz{i}^T\gamma - (\widehat \beta (\rz{i}^T \widehat \eta) + \rz{i}^T\widehat \gamma)\right]^2\widehat \beta^2 \\
  	&= \frac{1}{n} \sum_{i=1}^n \left[\widehat \beta(\rz{i}^T\widehat\eta) + \rz{i} (\widehat \gamma -\gamma)\right]^2\\
  	&\leq \widehat \beta^2 \frac{1}{n} \sum_{i=1}^n (\rz{i}^T\widehat \eta)^2 + \frac{1}{n} \sum_{i=1}^n (\rz{i}^T(\widehat \gamma - \gamma))^2\\
  	&\leq \widehat \beta^2 c_{\rz}^2\|\widehat{\eta}\|_2^2 + \|\widehat{\gamma} -\gamma\|_2^2 c^2_{\rz} =O_p\left(\frac{1}{n}\right)
  \end{align*}
The last line follows from the fact that $\widehat \beta^2$, $ \|\widehat{\gamma} -\gamma\|_2^2$ and $\|\widehat{\eta} -\eta\|_2^2$  are $O_p\left(\frac{1}{n}\right)$.

To show~\eqref{eq:variance-consistency-condition}, we observe that $\V_{\lawhat}[\sxi_i \mid \rz{i}, D_2] = \V_{\law}[\sxi_i \mid \rz{i}, D_2] = \widehat \beta^2$ from which it follows that the the LHS is exactly zero, and hence~\eqref{eq:variance-consistency-condition} is trivially true.

Next, we verify equation~\eqref{eq:variance-of-xi-hat-goes-to-zero}:
\begin{align*}
	\frac{\E\left[\V(\widetilde{\pxi}^2 \mid \rz, D_2)\mid D_2\right]}{\sigma_n^2} = \frac{\widehat{\beta}^4}{\widehat{\beta}^2} \E(\V(\delta'^2)) =\widehat{\beta}^2 = O_p\left(\frac{1}{n}\right)
\end{align*}

Finally, we verify equation~\eqref{eq:moment-condition-on-product-of-residuals-hat}:
\begin{align*}
		\frac{1}{\sigma_n^{2+\delta}}\E( |\peps\widetilde{\pxi}|^{2+\delta} \mid D_2) = \frac{1}{\widehat \beta^{2 + \delta}} \widehat\beta^{2+\delta}  \E(|\peps|^{2+\delta}) \E(|\delta|^{2+\delta}) = O_p(1) = O_p(n^\delta)
\end{align*}

\end{proof}

\begin{proof}[Proof of Lemma \ref{lemma:spline-equiv}]
	We analyze the spline model under the Model-X assumption, i.e., the covariate distribution \( \law_{\exa} \) is known and we may plug in \( \lawhat_{\exa_j \mid \rx{\mj}} = \law_{\exa_j \mid \rx{\mj}} \).
	
	This directly ensures that the following conditions are satisfied:
	\begin{itemize}
		\item~\eqref{eq:tower-regression-rate-assumption} and~\eqref{eq:variance-of-xi-rate-assumption} hold because the law \( \lawhat \) is correctly specified;
		\item~\eqref{eq:variance-consistency-condition} holds as the variance under \( \lawhat \) matches the true law.
	\end{itemize}
	
	Moreover, under this setting we have:
	\[
	\widehat \sxi_i = \sxi_i \quad \text{and} \quad \widetilde \sxi_i \mid \rx{i}{\mj} \overset{d}{=} \sxi_i \mid \rx{i}{\mj},
	\]
	which implies:
	\[
	\V(\widehat \sxi_i^2 \mid \rx{i}{\mj}, D_2) = \V(\sxi_i^2 \mid \rx{i}{\mj}, D_2) = \V(\widetilde \sxi_i^2 \mid \rx{i}{\mj}, D_2),
	\]
	so that conditions~\eqref{eq:variance-of-m-hat-given-Z} and~\eqref{eq:variance-of-xi-hat-goes-to-zero} both reduce to:
	\begin{equation}
		\label{eq:variance-of-sxi-goes-to-zero-supp}
		\frac{1}{\sigma_n^2} \E\left[ \V(\sxi^2 \mid \cx{\bullet}{\minus j}, D_2) \mid D_2 \right] \overset{p}{\to} 0.
	\end{equation}
	
	For condition~\eqref{eq:moment-condition-on-product-of-residuals-hat}, note that under the null hypothesis:
	\[
	\varepsilon_i \sxi_i \overset{d}{=} \varepsilon_i \widetilde \sxi_i,
	\]
	so this condition simplifies to the existing central limit condition~\eqref{eq:CLT-condition} (which was already verified to hold under the assumptions of Theorem \ref{thm:spline-type-I-error}).
	
	For the remaining assumption~\eqref{eq:regression-of-Y-on-Z-is-consistent-wrt-L-hat}, observe that:
	\begin{align*}
		(\widehat m(\rx{i}{\mj}) - m(\rx{i}{\mj}))^2 &= \left( \E[\widehat m(\exa{i}, \rx{i}{\mj}) \mid \rx{i}{\mj}, D_2] - m(\rx{i}{\mj}) \right)^2 \\
		&\leq \E\left[ (\widehat m(\exa{i}, \rx{i}{\mj}) - m(\rx{i}{\mj}))^2 \mid \rx{i}{\mj}, D_2 \right],
	\end{align*}
	and since \( \E[\widetilde \sxi_i^2 \mid \rx{i}{\mj}, D_2] = \E[\sxi_i^2 \mid \rx{i}{\mj}, D_2] \), condition~\eqref{eq:regression-of-Y-on-Z-is-consistent-wrt-L-hat} is implied by:
	\begin{equation}
		\frac{1}{n \sigma_n^2} \sum_{i=1}^n \E\left[ (\widehat m(\exa{i}, \rx{i}{\mj}) - m(\rx{i}{\mj}))^2 \mid \rx{i}{\mj}, D_2 \right] \E\left[ \sxi_i^2 \mid \rx{i}{\mj}, D_2 \right] \overset{p}{\to} 0,
	\end{equation}
	which is the regression consistency condition~\eqref{eq:consitency-of-reg-func}(which was already verified to hold under the assumptions of Theorem \ref{thm:spline-type-I-error}).
	
Finally, in the context of spline regression, condition~\eqref{eq:variance-of-sxi-goes-to-zero-supp} is implied by
\[
\|\Pi(\widehat \beta_{\exa})\|_\infty^2 = o_p(K^{-1}_{\prx}).
\]
To see this, note that \( \V(\sxi^2 \mid \cx{\bullet}{\minus j}, D_2) \leq \E(\sxi^4 \mid \cx{\bullet}{\minus j}, D_2) \), which yields
\[
\E\left[ \V(\sxi^2 \mid \cx{\bullet}{\minus j}, D_2) \mid D_2 \right] \leq \E(\sxi^4 \mid D_2).
\]
Applying the bound from~\eqref{eq:upper-bound-xi-simplified}, we have
\[
\E(\sxi^4 \mid D_2) \leq \|\Pi(\widehat \beta_{\exa})\|_\infty^4,
\]
where \( \widehat \beta_{\exa} \) denotes the fitted spline coefficients and \( \Pi \) is the spline basis matrix introduced in Section~\ref{sec:examples-type-1-error}. 

Combining this with the lower bound on \( \sigma_n^2 \) from~\eqref{eq:lower-bound-sigma-partial}, we obtain
\[
\frac{1}{\sigma_n^2} \E\left[ \V(\sxi^2 \mid \cx{\bullet}{\minus j}, D_2) \mid D_2 \right] \leq c' K_{\prx} \|\Pi(\widehat \beta_{\exa})\|_\infty^2,
\]
for some constant \( c' > 0 \). Thus, condition~\eqref{eq:variance-of-sxi-goes-to-zero-supp} follows directly from the decay of the spline coefficient norm:
\[
\|\Pi(\widehat \beta_{\exa})\|_\infty^2 = o_p(K^{-1}_{\prx}).
\]

In summary, the spline model under the Model-X assumption satisfies all the required conditions of Theorem~\ref{thm:HRT-PCM-equivalence}, establishing the asymptotic equivalence of the HRT and tPCM tests in this setting.

\end{proof}

\section{Computational cost of methods compared in GWAS-inspired setting} \label{sec:computational-cost-appendix}

\begin{algorithm}[h]
	\SetAlgoLined 
	\KwIn{Data $\{(\rx{i}{\bullet}, \ey{i})\}_{i = 1, \dots, 2n}$}
	Split the data into $D_1 \cup D_2$, with $D_1$ and $D_2$ containing $n$ samples each. \textcolor{blue}{Cost $O(1)$.}
	
	Estimate $\E[\ey \mid \rx]$ on $D_2$ \textcolor{blue}{via lasso}, call it $\widehat m(\rx)$. \textcolor{blue}{Cost $O(np)$.}
	
	\For{$j \gets 1$ \KwTo $p$}{
		Regress $\widehat m(\rx)$ on $\rx{\mj}$ using $D_2$ \textcolor{blue}{via lasso} to obtain $\widecheck m_j(\rx{\mj})$ and define $\widehat f_j(\rx)  \equiv \widehat m(\rx) - \widecheck m_j(\rx{\mj})$. \textcolor{blue}{Evaluating $\widehat m(\rx)$ on $D_2$ is free since it was fitted in step 2. Cost of lasso is $O(np)$.}

		Using $D_1$, regress $\ey$ on $\rx{\mj}$ \textcolor{blue}{via lasso} to obtain an estimate $\widetilde{m}_j(\rx{\mj})$ of $\E[\cy|\rx{\mj}]$. \textcolor{blue}{Cost $O(np)$.}

		Also on $D_1$, regress $\widehat f_j(\rx)$ \textcolor{blue}{via lasso} on $\rx{\mj}$ to obtain $\widehat m_{\widehat f_j} (\rx{\mj})$. \textcolor{blue}{Forming $\widehat f_j(\rx)$ on $D_1$ costs $O(ns)$. Running the lasso costs $O(np)$.}

		Compute $T^{\textnormal{vPCM}}_j$ based on equations~\eqref{eq:l-i-j} and~\eqref{eq:vpcm-test-stat}. \textcolor{blue}{All components of $T^{\textnormal{vPCM}}_j$ were computed previously. Cost $O(n)$.}

		Set $p_j \equiv 1 - \Phi(T^{\textnormal{vPCM}}_j)$. \textcolor{blue}{Cost $O(1)$.}
	}
	\Return $\{p_j\}_{j = 1, \dots, p}$.
	
	\caption{Vanilla PCM computational cost}
	\label{alg: pcm cost}
\end{algorithm}

\begin{algorithm}[h]
	\SetAlgoLined 
	\KwIn{Data $\{(\rx{i}{\bullet}, \ey{i})\}_{i = 1, \dots, 2n}$, $B_{\textnormal{HRT}}$ resamples.}
	Split the data into $D_1 \cup D_2$, with $D_1$ and $D_2$ containing $n$ samples each. \textcolor{blue}{Cost $O(1)$.}
	
	Estimate $\E[\ey \mid \rx]$ on $D_2$ \textcolor{blue}{via lasso}, call it $\widehat m(\rx)$. \textcolor{blue}{Cost $O(np)$.}

	Estimate $\law(\prx)$ on $D_2$ \textcolor{blue}{via Baum-Welch}, call it $\lawhat(\rx)$. \textcolor{blue}{Cost $O(np)$.}

	\textcolor{blue}{Compute all conditionals $\lawhat(\ex{i}{j} \mid \rx{i}{\mj})$ for all $i \in D_1$ and $j = 1, \dots, p$ via Perduca-Noel algorithm. Cost $O(np)$.}

	Compute test statistic $T^{\textnormal{HRT}}$ as in equation~\eqref{eq:HRT-test}. \textcolor{blue}{All components of $T^{\textnormal{HRT}}$ were computed previously. Cost $O(n)$.}
	
	\For{$j \gets 1$ \KwTo $p$}{
		\For{$b \gets 1$ \KwTo $B_{\textnormal{HRT}}$}{
			Sample $\exk{i}{j} \sim \lawhat(\ex{j} \mid \rx{i}{\mj})$ for all $i \in D_1$. \textcolor{blue}{Given step 4, this costs $O(n)$.}

			Compute $\widetilde T_j^b \equiv \frac{1}{n} \sum_{i=1}^n (\ey{i} - \widehat{m}(\exk{i}{j}, \rx{i}{\mj}))^2$. \textcolor{blue}{Evaluating the fitted lasso model $\widehat m$ on $n$ new observations costs $O(ns)$.}
		}

		Set $p_j \equiv \frac{1}{B_{\textnormal{HRT}} + 1} \left(1 + \sum_{b = 1}^{B_{\textnormal{HRT}}} \mathbbm{1}\left[ T^{\textnormal{HRT}} \leq \widetilde T_j^b \right] \right)$. \textcolor{blue}{Cost $O(B_{\textnormal{HRT}})$.}
	}
	\Return $\{p_j\}_{j = 1, \dots, p}$.
	\caption{Holdout Randomization Test computational cost}
	\label{alg: hrt cost}
\end{algorithm}

\begin{algorithm}[h]
	\SetAlgoLined 
	\KwIn{Data $\{(\srx_{i}, \ey{i})\}_{i = 1, \dots, 2n}$, $B_{\textnormal{tPCM}}$ resamples.}
	Split the data into $D_1 \cup D_2$, with $D_1$ and $D_2$ containing $n$ samples each. \textcolor{blue}{Cost $O(1)$.}
	
	Estimate $\E[\ey \mid \rx]$ on $D_2$ \textcolor{blue}{via lasso}, call it $\widehat m(\rx)$. \textcolor{blue}{Cost $O(np)$.}

	Estimate $\law(\prx)$ on $D_2$ \textcolor{blue}{via Baum-Welch}, call it $\lawhat(\rx)$. \textcolor{blue}{Cost $O(np)$.}

	\textcolor{blue}{Compute all conditionals $\lawhat(\ex{i}{j} \mid \rx{i}{\mj})$ for all $i \in D_1$ and $j = 1, \dots, p$ via Perduca-Noel algorithm. Cost $O(np)$.}
	
	\For{$j \gets 1$ \KwTo $p$}{
		Compute $\widehat{m}_j (\rx{i}{\mj}) \equiv \sum_{x_j \in \{0,1\}} \widehat{m} (\exk{i}{j} = x_j, \rx{i}{\mj}) \lawhat(\ex{i}{j} = \exs{j} \mid \rx{i}{\mj})$ for all $i \in D_1$. \textcolor{blue}{Since $\lawhat(\ex{i}{j} = \exs{j} \mid \rx{i}{\mj})$ were computed in step 4, the cost of this step is two evaluations of the fitted lasso model $\widehat m$ on $n$ new observations, costing $O(ns)$.}

		Define $R_{ij} \equiv (\ey{i}  - \widehat{m}_j (\rx{i}{\mj}))(\widehat m(\rx{i}{\bullet}) - \widehat{m}_j (\rx{i}{\mj}))$ for $i$ in $D_1$. \textcolor{blue}{The cost of this step is dominated by evaluating $\widehat m(\rx{i}{\bullet})$ for each $i \in D_1$, which costs $O(ns)$.}

		Compute $T^{\textnormal{tPCM}}_j \equiv \frac{\frac{1}{\sqrt n} \sum_{i=1}^n R_{ij}}{\sqrt{ \frac{1}{n}\sum_{i=1}^n R_{ij}^2 - \left(\frac{1}{n} \sum_{i=1}^n R_{ij}\right)^2}}$. \textcolor{blue}{Cost $O(n)$.}

		Set $p_j \equiv 1 - \Phi(T^{\textnormal{tPCM}}_j)$. \textcolor{blue}{Cost $O(1)$.}
	}
	\Return $\{p_j\}_{j = 1, \dots, p}$.
	\caption{Tower PCM computational cost}
	\label{alg: tpcm cost}
\end{algorithm}

\begin{algorithm}[h]
	\SetAlgoLined 
	\KwIn{Data $\{(\rx{i}{\bullet}, \ey{i})\}_{i = 1, \dots, 2n}$}
	
	Fit an HMM to $\law(\rx)$ on all observations via Baum-Welch, call it $\lawhat(\rx)$. \textcolor{blue}{Cost $O(np)$.}

	Based on $\lawhat(\rx)$, generate HMM knockoffs $\rxk$ for all observations via Algorithm 2 in \citet{SetS19}. \textcolor{blue}{Cost $O(np)$.}

	Run a lasso of $\cy$ on $[\mx, \mxk]$ using all observations. \textcolor{blue}{Cost $O(np)$.}

	Form knockoff statistics based on the coefficients fitted in Step 3 and apply the knockoff filter \citep{BC15} to select variables. \textcolor{blue}{Cost $O(p)$.}

	\Return selected variables.
	\caption{Model-X knockoffs computational cost}
	\label{alg: knockoffs cost}
\end{algorithm}

To justify the computational costs reported in Table~\ref{tab:gwas-example-computations}, we annotate Algorithms~\ref{alg: pcm} (PCM), \ref{alg: hrt} (HRT), and \ref{alg: tpcm} (tPCM) with specific computational choices and their associated costs in Algorithms~\ref{alg: pcm cost}, \ref{alg: hrt cost}, and \ref{alg: tpcm cost}, respectively. We do the same for model-X knockoffs (Algorithm~\ref{alg: knockoffs cost}), and omit GCM because its algorithm is similar to that of PCM and its computational cost is the same. To justify the costs reported in Algorithms~\ref{alg: pcm cost}, \ref{alg: hrt cost}, \ref{alg: tpcm cost}, and \ref{alg: knockoffs cost}, note that:
\begin{itemize}
\item Processing each variable in one iteration of the coordinate descent algorithm \citep{friedman2010regularization} requires a sum over $n$ observations, so processing all variables in one iteration requires $O(np)$ operations. Since we assume a constant number of iterations, the total cost is also $O(np)$.
\item The Baum-Welch algorithm with forward-backward message passing requires $O(np)$ operations per iteration \citep{Rabiner1989}. Since we assume a constant number of iterations, the total cost is also $O(np)$.
\item The Perduca-Noel algorithm for computing all conditionals $\lawhat(\ex{j} \mid \rx{\mj})$ requires $O(np)$ operations \citep{Perduca2013}.
\item Evaluating a fitted lasso model with $s$ nonzero coefficients on $n$ new observations requires the multiplication of an $n \times s$ matrix by an $s \times 1$ vector, which costs $O(ns)$ operations.
\item Constructing HMM knockoffs requires $O(np)$ operations \citep{SetS19}.
\end{itemize}
Putting together the costs reported in Algorithms~\ref{alg: pcm cost}, \ref{alg: hrt cost}, \ref{alg: tpcm cost}, and \ref{alg: knockoffs cost}, we obtain the computational costs reported in Table~\ref{tab:gwas-example-computations}.

\section{Methods compared in the simulation study} 

\subsection{Method implementation details} \label{sec:simulation-study-details}

\paragraph{tPCM} We apply tPCM (Algorithm~\ref{alg: tpcm}) with the \verb|ranger()| function from \texttt{ranger} package to fit a random forest for $\E[\ey \mid \rx]$. We used the \texttt{fastPhase} software for estimating the initial, transition, and emission probabilities of an HMM. These estimates were then converted into estimates of $\P(\ex{j} \mid \rx{\mj})$ using a forward-backwards algorithm; see Appendix B of \citet{Niu2025} for more details. We choose training proportion 0.45, determined as described in Appendix~\ref{sec:choosing-training-proportions}. 
	
	\paragraph{HRT} We apply the HRT (Algorithm~\ref{alg: hrt}) with the \verb|ranger()| function from \texttt{ranger} package to fit a random forest for $\E[\ey \mid \rx]$. We used the \texttt{fastPhase} software for estimating the initial, transition, and emission probabilities of an HMM. Like for tPCM, these estimates were then converted into estimates of $\P(\ex{j} \mid \rx{\mj})$ using a forward-backwards algorithm. We choose $B_{\text{HRT}} = 200$ resamples and training proportion 0.45. 

	\paragraph{PCM} We apply a variant of PCM that is closer to Algorithm 1 from \citet{Lundborg2022a} than vanilla PCM (Algorithm~\ref{alg: pcm} in Section~\ref{sec:vPCM}), as it includes Step 1 (iv). Step 1 (ii) was not possible in this case since we fit an interacted model. We continued to omit Step 2 of Algorithm 1 from \citet{Lundborg2022a}, which the authors claimed ``is not critical for good power properties." We also use \verb|ranger()| for fitting $\mathbb{E}[\ey \mid \rx{\mj}]$. Moreover, to maintain as fair a comparison as possible, we endow PCM with knowledge of the categorical nature of $\rx$. This meant that for any step where a function of $\rx$ is regressed on $\rx{\mj}$ (Steps 1 (iii) and 3 (i) from Algorithm 1 from \citet{Lundborg2022a}), we first estimated the conditional probability, $\widehat{\P}(\ex{j} \mid \rx{\mj})$, which was used to compute the conditional expectation estimate $\widehat{\E}[f_j(\rx) \mid \rx{\mj}] \equiv \sum_{\exs{j} \in \{0,1\}} f_j(\exs{j}, \rx{\mj}) \widehat{\P}(\ex{j} =\exs{j}\mid \rx{\mj})$. These estimates were obtained using the classification version of \verb|ranger()|. We choose training proportion 0.45, determined as described in Appendix~\ref{sec:choosing-training-proportions}.

	\paragraph{knockoffs} We implemented knockoffs using the \verb|knockoff()| from the \texttt{knockoff} package. The choice of test statistic was \verb|stat.random_forest()|. Knockoffs were sampled using the estimated initial probability, transition, and emission matrix estimates from \texttt{fastPhase}, which were then passed into the \verb|knockoffHMM()| function from the \texttt{SNPknock} package.

	\paragraph{Oracle GCM} We also compare to an oracle version of the GCM test that is equipped with the true $\law(\ey \mid \rx)$ and $\law(\rx)$, and using the same tower property trick as the tPCM test. Since there was no nuisance function estimation, there was no sample splitting, and so the Oracle GCM test had a larger sample size than the other splitting methods. 

\subsection{Methods excluded from comparison} \label{sec:omission-justification}

Here we justify the omission of two additional methods from our simulation study. First, we omit the holdout grid test (HGT), a faster version of HRT proposed by \citet{Tansey2018}. The HGT employs a discrete, finite grid approximation and a caching strategy. \citet{Tansey2018} theoretically demonstrated the validity of the procedure under the model-X assumption. We chose to omit this method because when the joint distribution of $\prx$ is not known but estimated, the method may no longer be valid and/or depends on the level of discretization chosen. Furthermore, this method trades off computational resources for memory resources, complicating the comparison. Second, we omit a method proposed by \citet{Williamson2021a} for testing whether the functional~\eqref{eq:mmse-gap} equals zero because simulations in \citet{Lundborg2022a} demonstrated a sizable gap in power when compared with PCM. 

\section{False discovery rate in the simulation study} \label{sec:fdr}
\begin{figure}[H]
	\centering
	\includegraphics{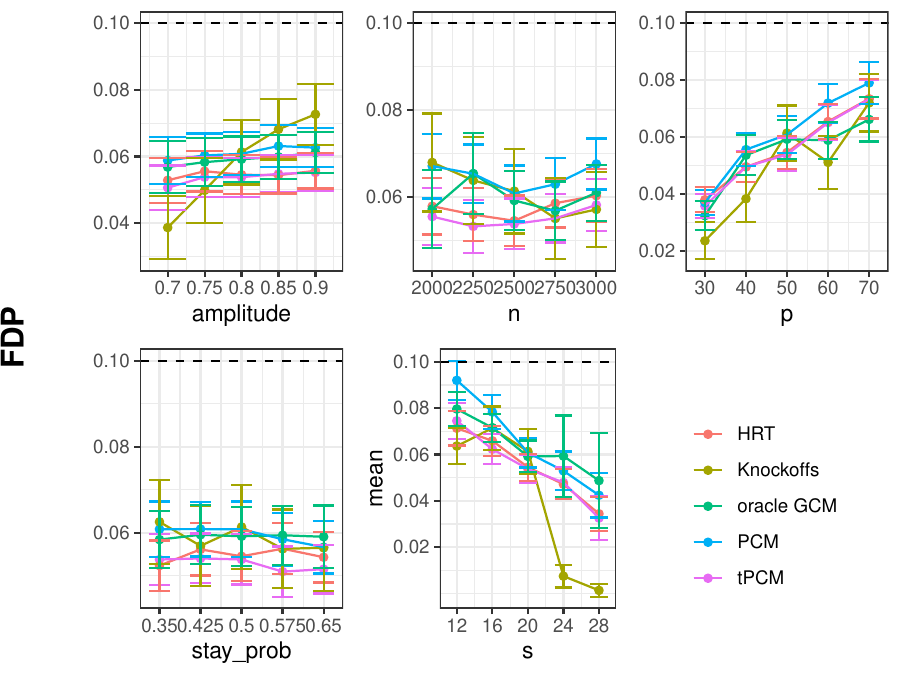}
	\caption{Type-I error control: in each plot, we vary one parameter. Each point is the average of 400 Monte Carlo replicates, and the error bars are the average $\pm 2 \times \widehat{\sigma}_f$, where $\widehat{\sigma}_f$ is the Monte Carlo standard deviation divided by $\sqrt{400}$.}
	\label{fig:fdr-rf}
\end{figure}

\section{Choosing the training proportions} \label{sec:choosing-training-proportions}
In this section, we justify our choice of the best training proportions for tower PCM and PCM. For tPCM, we compared training proportions in $\{0.3, 0.35, 0.4, 0.45, 0.5\}$. For PCM we compared training proportions in $\{0.4, 0.45, 0.5\}$. We plot the false discovery proportions and power for for each method in Figures \ref{fig:tPCM FDR proportion}, \ref{fig:PCM FDR proportion}, \ref{fig:tPCM power proportion hmm}, and \ref{fig:PCM power proportion hmm}. Generally, all choices of proportions seem to be controlling the type-I error for both tPCM and PCM.  It is unclear what we should expect, since smaller training proportion means more data for the in-sample fits on the test split, but a poorer estimate of the direction of the alternative on the training split. In terms of power, though there is not a single training proportion that dominates uniformly for both tPCM and PCM, 0.45 are generally the highest for both. We do note that the simulation setting for choosing the proportion was similar to but does not exactly match the simulation setting from the main text. Nevertheless, we utilized the 0.45 proportion for the simulation in the main text.

\begin{figure}
	\centering
	\includegraphics{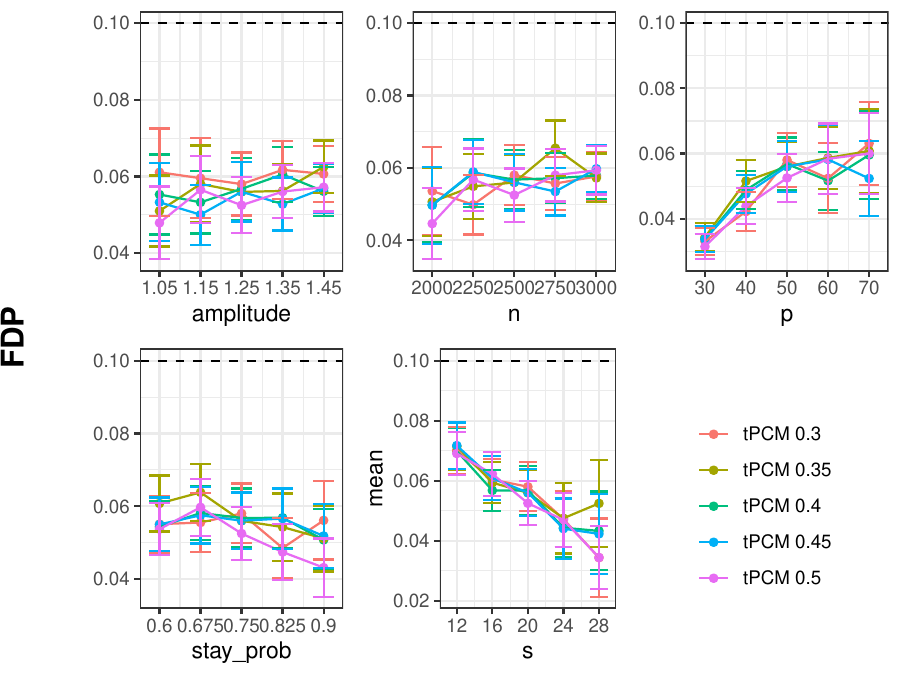}
	\caption{A false discovery rate comparison between different training proportions for tPCM.}
	\label{fig:tPCM FDR proportion}
\end{figure}

\begin{figure}
	\centering
	\includegraphics{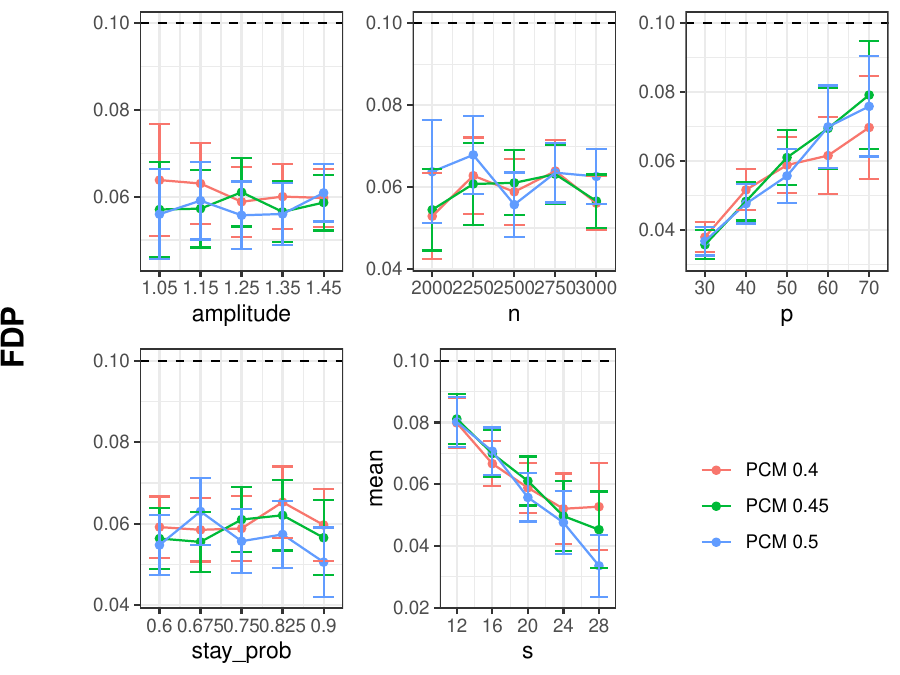}
	\caption{A false discovery rate comparison between different training proportions for PCM.}
	\label{fig:PCM FDR proportion}
\end{figure}

\begin{figure}
	\centering
	\includegraphics{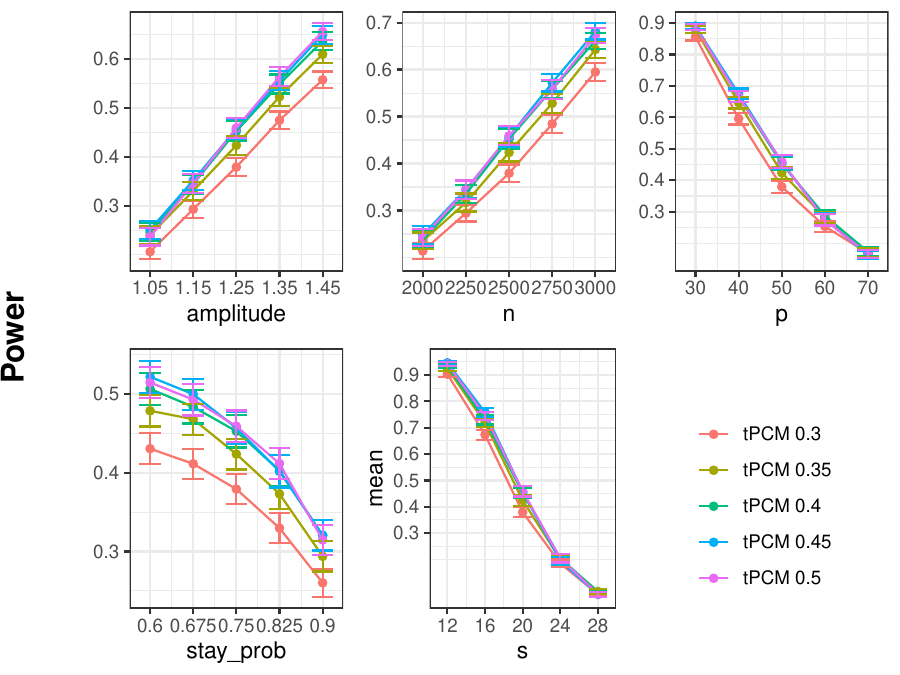}
	\caption{A power comparison between different training proportions for tPCM.}
	\label{fig:tPCM power proportion hmm}
\end{figure}

\begin{figure}
	\centering
	\includegraphics{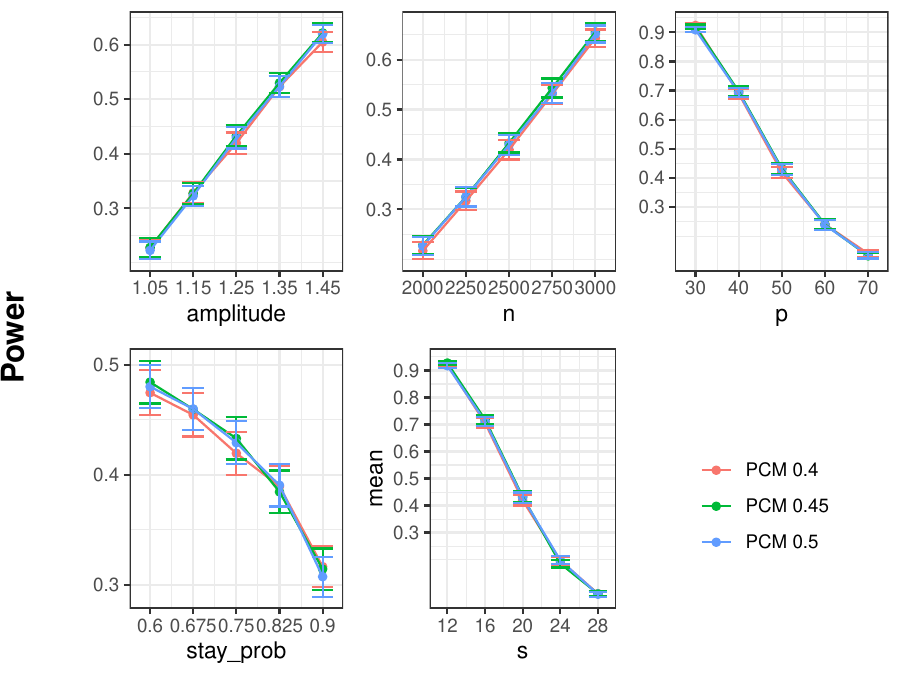}
	\caption{A power comparison between different training proportions for PCM.}
	\label{fig:PCM power proportion hmm}
\end{figure}

\section{Method implementation details in the data analysis} \label{sec:data-analysis-details}
\paragraph{HRT and tPCM}
HRT and tPCM utilized a 0.3 training proportion. On the training sample, we obtained fits for $\E[\ey \mid \rx]$ and $\law(\rx)$. Since $\ey$ is binary, for $\E[\ey \mid \rx]$, we used a classification random forest by converting the outcome to a factor and applying the \verb|ranger()| function from the \texttt{ranger} package. For $\mathcal{L}(\rx)$, we used the same fit as in \citet{Liu2020} and \citet{Li2021c}, which was the graphical lasso as implemented in the \texttt{CVglasso()} function from the \texttt{CVglasso} package with parameter \texttt{lam.min.ratio = 1e-6}. HRT utilized $B_{\text{HRT}} = 3300 \approx 2 \times p / \alpha$ resamples, and tPCM used $B_{\text{tPCM}} = 25$ resamples to approximate conditional means. 

\paragraph{PCM}
As in the simulation study, PCM was implemented as described in Algorithm 1 of \citet{Lundborg2022a}, except for Step 2. It did not included the extra step (1(ii)) that can be performed when the contribution to $\E[\ey \mid \rx]$ from $\ex{j}$ can be separated from the contributions from the other predictors, as we fit an interacted model. PCM also used a 0.3 sample split, and also used the a classification random forest as in the previous methods for fitting $\E[\ey \mid \rx]$ on the training split, as well as for fitting $\E[\ey \mid \rx{\mj}]$. We chose to fit $\E[f_j(\ex{j}) \mid \rx{\mj}]$ using the Lasso as implemented in the \texttt{glmnet} package on the evaluation split. We felt this choice was a reasonable analog to the graphical lasso fit used for tPCM and HRT.

\paragraph{knockoffs}
As in the simulation study, we used the \verb|stat.random_forest()| statistic from the \texttt{knockoff} package, with the outcome converted to a factor so that a classification forest was fit. We also sampled multivariate gaussian knockoffs using the graphical lasso with the same hyperparameters as used for HRT and tPCM.

\paragraph{tGCM}
tGCM is akin to the oracle GCM from the simulation, except $\E[\ey \mid \rx]$ and $\law(\rx)$ are estimated from the data. tGCM uses the same tower-based acceleration as the tPCM test. There is no danger of a degenerate limiting distribution under the null, so we can make use of the full sample for testing through 5 fold cross-fitting. For each of the five equally sized folds, $\E[\ey \mid \rx]$ and $\law(\rx)$ are estimated on the remaining 4/5 of the data using the same estimators as for HRT and tPCM. The tower trick is utilized to estimate $\E[\ey \mid \rx{\mj}]$ from the estimates for $\E[\ey \mid \rx]$ and $\law(\rx)$ using 25 resamples.

\section{Additional simulations results under an alternative DGP}
\label{sec:gam-banded-sim-results}
In this section, we investigate the finite-sample performance of tPCM with a simulation-based assessment of Type-I error, power, and computation time under a different data-generating process. The Type-I error of choice was the family-wise error rate at level $\alpha = 0.05$. We consider a generalized additive model (GAM) specification for the distribution of $Y \mid X$. The goal of the simulation is to corroborate the findings of the previous sections: (1) tPCM is computationally efficient, (2) tPCM controls the Type-I error, and (3) tPCM is as powerful as HRT and PCM. 
	\subsection{Data-generating model}

	We pick $s$ of the $p$ variables to be nonnull at random. Let $\mathcal{S}$ denote the set of nonnulls. We then define our data-generating model as follows:
	\begin{equation*}
	\law_n(\rx) = N(0,\Sigma(\rho)), \ \law_n(\ey \mid \rx) = N\left(\sum_{i \in \mathcal{S}, \textnormal{odd}} (X_i - 0.3)^2 / \sqrt{2} \theta + \sum_{i \in \mathcal{S}, \textnormal{even}} - \cos(X_i)\theta,1 \right)
	\end{equation*}
	Here, the covariance matrix $\Sigma(\rho)$ is an AR(1) matrix with parameter $\rho$; that is, $\Sigma(\rho)_{ij} = \rho^{|i-j|}$. Therefore, the entire data-generating process is parameterized by the five parameters $(n, p, s, \rho, \theta)$; see Table \ref{tab:gam sim parameters}. Since we utilized sample split proportions other than 0.5, in this section, we let $n$ denote the \emph{total} sample size, i.e. the combined size of $D_1$ and $D_2$. We vary each of the five parameters across five values each, setting the remaining to the default values (in bold).

	\begin{table}[h!]
		\centering
		\begin{tabular}{ccccc}
			\hline
			$n$ & $p$ & $s$ & $\rho$ & $\theta$ \\
			\hline
			800 & 30 & 4 & 0.2 & 0.15 \\
			1000 & 40 & 8 & 0.35 & 0.2 \\
			\textbf{1200} & \textbf{50} & \textbf{12} & \textbf{0.5} & \textbf{0.25} \\
			1400 & 60 & 16 & 0.65 & 0.3 \\
			1600 & 70 & 20 & 0.8 & 0.35 \\
			\hline
		\end{tabular}
		\caption{The values of the sample size $n$, covariate dimension $p$, sparsity $s$, autocorrelation of covariates $\rho$, and signal strength $\theta$ used for the
		simulation study. Each of the parameters $n$, $p$, $s$, $\rho$, $\theta$ was varied among the values displayed in the
		table while keeping the other four at their default values, indicated in bold. For example, $p = 50$, $s = 12$, $\rho = 0.5$, $\theta = 0.25$ were kept fixed while varying $n \in \{800, 1000, 1200, 1400, 1600\}$.}
		\label{tab:gam sim parameters}
	\end{table}

	\subsection{Methodologies compared}
	We applied the four methods tPCM, HRT, PCM, and oracle GCM in conjunction with a Bonferroni correction at level $\alpha = 0.05$ to control the family-wise error rate. For all methods, quantities such as $\E[\ey \mid \rx]$ and $\E[f_j(\ex{j}) \mid \rx{\mj}]$ were fit using (sparse) GAMs. tPCM and HRT exploited knowledge of the banded structure and so $\law(\rx)$ was fit using a banded precision matrix estimate. PCM was also endowed with knowledge of the banded covariance structure. This meant that for any step requiring a $\E[f_j(\ex{j}) \mid \rx{\mj}]$ fit, we actually only regressed $f_j(\ex{j})$ on $\ex{j-1}$ and $\ex{j+1}$, since $\ex{j}$ is independent of all other $\rx_k$ given $\ex{j-1}$ and $\ex{j+1}$. Oracle GCM was given knowledge of the true $\E[\ey \mid \rx]$ and $\law(\rx)$ models. More specifics are given below:
	
	\paragraph{tPCM} We apply tPCM (Algorithm~\ref{alg: tpcm}) with the  \verb|bam()| function from \texttt{mgcv} package for GAM fitting for $\E[\ey \mid \rx]$ with penalization parameter \verb|bs = "cs"|, and the banded precision matrix estimation from the \texttt{CovTools} package for $\law(\rx)$. We choose $B_{\text{tPCM}} = 25$ resamples and training proportion 0.4, the latter determined as described in Appendix~\ref{sec:choosing-training-proportions-gam}. 
	
	\paragraph{HRT} We apply the HRT (Algorithm~\ref{alg: hrt}) with the \verb|bam()| function from \texttt{mgcv} package for GAM fitting for $\E[\ey \mid \rx]$ and the banded precision matrix estimation from the \texttt{CovTools} package for $\law(\rx)$. We choose $B_{\text{HRT}} = 5000$ resamples and training proportion 0.4. Because HRT was the slowest of the methods considered, we only applied it to the default simulation setting for the sake of computational feasibility.

	\paragraph{PCM} We apply a variant of PCM that is closer to Algorithm 1 from \citet{Lundborg2022a} than vanilla PCM (Algorithm~\ref{alg: pcm} in Section~\ref{sec:vPCM}), as it includes Step 1 (ii) and Step 1 (iv). Step 1 (ii) was possible in this case since we fit a GAM. We continued to omit Step 2 of Algorithm 1 from \citet{Lundborg2022a}, which the authors claimed ``is not critical for good power properties." We also use \verb|bam()| for fitting $\mathbb{E}[\ey \mid \rx{\mj}]$. Moreover, to maintain a fair comparison, we endow PCM with knowledge of the banded covariance structure for the predictors. This meant that for any step where a function of $\ex{j}$ is regressed on $\rx{\mj}$ (Steps 1 (iii) and 3 (i) from Algorithm 1 from \citet{Lundborg2022a}), we actually only regressed $\ex{j}$ on $\ex{j-1}$ and $\ex{j+1}$, since $\ex{j}$ is independent of all other $\rx_k$ given $\ex{j-1}$ and $\ex{j+1}$ under the banded structure. These regressions were also performed using \verb|bam()|. We choose training proportion 0.3, determined as described in Appendix~\ref{sec:choosing-training-proportions-gam}.

	\paragraph{Oracle GCM} We also compare to an oracle version of the GCM test that is equipped with the true $\law(\ey \mid \rx)$ and $\law(\rx)$, as well as the same tower property-based acceleration as the tPCM test (also based on 25 resamples). Since there was no nuisance function estimation, there was no sample splitting, and so the Oracle GCM test had a larger sample size than the other methods. 

	\begin{remark}
	We omitted the two methods discussed in Section \ref{sec:omission-justification} for the same reasons. We also omitted model-X knockoffs since we desired family-wise error rate control, and model-X knockoffs is not designed to produce fine-grained $p$-values necessary to control this error rate.
	\end{remark}

	\subsection{Simulation results}

	Results for family-wise error, power, and computation time are presented in Figures \ref{fig:fwer-gam}, \ref{fig:power-gam}, and \ref{fig:computation-gam} respectively. Below are our observations from these results:

	\begin{itemize}
		\item As we expect, all methods tend in improve in terms of power as $n$ increases, amplitude increases, $p$ decreases, and $\rho$ decreases. For $s$, there is no such monotonic relationship.
		\item All methods control the family-wise error rate, indicating that in this setting, the $\law(\ey\mid \rx)$ and $\law(\rx)$ are learned sufficiently well.
		\item The oracle GCM has significantly lower power than the other methods, as the test statistic it is based on is most powerful against partially linear alternatives, which is not the case in the simulation design. The other methods have roughly equal power.
		\item Among the three powerful methods, tPCM is by far the fastest, with the gap widening as $p$ grows. 
 	\end{itemize}
	\begin{figure}[H]
	\centering
	\includegraphics{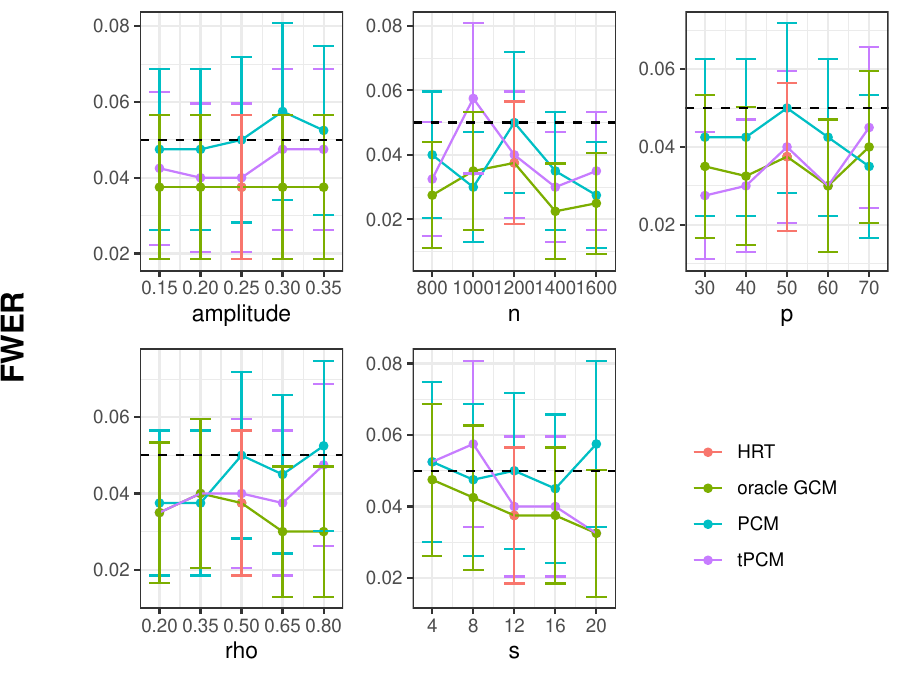}
	\caption{Type-I error control: in each plot, we vary one parameter. Each point is the average of 400 Monte Carlo replicates, and the error bars are the average $\pm 2 \times \widehat{\sigma}_f$, where $\widehat{\sigma}_f$ is the Monte Carlo standard deviation divided by $\sqrt{400}$.}
	\label{fig:fwer-gam}
\end{figure}
	\begin{figure}[H]
		\centering
		\includegraphics{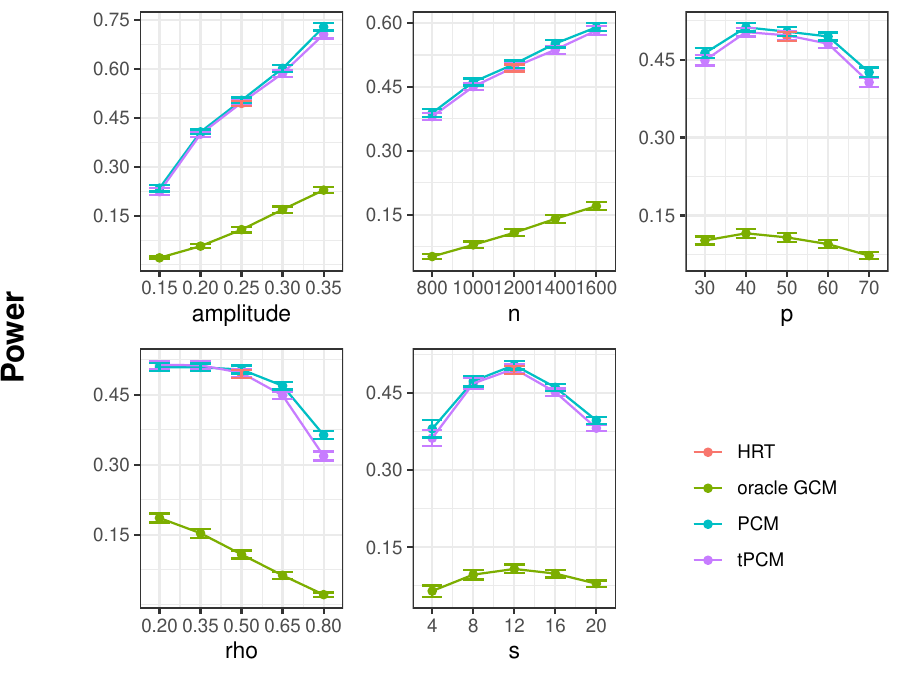}
		\caption{Power: in each plot, we vary one parameter. Each point is the average of 400 Monte Carlo replicates, and the error bars are the average $\pm 2 \times \widehat{\sigma}_p$, where $\widehat{\sigma}_p$ is the Monte Carlo standard deviation divided by $\sqrt{400}$.}
		\label{fig:power-gam}
	\end{figure}
	\begin{figure}[H]
		\centering
		\includegraphics{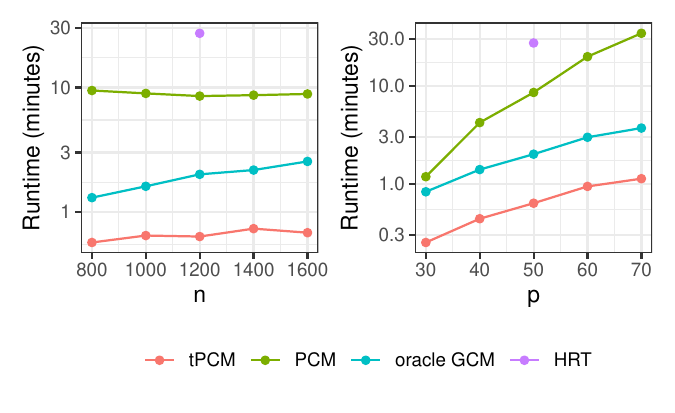}
		\caption{Computation: in each plot, we vary one parameter. Each point is the average of 400 Monte Carlo replicates.}
		\label{fig:computation-gam}
	\end{figure}

	\subsection{Computational comparison in a larger setting}
	In the main simulation setting, we chose smaller $n$ and $p$ so that it would be computationally feasible to run 400 Monte Carlo replicates of all methods to assess statistical performance. To further demonstrate the computational advantage of tPCM, we considered a larger setting with the same data-generating model as before, but with different parameters. Specifically, we fixed $n = 2500$, $p = 100$, $\rho = 0.5$, $\theta = 0.25$, $s = 15$, and varied $p \in \{100, 125, 150, 175, 200\}$.  We forego any statistical comparison and simply measure the time taken to perform each procedure once for each of the five settings of $p$. HRT, PCM, and tPCM all used a 0.4 training proportion, HRT used $5 \times p / 0.05$ resamples, and tPCM and oracle GCM used 25 resamples. These results are shown in the right panel of Figure \ref{fig:computation-rf}. As expected, the computational gap between tPCM and HRT and PCM widens as $p$ increases, and when $p = 200$, tPCM is more than 130 times faster than HRT and PCM.

	\subsection{Choosing the training proportions} \label{sec:choosing-training-proportions-gam}
In this section, we justify our choice of the best training proportions for tower PCM and PCM. For tPCM, we compared training proportions in $\{0.3, 0.4, 0.5, 0.6, 0.7\}$. For PCM we compared training proportions in $\{0.3, 0.4, 0.5\}$. We plot the family-wise error rates and power for for each method in Figures \ref{fig:tPCM FWER proportion}, \ref{fig:PCM FWER proportion}, \ref{fig:tPCM power proportion}, and \ref{fig:PCM power proportion}. In terms of type-I error for tPCM, 0.7 seems the most conservative which is perhaps not surprising, as it uses more data for the nuisances and less for testing. The rest of the proportions do not follow a monotonic trend, however. Generally, all proportions seem to be controlling the type-I error, though 0.5 and 0.6 exhibit some slight inflation for some settings. The type-I error rate for PCM is also not monotone. It is unclear what we should expect, since smaller training proportion means more data for the in-sample fits on the test split, but a poorer estimate of the direction of the alternative on the training split. In terms of power, though there is not a single training proportion that dominates uniformly for both tPCM and PCM, 0.4 and 0.3 are generally the highest, respectively. 

\begin{figure}
	\centering
	\includegraphics{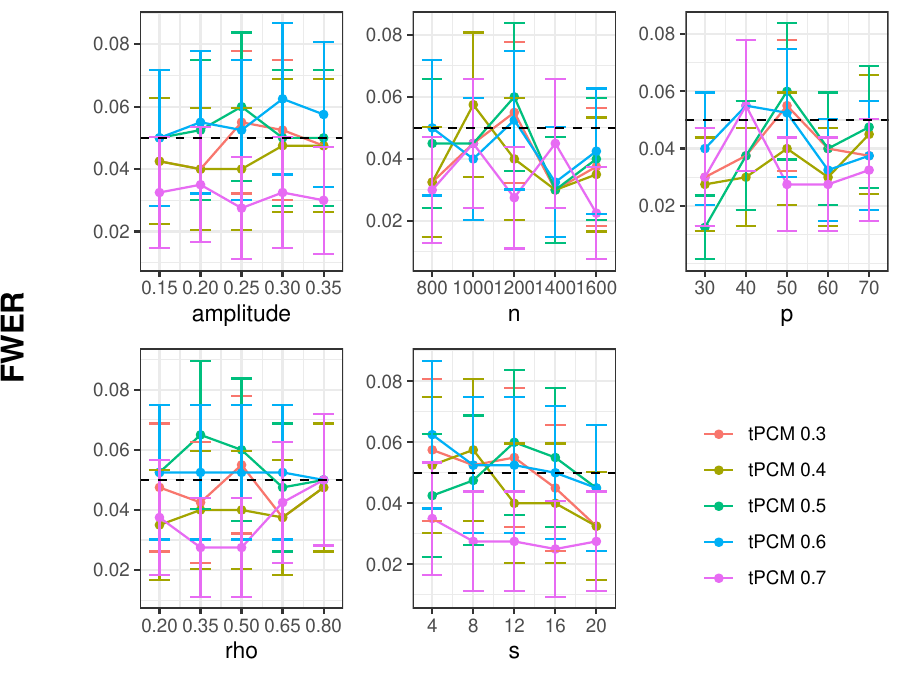}
	\caption{A family-wise error rate comparison of between  different training proportions for tPCM.}
	\label{fig:tPCM FWER proportion}
\end{figure}

\begin{figure}
	\centering
	\includegraphics{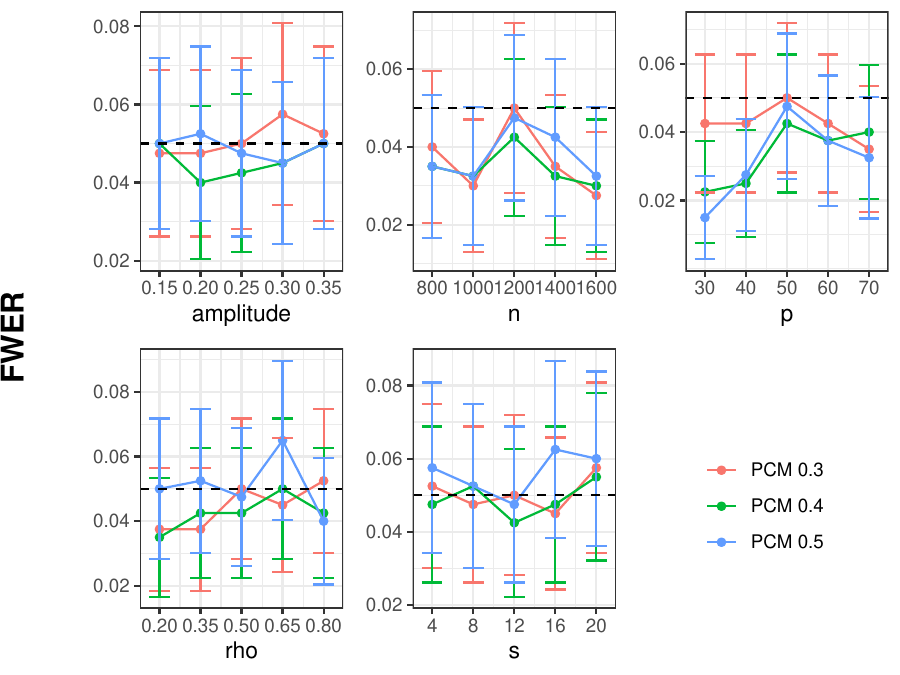}
	\caption{A family-wise error rate comparison between different training proportions for PCM.}
	\label{fig:PCM FWER proportion}
\end{figure}

\begin{figure}
	\centering
	\includegraphics{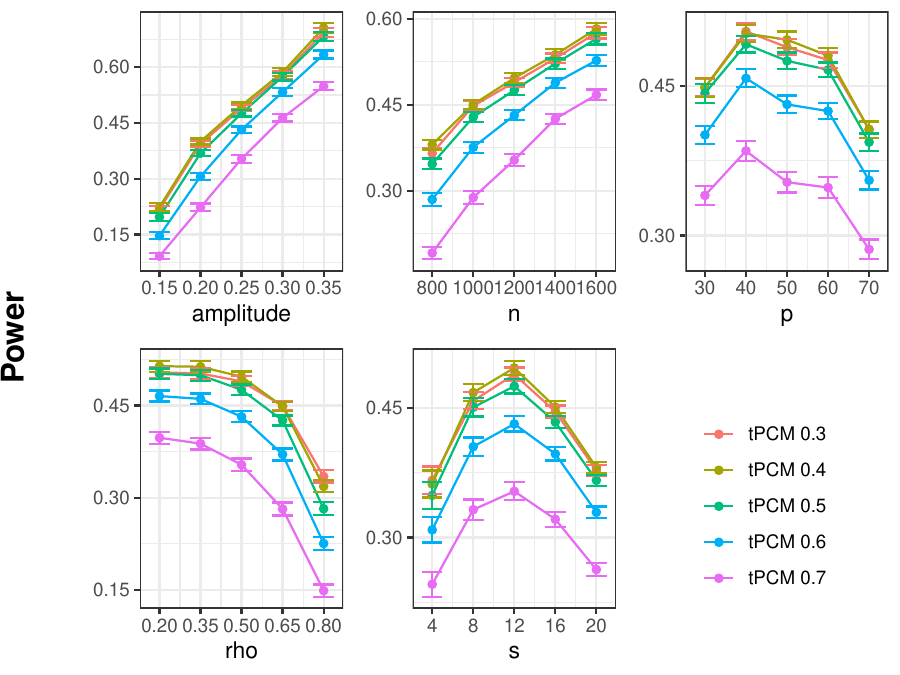}
	\caption{A power comparison between different training proportions for tPCM.}
	\label{fig:tPCM power proportion}
\end{figure}

\begin{figure}
	\centering
	\includegraphics{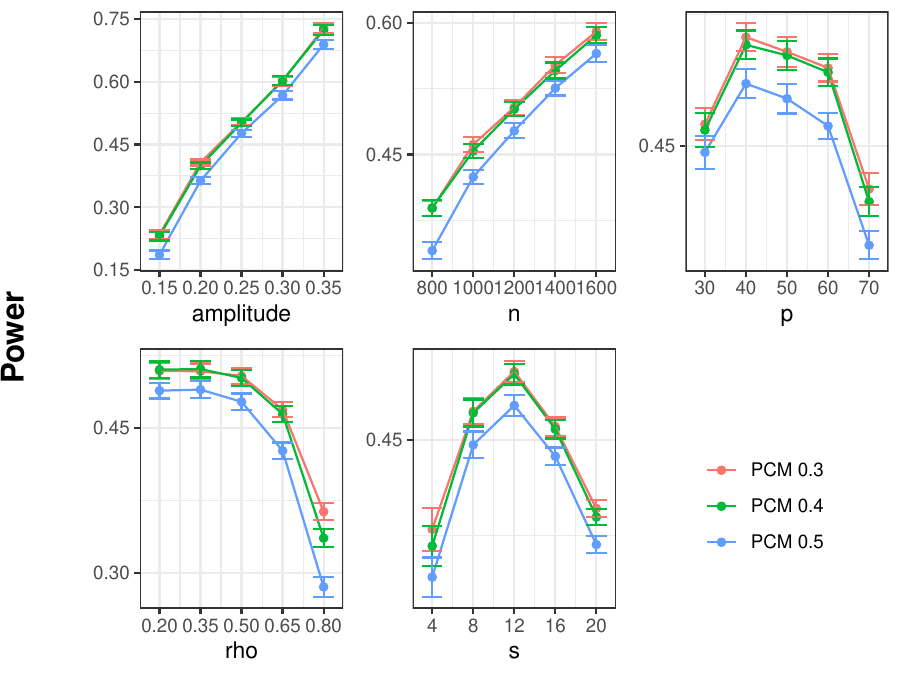}
	\caption{A power comparison between different training proportions for PCM.}
	\label{fig:PCM power proportion}
\end{figure}
\end{document}